\documentclass[a4paper,11pt,leqno]{amsart}
\usepackage{color,array,a4wide}
\usepackage{cite}
\usepackage{hyperref}
\usepackage{amsmath,amssymb,amsthm,color,ulem}
\hypersetup{linktocpage,colorlinks, citecolor={magenta}}
\usepackage[english]{babel}
\usepackage[utf8]{inputenc}

\newtheorem{manualtheoreminner}{Theorem}
\newenvironment{manualtheorem}[1]{%
  \renewcommand\themanualtheoreminner{#1}%
  \manualtheoreminner
}{\endmanualtheoreminner}

\usepackage{accents}

\usepackage{xspace}
\usepackage{bm}				
\usepackage{bbm,upgreek}		
\usepackage[e]{esvect}		
\usepackage{esint}			
\usepackage{etoolbox}			
\usepackage{calc}				
\usepackage{eso-pic}			
\usepackage{datetime}			

\usepackage{lmodern}
\numberwithin{equation}{section}

\usepackage{enumitem}
\setlist{nosep}
\setlist{noitemsep}

\newcommand{\R}{\mathbb{R}}

\newtheorem{theo}{Theorem}
\newtheorem{prop}{Proposition}[section]
\newtheorem{lem}[prop]{Lemma}
\newtheorem{coro}[prop]{Corollary}
\newtheorem{remark}[prop]{Remark}

\theoremstyle{plain}

\theoremstyle{definition}

\newtheorem{step}{Step}

\def \t0{\rightarrow 0} 

\def \be{\begin{equation}}
\def \ee{\end{equation}}

\def \hal{\frac{1}{2}}

\def \supp{\mathrm{supp }} 
\def \div{\mathrm{div} \,} 
\def \1{\mathbf{1}} 

\def \p{\partial}
\def \ep{\varepsilon}
\def \dist{\mathrm{dist}}

\def\nab{\nabla}

\def\({\left(}
\def\){\right)}

\def \PNbeta{\mathbb{P}_{N, \beta}} 

\def\Xint#1{\mathchoice
   {\XXint\displaystyle\textstyle{#1}}%
   {\XXint\textstyle\scriptstyle{#1}}%
   {\XXint\scriptstyle\scriptscriptstyle{#1}}%
   {\XXint\scriptscriptstyle\scriptscriptstyle{#1}}%
   \!\int}
\def\XXint#1#2#3{{\setbox0=\hbox{$#1{#2#3}{\int}$}
     \vcenter{\hbox{$#2#3$}}\kern-.5\wd0}}
\def\dashint{\Xint-}

\def \XN{X_N}
\def \YN{Y_N}

\def \Fluct{\mathrm{Fluct}}

\def\Esp{\mathbb{E}} 

\def \ZNbeta{Z_{N,\beta}^V}

\def \Ani{\mathsf{A}}

\def \id{\mathrm{Id}}

\def\K{\mathsf{K}}

\def \mut{\overline{\mu}_{t}}

\def\g{\mathsf{g}}
\def\d{\mathsf{d}}

\def\fd{f_{\mathsf{d}}}

\def \Main{\mathrm{Main}}

\def\indic{\mathbf{1}}

\def\rb{\rho_\beta}

\makeatletter
\def\namedlabel#1#2{\begingroup
    #2%
    \def\@currentlabel{#2}%
    \phantomsection\label{#1}\endgroup
}
\makeatother

\def \epsilon{\varepsilon}

\def \f{\mathsf{f}}
\def\F{\mathsf{F}}

\def \Error{\mathsf{Error}}

\def \veta{\vec{\eta}}

\def \HN{\mathcal{H}_N}
\def \rr{\mathsf{r}}

\def \cd{\mathsf{c}_{\d}} 

\def \PNbeta{\mathbb{P}_{N, \beta}} 

\def\mut{\mu_t}
\def\nut{\nu_\theta^t}

\def \f{\mathsf{f}}
\def \pa{\partial}
\def\D{B}

\def\car{\Box}

\def\muth{\mu_\theta}

\def\muv{\mu_\infty}

\def\mub{\mu_\theta}
\def\mut{\mu_{\theta}}
\def\mutt{\mu_{\theta}^t}

\def\m{m}

\def\mf{f_{\d}}

\def\F{\mathsf{F}}
\def\K{\mathsf{K}}

\def\Esp{\mathbb{E}}
\def\rr{\mathsf{r}}
\def\rrc{\tilde{\mathsf{r}}}

\def\mn{\mathrm{n}}

\def\Rem{\mathrm{Rem}}

\def\cor{\textcolor{black}}
\def\cord{\textcolor{black}}

\begin{document}
\title[Fluctuations for Coulomb Gases]{Gaussian Fluctuations and Free Energy Expansion  for  Coulomb Gases at Any Temperature}
\author[S. Serfaty]{Sylvia Serfaty}\address[S. Serfaty]{Courant Institute of Mathematical Sciences, New York University, 251 Mercer St., New York, NY 10012}
\email{serfaty@cims.nyu.edu}
\date{October 22, 2022}

\begin{abstract}
We obtain  concentration estimates for the fluctuations of Coulomb gases  in any dimension and in a broad temperature regime, including very small and very large temperature regimes which may depend on the number of points. 
We obtain a full Central Limit Theorem (CLT) for the fluctuations of linear statistics  in dimension 2, valid for the first time down to  microscales and for temperatures possibly tending to $0$ or $\infty$ as the number of points diverges.  We show that  a similar  CLT can  also be obtained in any larger dimension  conditional on a ``no phase-transition" assumption,  as soon as one can obtain a precise enough error rate for the expansion of the free energy  --  an expansion is obtained in any dimension, but the rate is \cord{so far not good enough to conclude}.
      These CLTs  can be interpreted as a convergence to the Gaussian Free Field.
 All the results are valid as soon as  the test-function lives on a larger scale  than the temperature-dependent minimal scale $\rho_\beta$ introduced in our previous work \cite{as}.
 
\end{abstract}

\maketitle

\noindent
{\bf keywords:} Coulomb gas, one-component plasma, concentration, Central Limit Theorem, Gaussian Free Field.

\section{Introduction}
\subsection{Setting of the problem}
In this paper we continue our investigation of $\d$-dimensional Coulomb gases (with $\d\ge 2$) at the \textit{inverse temperature} $\beta$, defined by the Gibbs measure
\begin{equation}\label{def:PNbeta}
d\PNbeta(\XN) = \frac{1}{\ZNbeta} \exp \left( - \beta N^{\frac{2}{\d}-1} \HN(\XN)\right) d\XN,
\end{equation}
where $\XN = (x_1, \dots, x_N)$ is an $N$-tuple of points in $\R^\d$ and $\HN(\XN)$ is the energy of the system in the state $\XN$, given by
\begin{equation} \label{def:HN}
\HN(\XN) := \hal \sum_{1 \leq i \neq  j \leq N} \g(x_i-x_j) +N \sum_{i=1}^N  V(x_i),
\end{equation}
where
\begin{align} 
\label{wlog2d} 
\g(x) : = 
\left\{ 
\begin{aligned}
& - \log |x| & \text{if} & \ \d=2, \\
& |x|^{2-\d} & \text{if} & \ \d \geq 3.
\end{aligned} 
\right.
\end{align}  We will denote in the whole paper by $\cd$ the (explicitly computable) constant such that $-\Delta \g=\cd \delta_0$ in dimension $\d$.
Thus the energy~$\HN(\XN)$ is the sum of the pairwise repulsive Coulomb interaction between all particles plus the effect on each particle of an external field or confining potential $NV$ whose intensity is proportional to $N$. The normalizing constant $\ZNbeta$ in the definition \eqref{def:PNbeta}, called the \textit{partition function}, is given by 
\begin{equation}\label{defZ}
\ZNbeta := \int_{(\R^\d)^N} \exp \left( - \beta N^{\frac{2}{\d}-1}  \HN(\XN)\right) d\XN.
\end{equation}
We have chosen particular units of measuring the inverse temperature by writing~$\beta N^{\frac{2}{\d}-1}$ instead of~$\beta$. As seen in~\cite{lebles} it turns out to be a natural choice by scaling considerations as our~$\beta$ corresponds to the effective inverse temperature governing the microscopic scale behavior, with a balance in the energy and entropy  competition at the local level. Of course, this choice does not reduce generality. Indeed, since our estimates are explicit in their dependence on~$\beta$ and~$N$, one may choose~$\beta$ to depend on~$N$ if desired. 

The Coulomb gas, also called ``one-component plasma" in physics, is a standard ensemble of statistical mechanics, which has attracted much attention in the mathematical physics literature, see for instance  \cite{martinreview,aj,CDR,sm,kiessling,messerspohn,imbrie,brydgesfeder} and references therein. Its study in  the two-dimensional case is more developed, thanks in particular to its connection with Random Matrix Theory (see \cite{dyson,mehta,forrester}): when
$\beta=2$ and $V(x)=|x|^2$, \eqref{def:PNbeta} is the law of the (complex)  eigenvalues of the Ginibre ensemble of  $N\times N$ matrices with  normal Gaussian i.i.d entries \cite{ginibre}. Several additional motivations come from quantum mechanics, in particular via the plasma analogy for the fractional quantum Hall effect \cite{Gir,stormer,laughlin2}. For all these aspects one may refer to \cite{forrester}. The Coulomb case with $\d = 3$, which can be seen as a toy  model for matter has been for instance studied in \cite{jlm,LiLe1,LN}. The study of higher-dimensional Coulomb systems is not
as much developed.   In contrast the one-dimensional log gas analogue has been extensively studied, with many results of CLTs for fluctuations, free energy expansions, and universality \cite{joha,shch,BorGui1,BorGui2,bey1,bey2,bfg,bl,hardylambert}.
\cor{The case of the one-dimensional Coulomb gas, for which the interaction is $\g(x)=-|x|$ was studied quite thoroughly in \cite{lenard1,lenard2,kunz}.} 

In  Coulomb systems, if $\beta $ is fixed   and if $V$ grows fast enough at infinity, then as $N \to \infty$, the empirical measure
$$\mu_N := \frac{1}{N} \sum_{i=1}^N \delta_{x_i}$$ converges almost surely  under the Gibbs measure to a deterministic equilibrium measure $\mu_\infty$  with compact support and density equal to $\cd^{-1}\Delta V$ on its support, which can be identified as the unique minimizer among probability measures of the quantity
 \begin{equation}
  \label{defE}
 \mathcal E^V(\mu)= \hal \int_{\R^\d\times \R^\d} \g(x-y) d\mu(x) d\mu(y)+ \int_{\R^\d} V(x) d\mu(x),\end{equation}
see for instance~\cite[Chap. 2]{ln}.
This behavior in fact persists when $\beta $ tends to $0$ as $N \to \infty $ as long as $\beta \gg N^{-2/\d}$, as we will see just below.

 The lengthscale of the support of $\mu_\infty$, independent of $N$,  is of order~$1$, it is called the {\it macroscopic scale}, while the typical interparticle distance is of order $N^{-1/\d}$ and is  called  the {\it microscopic scale} or \emph{microscale}. Intermediate lengthscales are called \emph{mesoscales}.

\smallskip

Following \cite{as}, instead of $\mu_\infty$ we work with a deterministic correction to the equilibrium measure which we call the  {\it thermal equilibrium measure}, which is appropriate for all temperatures and defined as the probability density $\mub$ minimizing 
\be \label{1.9}
\mathcal{E}_{\theta}^V (\mu):= \mathcal{E}^V (\mu) + \frac{1}{\theta} \int_{\R^\d} \mu\log \mu\ee
with
\be\label{relbt}
\theta:= \beta N^{\frac{2}{\d}}.\ee
Here, and in the sequel, we use the same notation for the measure $\mu$ and its density.
 By contrast with $\mu_\infty$, $\mub$ is positive and regular in the whole of $\R^\d$ with exponentially decaying tails. 
It is well-known to be the limiting density of the point distribution in the regime in which $\theta$ is fixed independently of $N$ and we send $N\to \infty$, that is, for $\beta \simeq N^{-\frac 2d}$; see for instance \cite{kiessling,messerspohn,CLMP,bodgui,ab}.   

  The precise dependence of $\mub$ on $\theta$ has been studied in \cite{ascomp} where
it is shown that when $\theta \to \infty$, then $\mub$ converges to $\mu_\infty$, with quantitative estimates (see below). 
Using the thermal equilibrium measure allows us to obtain more precise quantitative results valid for the full range of~$\beta$ and~$N$ such that $\theta\gg 1$, allowing in particular regimes of   small~$\beta$. Our method would also work for fixed $\beta$ using the standard equilibrium measure $\mu_\infty$, as in \cite{ls2}, but the thermal equilibrium measure always yields more precise results  and a more precise description of the point distribution.

\cor{The inverse temperature $\theta$ is an important parameter in this problem because  $\theta^{-1/2}= \beta^{-1/2} N^{-1/\d}$ turns out to be the {\it characteristic lengthscale}  that governs both the  {\it macroscopic and microscopic distributions of the particles}. Concerning the macroscopic distribution we mean that $\theta^{-1/2} $ is the lengthscale of the tails of the noncompactly supported equilibrium measure $\mut$, as was shown in \cite{ascomp}. Concerning the microscopic distribution, we mean that $\theta^{-1/2}$ is very similar to the {\it minimal lengthscale for rigidity} $\rho_\beta$ introduced in \cite{as} and described further below.   
As we are interested for the first time in getting results that are valid for $\beta$ possibly depending on $N$, it turns out that {\it the only parameters that  matter are $\theta$ and the ratio of the considered lengthscale to the minimal lengthscale}, essentially  our results hold whenever  both are large.
 }
\smallskip


We are interested in two related things: one is in obtaining free energy expansions with explicit error rates as $N\to \infty$, and the other is in obtaining Central Limit Theorems for the fluctuations of linear statistics of the form 
\be \label{defFluct}\Fluct(\xi):=\sum_{i=1}^N \xi(x_i)-N\int \xi d\mut(x),\ee
with $\xi$ regular enough.
These two questions are directly related, indeed, as is well-known and first observed in this context by Johansson \cite{joha}, studying the  fluctuations  is conveniently done by computing  their Laplace transform, which then reduces the problem to computing the ratio of partition functions of two Coulomb gases with different potentials, and so obtaining very precise expansions for these partition functions is key. 
In this paper we will show that if one has an expansion of $\log \ZNbeta$ with a sufficiently good error rate, then one can obtain a CLT for the fluctuations in all dimensions.
 The needed rate will be obtained and thus the proof completed in dimension $2$, \cord{with quantitative convergence. The needed  rate is so far not available in dimension $3$ and larger, however we believe that the rate we obtain here is suboptimal (it is for instance worse than the one obtained in \cite{as} for the case with uniform background)  leading to expect that a CLT should hold for  larger dimensions as well. }

\subsection{Comparison with the literature} The study of free energy expansions for Coulomb gases in general dimensions $\d\ge 2$ was initiated in \cite{ss1,rs,lebles}. 
This program of using free energy expansions to derive CLTs for fluctuations of linear statistics was already accomplished in dimension $2$ in \cite{ls2} and \cite{bbny2} with  a slightly different proof (one based on transport, the other on loop equations), however only the case of fixed $\beta$ was treated.  Prior CLT results restricted to the determinantal case $\beta=2$ were obtained in  \cite{ridervirag,ahm}.  The results of \cite{ahm,ls2} were the only ones to treat the case where the support of $\xi$ can overlap the boundary of the support of $\mu_V$. Recently, \cite{lz} obtained the first  ``local CLT" in dimension $2$ by using the transport method of \cite{ls2} on the characteristic function.

Here we are particularly interested, like in the companion paper \cite{as}, in obtaining such  results  for a broad range of regimes of $\beta$, possibly depending on $N$ and allowing for very large or very small temperatures. Also, while the results in \cite{ls2,bbny2} were the first ones to obtain {\it mesoscopic CLTs} in dimension 2, i.e. to treat the case of $\xi$ supported on small boxes, they were limited to  lengthscales $\ell\ge N^{\alpha}$, $\alpha>-1/2$, i.e. to mesoscales, while here we can treat all scales down to the {\it temperature-dependent microscale} $\rho_\beta$  introduced in \cite{as} and defined in \eqref{rhobeta},  below which rigidity is expected to be lost.  We will not however treat the boundary case as in \cite{ahm,ls2} and will restrict to functions $\xi$ that are both sufficiently regular and  supported in the ``bulk", here defined as the set where the density of $\mut$ has a good bound from below. \cor{In fact it is known in the physics literature that the Coulomb gas density has more fluctuations near the edge, see \cite{csa} and references therein,  so this limitation is not purely technical, and we do not expect similar results to ours to hold near the boundary when looking at small scales}.

Treating the case of  nonsmooth $\xi$, in particular  equal to a characteristic function of a ball or cube in \eqref{defFluct}, i.e. evaluating the number of points in a given region, remains a (significantly more) delicate problem.  In particular one would like to show whether {\it hyperuniformity} (see \cite{To})  holds, i.e. whether the variance of the number of points in boxes is smaller than that of a Poisson point process.

The study of fluctuations of linear statistics is much more developed for ensembles in dimension 1, particularly log gases (or $\beta$-ensembles), see the works of \cite{joha,shch,BorGui1,BorGui2,llw,bmp}, and also recently Riesz gases \cite{boursier}. 
In terms of temperature regimes, the study of fluctuations for large temperature regimes has only recently attracted attention, also mostly for log gases or $\beta$-ensembles in dimension $1$, see \cite{bgp,nt1,nt2} and  \cite{hardylambert} (itself based on the Stein's method approach of \cite{llw}). 

In terms of studying fluctuations for dimensions larger than 2, the main progress was made in work of Chatterjee  \cite{chatterjee}, followed by Ganguly-Sarkar  \cite{ganguly}, who  analyzed a {\it hierarchical Coulomb gas model}.   
This is a simplified model, introduced by Dyson, in which the interaction  is coarse-grained at dyadic scales. They studied it   in a temperature regime that corresponds to $\beta = N^{1/3}$ (i.e. very low temperatures) in our setting. They  obtained  bounds on the  variance of  number of points in boxes and of linear statistics, but still no CLT. 

 In the physics literature, the papers \cite{jlm,lebowitz} (see also \cite{martinreview,martinyalcin}) contain  a well-known prediction of  an order $N^{1-1/\d}$ for the variance of the number of points in a domain, however there is no prediction for the order of fluctuations of  smooth linear statistics. In \cite{chatterjee} the  order of fluctuations of smooth linear statistics was speculated upon ($N^{1/3}$ vs. $N^{1/6}$) with supporting arguments from the example of orthogonal polynomial ensemble treated in \cite{bhardy} in favor of $N^{1/3}$, and finally it was shown in \cite{ganguly} to be in $N^{1-2/\d}$, again still for the hierarchical model instead of the full model and for $\beta$ of order $N^{1/3}$. 
 \medskip

Going from free energy expansion to CLT involves a step which is often treated in dimension 1 or 2  via ``loop equations" also called  ``Dyson-Schwinger equations" (see \cite{bbny2}) or ``Ward identities" (see \cite{ridervirag,ahm}) and techniques related to complex analysis, which are inherently two-dimensional. These equations involve singular terms which are delicate to control. In \cite{ls2}, we introduced a transport approach, based on a  change of variables transporting the original equilibrium measure to the perturbed one (perturbed by the effect of changing $V$ into $V+t\xi$),  which essentially replaces the loop equations. 
It was a question whether that approach could be extended to dimensions $\d\ge 3$  where the ``loop equations" are even more singular. Here we show that it is possible, and that to do so the terms arising in the loop equations have to be understood in a properly ``renormalized" way which allows to bound them by the energy. The main result expressing this is Proposition \ref{prop:comparaison2} which allows to control the first and second derivatives of the energy of a configuration along a transport path by the energy itself, see also Remark~\ref{rem4} which explains how to renormalize the loop equation terms. That crucial proposition is in line with a similar result in \cite{ls2} but it is significantly improved compared to \cite{ls2}: first it is extended to arbitrary dimension, and second the estimates are refined to give a control not only of the first but also of the second derivative.

In \cite{as},  a free energy expansion with a rate was obtained in the case of a uniform background measure (or equilibrium measure). Here, a free energy expansion is obtained for  a varying equilibrium measure by transporting it (locally) to a uniform one and using the aforementioned proposition to estimate the difference. The error rate obtained this way is \cord{less good than the one in the uniform measure case, and we believe there is room for improving that rate, which would be sufficient to conclude at least in dimension $3$ for small enough $\beta$}.

The proof crucially leverages on the local laws obtained in \cite{as}, which is the reason we cannot go below the scale $\rb$ at which local laws hold (and do not necessarily  expect the same CLT to hold then)  and on the use of thermal equilibrium measure introduced there.

\section{Main results}

\cor{In all the paper, we will denote by $C$ a generic positive constant independent of the parameters of the problem, but which may change from line to line.}
We will use the notation $|f|_{C^\sigma}$ for the  H\"older semi-norm of order $\sigma$ for any $\sigma \ge 0$ (not necessarily integer). 
For instance 
$|f|_{C^0}=\|f\|_{L^\infty}$, $|f|_{C^k}= \|D^k f\|_{L^\infty}$ and if $\sigma \in (k, k+1)$ for some $k$ integer, we let
$$ |f|_{C^\sigma(\Omega)}= \sup_{x\neq y \in \Omega} \frac{|D^k f (x)- D^k f (y)|}{|x-y|^{\sigma-k}}.$$
We emphasize that with this convention $f\in C^k$  does not mean that $f$ is $k$ times differentiable but rather that \cor{$D^{k-1} f$} is Lipschitz.

\subsection{Assumptions and  further definitions}\label{sec21}
We assume 
\be \label{assumpV1} V \in C^{2m+\gamma}\qquad \text{for some integer} \, m \ge 2 \ \text{and some} \ \gamma \in (0,1],\ee 
\be\label{assumpV2} 
\begin{cases} V\to +\infty \ \text{as} \ |x|\to \infty & \text{if}\ \d\ge 3\\
 \liminf_{|x|\to \infty} V+\g =+\infty& \text{if} \ \d=2,
\end{cases}\ee
 \be\label{assumpV3}
 \left\{
\begin{aligned}
& \int_{|x|\ge 1} \exp\(-\frac{\theta}{2}  V(x) \)dx<\infty, 
& \mbox{if} & \ \d\geq 3, \\
& \int_{|x|\ge 1} e^{-\frac{\theta}{2}  (V(x)-\log |x| )} \,dx
+ 
\int_{|x|\ge 1} e^{ -\theta  (V(x)-\log |x| ) }|x| \log^2 |x|\, dx<\infty
& \text{if} & \ \d=2,
\end{aligned}
\right. 
\ee
These assumptions ensure the existence of the standard equilibrium measure $\muv$ and the thermal equilibrium measure $\mut$ (see \cite{ascomp} for the latter).  We recall that the equilibrium measure is characterized by the fact that there exists a constant $c$ such that $\g* \muv+V-c$ is $0$  in the support of $\muv$ and nonnegative elsewhere. 
 We  let $\Sigma:=\supp\,  \muv$ and assume that $\partial \Sigma \in C^1$. 
 We also assume \cor{the nondegeneracy conditions} that 
\be \label{assumpV4} \Delta V \ge\alpha>0 \quad \text{in a neighborhood of }\, \Sigma\ee and that 
$$\g* \muv + V-c \ge \alpha \min (\dist^2(x, \Sigma) , 1) $$ \cor{which for instance hold if $V$ is strictly convex}.
Note that \eqref{assumpV1} and \eqref{assumpV2} imply that $V$ is bounded below.

These assumptions allow us to use the results of \cite{ascomp} on the thermal equilibrium measure, which we now recall.
They show that $\muv$ well approximates $\mut$ except in a boundary layer of size $\theta^{-\hal}$ near $\partial \Sigma$. More precisely 
there exists $C>0$ (depending only on $V$ and $\d$) such that 
\be\label{lbmuv}
 \mut \ge \frac{1}{C} \quad \text{in } \Sigma,\ee
\be \label{intromutsc}\mut(( \Sigma)^c)\le C \theta^{-\hal}, \qquad \left|\int_{ \Sigma^c}\mut \log \mut \right|\le  C \theta^{-\hal} ,\ee
and letting  $f_{k}$ be defined iteratively by 
\be \label{41} f_0= \frac{1}{\cd}\Delta V\qquad f_{k+1}=\frac{1}{\cd}\Delta V+ \frac{1}{\theta \cd} \Delta \log f_k\ee
we have  $|f_k|_{  C^{2(m-k-1)+\gamma} (\Sigma)}\le C $ and  for every even integer $n\le 2m-4$ and $0\le \gamma'\le \gamma$,  if $\theta\ge \theta_0(m)$, we have 
 for all $U \subset \Sigma$
\be\label{mubfm} |\mub-f_{m-2-n/2}|_{C^{n+\gamma'}( U)} \le C \theta^{\frac{n+
\gamma'}{2}}\exp\(- C \log^2 (\theta \dist^2(U, \partial \Sigma))\)  
+ C\theta^{1+n-m+\frac{\gamma'}{2}} .\ee
The functions $f_k$ provide a sequence of improving approximations (which are absent if $V$ happens to be quadratic) to $\mut$ defined iteratively.  Spelling out the iteration we easily find the following approximation in powers of $1/\theta$
\be\label{corrections}\mub\simeq \frac{1}{\cd}\Delta V+ \frac{1}{\cd \theta} \Delta \log \frac{\Delta V}{\cd} + \frac{1}{\cd \theta^2}\Delta \( \frac{\Delta \log \frac{ \Delta V}{\cd}}{\Delta V}\)+... \quad \text{well inside }   \Sigma\ee
up to an order dictated by the regularity of $V$ and the size of $\theta$. In our proof,  we will have to stop the approximations at a level which we denote $q$ and which will depend on the regularity of $V$, i.e. on $m$.

In all the explicit formula in the results, the quantity  $\mub$ could thus be replaced by $\muv$ or more precisely by \eqref{corrections} if $\theta$ is large enough, while making a small error quantified by \eqref{mubfm}.

 Throughout the paper, as in \cite{as} we will use the notation  
  \be\label{defchib}
\chi(\beta)= \begin{cases}   1  & \text{if} \ \d\ge 3\\
  1 +\max (-\log \beta, 0)& \text{if} \ \d=2,\end{cases}  \end{equation}
  and emphasize that $\chi(\beta)=1$ unless $\d=2$ and $\beta $ is small.  The correction factor $\chi(\beta)$ arises in dimension 2 at small $\beta$ and reflects the fact that the Poisson point process is expected (in dimension 2 only) to have an infinite Coulomb interaction  energy (see the discussion in \cite{as}).

   In \cite{as} we introduced the  minimal scale $\rb$ which is defined as 
\be \label{rhobeta}  \rb  = C \max\(1, \beta^{-\hal} \chi(\beta)^{\frac12} ,\beta^{\frac{1}{\d-2}-1}\indic_{\d\ge 5}\) 
\ee for some specific $C>0$, with $\chi$ defined above. We believe that $\rb$ should really be just $\max(1, \beta^{-\hal} \chi(\beta)^{\hal}) $, the third term in \eqref{rhobeta} appearing only for technical reasons.
 Note this lengthscale is measured in blown-up coordinates, in original coordinates the minimal lengthscale for ``rigidity"  is thus $N^{-1/\d} \rho_\beta$.
\cor{Neglecting the logarithmic correction in dimension 2, this lengthscale is thus expected to be $N^{-1/\d}\max(1, \beta^{-1/2})$ i.e. $\max(N^{-1/\d}, \theta^{-1/2})$, hence our claim at the beginning that $\theta^{-1/2}$ is also  the characteristic lengthscale for the microscopic distribution of the points (when $\beta\le 1$).}

In \cite{as} we proved that \cor{wherever} $\mut$ is bounded below,  \cor{for instance}   in $ \Sigma$ (by \eqref{lbmuv}),  local laws controlling   the energy in mesoscopic boxes down to the minimal scale $\rb$ (see Proposition~\ref{th3} for a precise statement) hold at a distance $\ge d_0$ from the boundary, where $d_0$ is defined by 
 \be\label{defd0}d_0:=
C \max\( \(\frac{N^{\frac1\d} }{ \max(1, \beta^{-\hal} \chi(\beta)^{\frac12})}\)^{-\frac23}   ,   N^{\frac{1}{ \d+2}-\frac1\d}\) \ee
for some appropriate $C>0$. \cor{Again, we do not expect such local laws to necessarily hold up to the boundary, due to the high oscillations of the gas there, see \cite{csa}.}

This leads us to defining a set $\hat\Sigma$ as a subset of $ \Sigma$ made of  those $x$'s  such that  
\be \label{defU}\theta^{m-2+\frac{\gamma}{2}} \exp\(-C\log^2 \(\theta \, \dist^2(x, \pa \Sigma) \) \) \le C \quad \text{and } \dist(x, \pa \Sigma) \ge d_0.\ee  For any $\ep>0$,  a  distance $\ge \theta^{\ep-\hal}+d_0$ from $\pa \Sigma$ suffices  to satisfy the first condition. But by definition 
\be\label{d0min} d_0 \ge  CN^{-\frac1\d} N^{\frac{1}{3\d}} \beta^{-\frac13} = C \theta^{-\frac13} ,\ee  hence the desired condition is satisfied. 
Thus we may absorb $\theta^{\ep-\hal}$ into $d_0$ and simply define 
\be\label{defocs}
\hat{\Sigma}:=\{x\in  \Sigma, \dist(x, \p \Sigma) \ge d_0\}.\ee
With this choice, in view of \eqref{mubfm}, we  have  that for $\theta\ge \theta_0(m)$
\be \label{bornemutcn} \forall \sigma\le 2m+\gamma -4,\quad
|\mut|_{C^{\sigma} \(\hat\Sigma\)} \le C .\ee

\medskip  

In all the sequel we need to assume that our test function is \cor{supported in a  cube of sidelength}  $\ell$ (possibly depending on $N$) with  
\be\label{ass1}
 \rb N^{-\frac1\d} < \ell \le C\ee  i.e. 
larger or equal to the minimal lengthscale for rigidity. This natural condition, implies in view of the definition of $\theta$ and since $\rb\ge \beta^{-\hal}$,  that 
\be \label{ass1coro}
C \ge \ell \ge \theta^{-\hal} \ee 
will always be verified. This in turn implies that   $\theta $ is bounded below independently of $N$, up to changing $C$ we may say $\theta \ge \theta_0(m) $, hence in particular \eqref{bornemutcn} holds. 
Our results will require some regularity on $V$ and $\xi$, we have not tried to optimize  the regularity assumptions.  Most of our results will not really depend on $V$  but will be valid for general  background densities $\mu$ (generalizing $\mut$) with perturbations taken  in a region where $\mu$ is bounded below and where  the properties  \eqref{bornemutcn} hold. 
All the parameters in our results, in particular $\beta$, the lengthscale $\ell$ and the test function $\xi$ may depend on $N$, but all the constants in our statements will depend only on $\d$ and on $V$ (really via the bounds \eqref{bornemutcn} and a lower bound on the density $\mu$).
\subsection{Concentration results: bounds on fluctuations}

We start with a first bound on the fluctuations (as defined in \eqref{defFluct}) with minimal assumptions on the regularity of  the test function $\xi$.
Let us emphasize that this result requires no heavy lifting, it is a rather quick consequence of  our energy splitting with respect to the equilibrium measure and electric formulation. 

In the sequel $Q_\ell$ will denote a hyperrectangle with sidelengths in $[\ell, 2\ell]$, not necessarily centered at $0$.

\begin{theo}[First bound on the Laplace transform in any dimension] \label{firsttheo}
Let $\d\ge 2$. Assume  $V\in C^7$, \eqref{assumpV2}--\eqref{assumpV4} hold, and $\xi\in C^{3}$, $\supp\, \xi \subset Q_\ell  \subset \hat \Sigma$  with $\ell$ satisfying \eqref{ass1}. There exists $C>0$ depending only on $\d$ and $V$ such that the following holds.  
For every $t$ such that 
\be\label{assmth1} C| t|\max\(   |\xi|_{C^2}, |\xi|_{C^1}\)<1
\ee
we have
\begin{multline}\label{bornelapq0}
\left|\log\Esp_{\PNbeta} \(\exp\(\beta tN^{\frac{2}{\d}}\Fluct(\xi) \)\) \right|
\le C |t|\beta N \ell^\d\(    \chi(\beta) \ell|\xi|_{C^{3}} +   \frac{1}{\beta}    |\xi|_{C^{2} }   \)  \qquad \qquad \qquad
\\+ Ct^2  \(N \ell^\d  |\xi|_{C^2}^2+|\supp \nab \xi| \beta N \(   N^{\frac2\d}  |\xi|_{C^1}^2+ \frac{1}{\beta}|\xi|_{C^2}^2   \) \)
+  
 C  N \ell^\d  t^4    |\xi|_{C^{2}}^4
 \end{multline} \cor{where $|\supp \nab \xi|$ denotes the volume of the support of $\nab \xi$.}
   \end{theo} \cor{This result is already stronger (in terms of regularity required for $\xi$ and bounds obtained) and more general  (in terms of temperature regime and dimension) than the prior results such as \cite{rs,ls2,bbny2} obtained for fixed $\beta>0$}.
 To illustrate, let us consider that $|\xi|_{C^k}\le C \ell^{-k} $ which happens for instance if $\xi$ is the rescaling at scale $\ell$ of a fixed function. Applying  Theorem \ref{firsttheo} in dimension~2 in this situation we get
\begin{coro}  \label{corod2} Assume the same as above, that $\d=2$ and  $|\xi|_{C^k} \le M \ell^{-k}$ for $k\le 3$  with $\ell$ satisfying \eqref{ass1}.  Then 
\cor{for all $|\tau|< C^{-1}M^{-1} (N^{1/\d}\ell)^2 \max(\beta, 1)$\footnote{
\label{note1}
In particular $|\tau |<C^{-1}$ suffices in view of \eqref{ass1}.}
we have 
\be \left|\log\Esp_{\PNbeta} \(\exp\( \tau  \min(1, \beta) |\Fluct(\xi)| \)\) \right|
\le C (1+\tau^4 M^4)  \ee 
where $C$ depends only on $V$.}
\end{coro}
\cor{By Tchebychev's inequality, it immediately implies a concentration result:  for any $t>0$, we have
\be\label{tcheb}\PNbeta(\min(1, \beta) |\Fluct (\xi) |>t) \le \exp\( -  t + C(1+M^4)\)\ee where $C$ depends only on $V$, 
thus this immediately  implies that $\Fluct(\xi)$ is typically bounded by $ \max(1,\beta^{-1})$ as $N\to \infty$, a result which is new if $\beta \ll 1$.
}

An analogous result to Corollary \ref{corod2} in dimension $\d\ge 3$ is stated in Corollary~\ref{5.5}.
If $\xi$ is  assumed to be more regular, we can obtain  a better  estimate in dimension \cor{$\d \ge 2$}, for instance
we have
\begin{coro}\label{coro22} Assume the same as above, \cor{$\d \ge 2$}, $V\in C^\infty$ and $\xi\in C^\infty$ with  $|\xi|_{C^k} \le M \ell^{-k}$ for each $k$,  with $\ell$ satisfying \eqref{ass1}.
Then if $\beta \ge (\ell N^{\frac1\d} )^{2-\d}$, for all $|\tau|<C^{-1}M^{-1}\beta N \ell^\d$, \emph{\textsuperscript{\ref{note1}}} we have 
$$
\left|\log\Esp_{\PNbeta} \(\exp\( \tau ( N^{\frac{1}{\d} }\ell)^{2-\d}
 |\Fluct(\xi)  |\)\) \right|\\ \le C |\tau | M + C \tau^4 M^4$$
while if $\beta \le  (\ell N^{\frac1\d} )^{2-\d}$, for all $|\tau|<C^{-1}M^{-1} \beta^{1/2} (N^{1/\d}\ell)^{1+\d/2}$, \emph{\textsuperscript{\ref{note1}}}
 we have 
$$\left|\log\Esp_{\PNbeta} \(\exp\(\tau \beta^{\hal}  (  N^{\frac{1}{\d}}\ell)^{1-\frac{\d}{2}}
 |\Fluct(\xi)| \)\) \right|\le C |\tau| M+ C\tau^4  M^4,$$
 where $C$ depends only on $V$ and $\d$.
\end{coro}
A more general estimate is obtained in  \eqref{bornelap}. 
Let us also point out that we expect the quantities that we have bounded to have a divergent mean (unless $\beta$ tends to $0$) and a smaller variance, see Theorem \ref{th5}, thus once that mean is removed, the bounds we have obtained should not typically be optimal.

  These results, or their reformulation as in \eqref{tcheb},  may be compared  to prior concentration results of \cite{chm,ber}  in the regime of fixed $\beta$ and  to the recent result of  \cite{pg} which is the first one in the framework of the thermal equilibrium measure, thus allowing $\beta \to 0$. These prior results are in terms of distance from the empirical to equilibrium measure, rather than in terms of direct  bounds on fluctuations. 
   \subsection{Next order free energy expansion}

Our next result concerns free energy expansions with a rate for  general equilibrium measures  whose density varies.  In \cite{as} we obtained  a free energy expansion for uniform equilibrium measures in cubes, with an explicit error term proportional to the surface. It is expressed in terms of  a function $\mf(\beta)$, \cor{the free energy per unit volume}, characterized variationally in \cite{lebles} as the minimum over stationary point processes of $\beta $ times  a Coulomb ``renormalized energy" (from \cite{ss1,rs}) plus  a (specific) relative entropy. More precisely, 
 there is a constant $C>0$ depending only on $\d$ such that   \be\label{bornesurf}
 -C\le \mf(\beta)\le C \chi(\beta)\ee 
 \be\label{bornesurfp} 
 \text{$\mf$ is locally Lipschitz in $(0,\infty)$ with} \  |\mf'(\beta)|\le \frac{C\chi(\beta)}{\beta}
 ,\ee 
 and such that 
 if $R^\d$ is an integer we have
  \be \label{1.26} \frac{\log \K(\car_R,1)}{\beta R^\d}=-   \mf(\beta) + 
 O\( \chi(\beta)\frac{\rb }{R}   + \frac{\beta^{-\frac1\d} \chi(\beta)^{1-\frac1\d} }{R}\log^{\frac1\d} \frac{R}{\rb}  \)
 \ee
 where $\rb$ is as in \eqref{rhobeta}.  Here $\K(\car_R,1)$ is the appropriate partition function for a zoomed Coulomb gas with density $1$ in $\car_R$, the cube of sidelength $R$ when $R^\d$ is integer, see \eqref{defK} for a precise definition.   The existence of the large volume limit of the free energy per unit volume $\mf(\beta)$ was shown in dimension 2 in \cite{sm} and dimension 3 in  \cite{LN}, but here the novelty is in the rate of convergence.   
 The proof in \cite{as} relies on showing almost additivity of the free energy over cubes, via  comparison with a subadditive and a superadditive quantity. In effect, this amounts to showing that in the large volume limit, the free energy does not depend on the boundary conditions chosen, up to surface energy errors. This is very much in line with the physics literature, for instance \cite{brydgesfeder,imbrie,kunz}
 and accomplished  via a screening procedure, originating in the Coulomb gas context \smallskip in \cite{ss1}.

From this, the idea is to obtain expansions for general equilibrium measures by partitioning the system into small cubes over which $\mut$ is close to uniform, using the almost additivity of the free energy of \cite{as}, and computing the difference in free energies in each cube by transporting the almost uniform measure $\mut$ to the uniform measure of density equal to the average value.
The errors will  lead to a degraded error estimate compared to \eqref{1.26}. We believe that such a degradation is unavoidable by this method as the variations of $\mut$ introduce a ``soft kind" of  boundaries between regions of different point densities, and we  do not believe our estimates to be  optimal. 

\begin{theo}[Free energy expansion, general background]\label{thglob}  Assume $\d \ge 2$.
Assume $V\in C^5$ satifies \eqref{assumpV1}-- \eqref{assumpV4}.  
We have
\begin{multline}
\label{expvar}
 \log \ZNbeta=-\beta N^{1+\frac{2}{\d}}\mathcal E_\theta^V(\mub) +\frac{\beta}{4} (N\log N) \indic_{\d=2} -N  \frac{\beta}{4} \( \int_{\R^\d} \mub \log \mub\)
  \indic_{\d=2} \\
   + N \beta \int_{\R^\d} \mub^{2-\frac2\d} \mf(\beta \mub^{1-\frac{2}{\d}} )
   +  N\beta \chi(\beta) O( \mathcal{R})  
   \end{multline}  
where $\mathcal{R}\to 0$ as a power of $\rb N^{-1/\d}$.
\end{theo}
An explicit form of the error term is  given in the paper  in Theorem \ref{thglob2}. To illustrate, if $\beta$ is of order $1$, then the error obtained is 
$NO(\mathcal R)= O( N^{1 -\frac1\d+ \frac1{\d+2}} \log N  )$. This is a degradation compared to the rate in $(N^{\frac1\d})^{\d-1}$  obtained in \cite{as} for uniform densities (and corresponding to a surface error). The
 largest part of the error is anyway created by a boundary layer imprecision due to the lack of local laws near the boundary.
Our results of course agree with previous ones \cite{lebles,loiloc, bbny2} \footnote{The formula appears different because it is expressed in terms of $\mathcal E_\theta^V$ instead of $\mathcal E^V$ as is usually done, so some order $N$ terms are hidden in the entropy part of  $\mathcal E_\theta^V$. }  and improve them with the explicit rate,  and also agree with the predicted formulas for two dimensions in particular in \cite{zw}, see also \cite{dyson}. 

Note also that in dimension 2 and in the case of  quadratic $V$, \cite{cftw,sha} predict an expansion for $\log \ZNbeta$ in powers of $N^{\hal}$ hence where  the next order term  is $\sqrt N$, which corresponds to a boundary term. This $\sqrt N$ term was missing in \cite{zw}.

What will be crucial for us in the sequel is that we can also obtain a localized version  for the relative expansion of the free energy, i.e. the difference of $\log \ZNbeta$ for two different equilibrium measures which only differ in a cube of size $\ell$ included in $\hat \Sigma$, see Proposition \ref{th1b} for a full statement. 
Then there is no boundary local law error and the error rate $\mathcal R$ can be expressed as a power of $\rb N^{-1/\d} \ell^{-1}$.  The power that we can obtain in all generality  is  $1/2$ yielding an error term $(N\ell^\d) ^{1-\frac{1}{2\d}}$  (modulo a logarithmic correction) for fixed $\beta$ of order 1, which is  still a degradation compared to the rate in $(N^{\frac1\d}\ell)^{\d-1}$  obtained in \cite{as} for uniform densities. It 
 is however sufficient to proceed with the proof of the Central Limit Theorems below in dimension $\d=2$ \cord{and barely fails to be sufficient in dimension $3$}.

\subsection{Central Limit Theorems}

To state the results, let us define the operator
$$L= \frac{1}{\cd \mut} \Delta .$$

We phrase the results as the convergence of the Laplace transform of the fluctuations to that of  a Gaussian. Compared to known results, explicit corrections to the known variance in $\int |\nab \xi|^2$ are given as powers of $\theta^{-1}$ when $\xi$ is regular enough, indicating the change in the formula for the variance in the cross-over regime when $\beta$ becomes small (reminiscent for instance of \cite{hardylambert}).  Moreover the variance $\int |\nab \xi|^2$ corresponds to a convergence (of $\Delta^{-1}(\sum_{i=1}^N \delta_{x_i} - N \mut)$)  to the Gaussian Free Field (GFF) while the expected variance when $\theta $ becomes order $1$ no longer corresponds to the GFF but rather to another Gaussian Field.

We should also emphasize that the normalization of the variable is
$\beta^\hal (N^{\frac1\d}\ell)^{1-\frac\d2}$, and 
 not $\frac1{\sqrt{N\ell^\d}}$ as in the usual CLT for a sum of independent variables. It is in fact a CLT for very nonindependent random variables. However, 
 $$\beta^{\hal}  (N^{\frac1\d}\ell)^{1-\frac\d2} \sim \frac{1}{\sqrt{N\ell^\d} } \( \frac{N^{\frac1\d}\ell}{\rb}\)$$   if one believes that $\rb \sim \beta^{-\hal}$ (see \eqref{rhobeta} and comments below)
 so in the extreme regime where $ N^{\frac1\d}\ell=\rb$ (which one can also read as $\theta \ell^\d =1$ or the large temperature regime) 
 we recover the standard CLT normalization for iid variables, because $N\ell^\d$ is the number of points in the support of $\xi$.
Physically, this means that the system is very rigid, and becomes less and less so as one approaches the minimal scale. When $\beta \ll 1$, there is a gap between the minimal scale $\rb N^{-1/\d}$ and the microscale $N^{-1/\d}$ and we expect that the system becomes Poissonian below the minimal scale, based on the lose heuristic that in  the Langevin dynamics the diffusion should dominate at small enough scale depending on temperature, and also by analogy with the case of $\beta$-ensembles \cite{bgp,nt1,nt2}.

\subsubsection{The case of dimension 2}
The first result is in dimension $2$ and extends the result of \cite{bbny2,ls2} to possibly small  or large~$\beta$  \cor{and down to the minimal scale.}

\begin{theo}[CLT in dimension 2 for possibly small $\beta$]
\label{th2} 
Let $\d=2$. Let $q \ge 0$ be an integer.  Assume  $V\in C^{2q+7}$, \eqref{assumpV2}--\eqref{assumpV4} hold, and $\xi\in C^{2q+4}$, $\supp\, \xi \subset Q_\ell \subset \hat\Sigma$  with $\ell$ satisfying \eqref{ass1}. 
Assume $N^{\frac1\d}\ell \gg \rb$ as $N \to \infty$, \footnote{ which implies $\theta \gg 1$, as seen before.}  and
 \be\label{condsup3d2}
\beta^{\hal} \ll (N^{\frac1\d}\ell)^{\frac14}\log^{-\cord{\frac34}}(N^{\frac1\d}\ell).\ee
Then  \cor{for any fixed $\tau$,} \footnote{\label{note3} An explicit rate of convergence in inverse powers of $N^{\frac1\d}\ell\rb^{-1}$, also depending on the rate in \eqref{condsup3d2}, is provided.}
\be \label{7230}\left|
\log \Esp_{\PNbeta} \(\exp\(-\tau \beta^\hal   \Fluct(\xi) \)\)+
\tau m(\xi)   - \tau^2  v(\xi) \right|\to 0 \quad \text{as} \  \frac{N^{\frac1\d}\ell }{\rb} \to \infty\ee
where
\be
v(\xi)=-
\frac{ 1}{2\cd}  \int_{\R^\d} \left|   \sum_{k=0}^q  \frac{1}{\theta^k} \nab L^k (\xi)\right|^2
+ \frac{1}{\cd}\int_{\R^\d}  \sum_{k=0}^q \nab \xi \cdot \frac{ \nab L^{k}( \xi) }{\theta^k} -
\frac{1}{2 \theta}\int_{\R^\d} \mut \left|\sum_{k=0}^q\frac{ L^{k+1}(\xi)}{\theta^k} \right|^2
\ee and 
$$m(\xi)=  -\frac{\beta^\hal }{4}\displaystyle \int_{\R^\d} 
 \( \sum_{k=0}^q\frac{ \Delta L^k (\xi)}{
\cd\theta^k}\)
 \log \mut .$$
\end{theo}
When neglecting the corrections in inverse powers of $\frac{1}{\theta}$  in the expressions for $m(\xi)$ and $v(\xi)$ as $\theta \to \infty$ we obtain
\begin{coro}Under the same assumptions, 
assume  $\xi=\xi_0(\frac{x-x_0}{\ell})$ for $\xi_0$ a fixed $C^4$ function. Then  $\beta^{1/2}\( \Fluct(\xi)+ \frac{1 }{4 \cd} \int_{\R^\d} ( \Delta \xi   ) \log \mut\) $  converges \footnote{\label{note2} The convergence is in the sense of convergence of the Laplace transforms, which implies convergence in law  but is in fact a bit stronger.}  as $N \to \infty $ to a Gaussian  of mean $0$ and variance  $ \frac1{\cd}\int_{\R^2} |\nab \xi_0|^2 $.
\end{coro}
By definition of the Gaussian Free Field (GFF) the convergence to a Gaussian with this specific variance can be expressed as a convergence  of $\beta^{1/2}$ times the electrostatic potential  (see Section \ref{secelec}) 
$$\Delta^{-1}\( \sum_{i=1}^N \delta_{x_i} - N \mut \),$$ suitably shifted, to the GFF, and  the same applies to the result in dimension 3 below.
Note here that the mean $m(\xi)$ may be an unbounded deterministic shift to the fluctuation, since $\beta$ may tend to $\infty$ as $N\to \infty$. Also the expression for $m(\xi)$ differs from that appearing in \cite{ls2} because the fluctuation is computed with respect to $\mut$ instead of $\muv$, and these differ by $\frac{1}{\cd\theta} \Delta \log \frac{\Delta V}{\cd}$  to leading order (see \eqref{corrections}). This difference  exactly matches the discrepancy in the expression for the mean.

When $\beta $ is so large that \eqref{condsup3d2} fails, \cor{we do not expect the same CLT to hold but} we can normalize $\Fluct(\xi) $ differently to obtain a convergence result. \cor{The fact that $\Fluct(\xi)$ without normalization converges to a limit again reflects a strong rigidity of the system, consistent with the fact that as $\beta \to \infty$, in this dimension we expect   crystallization to  a triangular lattice to happen, related to conjectures of   Cohn-Kumar and Sandier-Serfaty (see \cite{lebles,ps2})}.
\begin{theo}[Low temperature and minimizers] \label{thlowt2}Let $\d=2$. \cor{
  Assume  $V\in C^{7}$, \eqref{assumpV2}--\eqref{assumpV4} hold, and $\xi\in C^{4}$, $\supp\, \xi \subset Q_\ell \subset \hat \Sigma$  with $\ell$ satisfying \eqref{ass1}. 
Assume $\beta \gg 1$  and $N^{\frac1\d}\ell \gg 1$ as $N \to \infty$. Then as $N \to \infty$, we have \emph{\textsuperscript{\ref{note2}}}
$$\Fluct (\xi)+  \frac{1}{4\cd }   \displaystyle \int_{\R^\d} 
(\Delta \xi )
 \log \mut   \to 0.$$ }
If $\XN $ minimizes $\HN$ then the same result holds.
\end{theo}

The case of minimizers of $\HN$ corresponds to $\beta=\infty$ and can be obtained by simply letting $\beta \to \infty$ in the case with temperature since the constants are independent of $\beta$. Note that this generalizes \cite{ls2} and also completements the results on minimizers or very low temperature states in \cite{aoc,rns,ps,ar}.

\subsubsection{The case of higher dimension}

We now turn to dimension $3$ and higher. As announced above, we will need to assume more regularity on $\mf$, and even make a quantitative regularity assumption.

While we know that $\mf$ is locally Lipschitz (see \eqref{bornesurfp}), its higher regularity is not known and is a delicate question, since points of nondifferentiability of $\mf' $ correspond by definition to phase-transitions.
Assuming that $\mf''$ is bounded can thus be interpreted as assuming that there are no first order phase-transitions at the effective temperatures we are considering: it was noted in \cite{lebles} that in dimension  $\d \ge 3$ an effective temperature $\beta \mut(x)^{1-\frac2\d}$ appears, which depends on  both $\beta$ and the local particle density.

In the physics literature, the existence of phase transitions  is discussed in  dimensions 2 and 3, and is described as  ``despite an extensive literature, still a
subject of controversy" according to the recent paper \cite{csa}. But several papers discuss a phase transition observed numerically  in dimension 2 around $\beta =140$   \cite{clhw} and in dimension 3 around $\beta=175$  \cite{brush,jc}, see also the review \cite{kk} which proposes explicit expression for $f_\d (\beta)$.  
So in any case,  we expect  that the condition we place should be true for all but a finite number of $\beta$'s. Note also that our 2D result did not require any condition on $\beta$ despite the possible existence of a phase-transition, this is  due to the lack of $\mub(x)$-dependence in the expression involving $\mf$ in \eqref{expvar} in contrast with the case $\d \ge 3$.

When $\beta$ is very small or very large (i.e. when it tends to $0$ or to $\infty$ as $N $ diverges) and when $\d \ge 3 $ only,  we will  need  a quantitative assumption on the derivative of $\mf$: we will assume 
that 
\be \label{assfs}
\|\fd''\|_{ \mut, U}\le C\beta^{-2}\ee for some $C$ independent of $\beta$, 
where  \cor{for a generic set $U$} 
we denote 
\be \label{defnorme}
\|\fd''\|_{ \mut,U}=\sup_{x\in U} |\fd''(\beta \mut(x )^{1-\frac{2}{\d}})|.\ee
Note that we could do with just the assumption that $\fd'$ is bounded in some H\"older space $C^{0, \alpha}$, hence for simplicity we have assumed $\alpha=1$. When $\beta$ is fixed \eqref{assfs} is just a regularity assumption. When $\beta \to 0$ it is more quantitative, and  
it seems reasonable if one extrapolates from \eqref{bornesurfp}, assuming a regular behavior for the function $\fd(\beta)$,    
however we do not have a further basis for its reasonableness. One may refer to \cite{imbrie,brydgesfeder} for a treatment of the low $\beta$ regime by cluster expansions.

The improved  rate in $N^{1-\frac{1}{2\d}}$ obtained in Proposition \ref{th1b} \cord{does not quite suffice to deduce a CLT in dimension $\d \ge 3$, but as mentioned above we do not believe it to be optimal.}


The larger the dimension or the smaller the temperature, the more regular we need $\xi$ to be.

\begin{theo}[\cord{Conditional} CLT in dimension $\d \ge 3$ for possibly small $\beta$]
\label{th5} 
Let $\d\ge 3$.  Let $q>\frac\d4-1$ be a nonnegative integer.  Assume  $V\in C^{2q+7}$, \eqref{assumpV2}--\eqref{assumpV4} hold, and $\xi\in C^{2q+4}$, $\supp\, \xi \subset Q_\ell  \subset \hat \Sigma$  with $\ell$ satisfying \eqref{ass1}.    \cor{Assume   that  \eqref{assfs} holds relative to  $Q_\ell$}.
If $\beta \to 0 $ assume in addition that 
 \be \label{condsup3d}
\frac{N^{\frac{1}{\d}}\ell}{\rb} \ge N^{\ep} \quad \text{for some } \ \ep>0\ee
and  that  $q$ is larger than a constant depending on $\ep$.

 If a free energy expansion with  rate $\mathcal R$ as in Proposition \ref{coro24} is found to hold \footnote{when assuming $|\mu|_{C^1} \le \ell^{-1}$ in that proposition}
 with  \be\label{assR}
 \mathcal R \ll (N^{\frac1\d} \ell)^{\cord{\frac2\d-\d }} \beta^{-1},\ee
 then 
  for any fixed $\tau$, we have
\be \label{7240}\left|
\log \Esp_{\PNbeta} \(\exp\(-\tau \beta^\hal (N^{\frac1\d}\ell)^{1-\frac\d2}   \Fluct(\xi) \)\)+
\tau m(\xi)   - \tau^2 \ell^{2-\d} v(\xi) \right|\to 0 \quad \text{as} \  \frac{N^{\frac1\d}\ell }{\rb} \to \infty\ee
where
\be
v(\xi)=-
\frac{ 1}{2\cd}  \int_{\R^\d} \left|   \sum_{k=0}^q  \frac{1}{\theta^k} \nab L^k (\xi)\right|^2
+ \frac{1}{\cd}\int_{\R^\d}  \sum_{k=0}^q \nab \xi \cdot \frac{ \nab L^{k}( \xi) }{\theta^k} -
\frac{1}{2 \theta}\int_{\R^\d} \mut \left|\sum_{k=0}^q\frac{ L^{k+1}(\xi)}{\theta^k} \right|^2
\ee and 
$$m(\xi)= -
N \ell^2 \beta^{\hal} (N^{\frac1\d}\ell)^{- 1-\frac\d2}  \( 1-\frac2\d\)  \displaystyle \int_{\R^\d} \(  \sum_{k=0}^q\frac{ \Delta L^k (\xi)}{
\cd\theta^k}\)\( \fd(\beta \mut^{1-\frac2\d})+\beta \mut^{1-\frac2\d}\fd'(\beta \mut^{1-\frac2\d})\) .$$

\end{theo}

\begin{coro}Under the same assumptions, if \cor{$\xi=\xi_0(\frac{x-x_0}{\ell})$
for $\xi_0 $ a $C^{2q+4}$ function with $q$ large enough for $\beta \ll 1$, or $\xi_0\in C^4$ otherwise, 
then $\beta^\hal (N^{\frac1\d}\ell)^{1-\frac\d2}   \Fluct(\xi)+ m(\xi)
$ converges \emph{\textsuperscript{\ref{note2}}}
to a Gaussian of mean $0$ and variance $\frac1{2\cd}\int|\nab \xi_0|^2$.
}
\end{coro}
\cor{Just as in dimension 2, this can be interpreted as a convergence of $\beta^\hal (N^{\frac1\d}\ell)^{1-\frac\d2}\Delta^{-1}( \sum_{i=1}^N \delta_{x_i}- N \mub)$ to the GFF. This reveals a strong rigidity down to the minimal scale, but decreasing  as $\beta $ decreases. }

Again, when $\beta$ is large \cor{we expect crystallization to a lattice to happen \cite{lebles} and  do not expect the same result to hold.} We can instead obtain 
\begin{theo}[Low temperature and minimizers in dimension $\d \ge 3$] \label{thlowt3}
Let $\d \ge 3$.
  Assume  $V\in C^{7}$, \eqref{assumpV2}--\eqref{assumpV4} hold, and \cor{$\xi\in C^{4}$}, $\supp\, \xi \subset Q_\ell  \subset \hat\Sigma$ with $\ell$ satisfying \eqref{ass1}. Assume $\beta \gg 1$  as $N \to \infty$  and assume in addition that  \eqref{assfs} holds relative to  $Q_\ell$.
If a free energy expansion with a rate $\mathcal R$ as in Proposition \ref{coro24} is found to hold  
 with  \be\label{assR2}
 \mathcal R \ll (N^{\frac1\d} \ell)^{\cord{\frac2\d-\d}},\ee we have, as
 $N \to \infty$, \emph{\textsuperscript{\ref{note2}}}
\be\label{casmin}   (N^{\frac{1}{\d}}\ell)^{1- \frac\d2}\(\Fluct(\xi) - \frac{ N^{\frac{1}{3}} }{3\cd}
  \displaystyle \int_{\R^\d}   \Delta \xi\( \fd(\beta \mut^{1-\frac2\d})+\beta \mut^{1-\frac2\d}\fd'(\beta \mut^{1-\frac2\d})\)\)  \to 0.\ee
In particular if $\XN$ minimizes $\HN$ then 
the same result holds.
\end{theo}

Note the temperature regime studied in \cite{chatterjee,ganguly} corresponds to $\beta = N^{\frac13} $ for us and was in fact a  low temperature regime. Our result, conditional to \eqref{assfs} and an improved rate, would  thus be in agreement with (but   in principle stronger than) the result of variance in $N^{1/3}$ for $\Fluct(\xi)$  proved in 
\cite{ganguly} for the hierarchical model. 

\subsection{Outline of the proof}\label{secoutline}
As in our prior work \cite{ss1,rs,lebles,ls2,as}, the starting point is to use a next order Coulomb energy, defined for any probability density $\mu$ as 
\begin{equation}
\label{defF}
\F_N(\XN, \mu)=  \hal \iint_{\R^\d \times \R^\d \backslash \triangle} \g(x-y)d \(  \sum_{i=1}^N \delta_{x_i} -N \mu\) (x) d\(  \sum_{i=1}^N \delta_{x_i} - N\mu\) (y),\end{equation}
where $\triangle$ denotes the diagonal of $\R^\d \times \R^\d$. 
This next order energy appears when expanding $\HN$ around the appropriate measure, which is here $\mub$.
Recalling that $\theta= \beta N^{2/\d}$ and that the thermal equilibrium measure  minimizing  \eqref{1.9} 
satisfies 
\be \label{eqmb}
\g* \mub+V + \frac{1}{\theta}\log \mub=C_\theta\quad \text{in} \ \R^\d\ee
where $C_\theta$ is a constant, we obtain through an elementary computation  the following ``splitting formula",  found in \cite{as}:  for all configurations $\XN \in (\R^\d)^N$ with pairwise distinct points, \footnote{We can proceed as if configurations all had pairwise distinct points, since the complement event has  measure zero.} we have 
\be\label{splitting}
\HN(\XN)= N^2 \mathcal {E}^V_\theta(\mub)-\frac{N}{\theta}\sum_{i=1}^N \log \mub(x_i) +\F_N(\XN, \mub)
\ee
where $\mathcal{E}^V_\theta$ is as in \eqref{1.9}, $\F_N$ as in  \eqref{defF},  
and $\triangle$ denotes the  diagonal in $\R^\d\times \R^\d$. This separates the leading order $N^2 \mathcal{E}_\theta^V(\mub)$ from next order terms.
We see here $-\frac{1}{\theta} \log \mub$ playing the role of an effective confinement potential for the system at next order.

We may then define for any probability density $\mu$ the next order partition function
\begin{equation}\label{pdef}
\K_N(\mu)=  \int_{(\R^\d)^N} \exp\( - \beta N^{\frac2\d-1} \F_N(\XN, \mu)\) d\mu(x_1) \dots d\mu(x_N)\end{equation} 
and $\mathsf{Q}_N(\mu)$ the associated Gibbs measure.
Observe here that the integration is with respect to $\mu^{\otimes N}$ instead of the usual $N\d$-dimensional Lebesgue measure.

The use of the thermal equilibrium measure allows, via the splitting formula, for a remarkably simple rewriting of the partition function as 
\begin{equation}\label{resplitz}
\ZNbeta= \exp\(-\beta N^{1+\frac2\d} \mathcal {E}_\theta^V(\mut) \) \K_N(\mut)\end{equation} 
which is directly obtained by inserting \eqref{splitting} into \eqref{defZ} and using \eqref{pdef}. In prior works such as \cite{ss1,rs,ls2} the energy was split with respect to the usual  equilibrium measure, and this led to a less simple formula, involving 
an effective confinement potential. 

The study of the free energy expansion now reduces to the analysis of partition functions $\K_N(\mu)$ for general positive densities $\mu$.

The control of fluctuations and proof of the CLT is based on the Johansson approach \cite{joha} which consists in  evaluating the Laplace transform of the fluctuations, then directly reducing to evaluating the ratio of two partition functions, that for the Coulomb gas
with potential $V$ and that for the Coulomb gas with potential $V_t:= V+t\xi$ for a  small $t$. In the formulation with the thermal equilibrium measure, in view of \eqref{resplitz} this takes the simple form
\be \label{laplace0}
\Esp_{\PNbeta} \( e^{-\beta t N^{\frac{2}{\d}}   \sum_{i=1}^N \xi(x_i)  }    \)
=    \frac{Z_{N,\beta}^{V_t}}{\ZNbeta} = \exp\( - \beta N^{1+\frac{2}{\d}} (\mathcal{E}_\theta^{V_t} (\mutt) - \mathcal{E}_\theta^V (\mut) )  \)\frac{\K_N(\mutt)}{\K_N(\mut)}.
\ee where 
$\mutt$ is the thermal equilibrium measure associated to $V_t$.
In order to prove the CLT, one needs to show that the right-hand side converges as $N \to \infty$ to the Laplace transform of an appropriate Gaussian law. 
The precise value of $t$ to be taken here always ends up being small,  to be precise it is 
$t=\tau \ell^2 \beta^{-\hal } (N^{\frac1\d}\ell)^{-1-\frac{\d}{2}}$ for  fixed $\tau$
 and is chosen to obtain a finite variance in the limit (but not necessarily a bounded mean), this is what yields the factor in front of the fluctuation in \eqref{7230} and \eqref{7240}. 

The evaluation of the fixed term $\exp\( - \beta N^{1+\frac{2}{\d}} (\mathcal{E}^{V_t}_\theta (\mutt) - \mathcal{E}^V_\theta (\mut) )  \)$ above is not difficult and is   done in Lemma \ref{lemleadord}, and the main work is 
 to evaluate the ratio  $\frac{\K_N(\mutt)}{\K_N(\mut)}$. A first difficulty is that,  while it is easy to describe the perturbed  {\it usual} equilibrium measure in the interior case (it is just $\muv +t \cd^{-1}  \Delta \xi$, see \cite{serser} for the more delicate boundary case), describing the  
perturbed {\it thermal} equilibrium measure $\mutt$ exactly is more difficult and has not yet been done in the literature. Instead we replace $\mutt$ by two successive  good approximations $\nu_\theta^t$ and $\tilde \mutt$ described in Section \ref{sec5}. This induces an error which can be evaluated once one knows good first  bounds  on  $\log \K_N(\mu_1)- \log \K_N(\mu_0)$ for  general probability densities $\mu_0$ and $\mu_1$, see Lemma \ref{lemcompdesk}. 

Our method here is the transport-based approach  of \cite{ls2}, and 
the approximation $\tilde \mutt$ is chosen because it is expressed as a simple transport of $\mut$, in the form $(\id + t\psi)\# \mut$ (here $\# $ denotes the push-forward of probability measures) where $\psi$ is an explicit transport map. The map $\psi$ is itself a (truncated) series in inverse powers of $\theta^{-1}$. The number of terms kept in the series, or level of approximation, is the parameter $q$ in our results. It can be chosen at will, the larger the $q$ the more precise the approximation (especially when $\theta$ is not tending to $\infty$ very fast) but the more regularity of $\xi$ and $V$ it requires.

 We are then left with evaluating the change of $\log \K_N(\mu)$ along a  transport. But if $\mu$ and $\Phi\#\mu$ are two probability densities, by definition we have\begin{multline}\label{formul}
\frac{\K_N(\Phi\#\mu)}{\K_N(\mu)}=
\frac{1}{\K_N(\mu)} \int_{(\R^\d)^N} \exp\(-\beta N^{\frac2\d-1}\F_N(\XN, \Phi\#\mu)\) d(\Phi\#\mu)^{\otimes N}(X_N) 
\\= \frac{1}{\K_N(\mu)}\int_{(\R^\d)^N} \exp\(-\beta N^{\frac2\d-1}\F_N( \Phi(\XN), \Phi\#\mu) \) d\mu^{\otimes N}(X_N) \\
= \Esp_{\mathsf{Q}_N(\mu)} \(  \exp\(-\beta N^{\frac2\d-1}(\F_N(\Phi(\XN),\Phi\#\mu)-\F_N(\XN, \mu)\)\)\end{multline} 
\cor{with $\mathsf{Q}_N$ the Gibbs measure defined just after \eqref{pdef}}. 
Thus we just need to evaluate  the  variation of the energy $\F_N$ along a  transport. Note that here it  is particularly convenient that we have an integral against $\mu^{\otimes N}$ instead of the Lebesgue measure, thanks to the use of the thermal equilibrium measure. This makes the formula \eqref{formul} exact, with no Jacobian term contrarily to \cite{ls2}.

We thus work at evaluating  the variation of $\F_N (\Phi_t(\XN), \mu_t)$ along a transport $\Phi_t=\id + t\psi$, with $\mu_t= \Phi_t\# \mu$ for a generic probability density $\mu$, when $t$ is small enough. This is done in  Proposition \ref{prop:comparaison2}.  
The result is that the first and second derivatives in $t$  of the energy  $\F_N(\Phi_t(\XN), \mu_t)$ are both bounded by $C\F_N(\XN, \mu)$, i.e. the energy itself. This extends the result of \cite{ls2} to higher dimension and is an improvement even in dimension 2 since in \cite{ls2}  only the first derivative was fully controlled, and this turns out crucial later.
The proof  relies in an essential way on the electric formulation of $\F_N$ (see Section \ref{secelec}) first introduced in \cite{ss1,rs} and on some new technical energy control estimates, proven in Section \ref{sec33}.   
The first derivative in $t$ of $\F_N (\Phi_t(\XN), \mu_t)$   involves a singular integral term, which we had called ``anisotropy" in dimension 2 in \cite{ls2}, but is even more singular thus harder to handle in higher dimension. We show here how to give it a meaning via the electric formulation, effectively describing how to ``renormalize the loop equations".

Thanks to this control we can deduce the bound (of  Lemma \ref{lemcompdesk}) on differences of the form
$\log \K_N(\mu_1)- \log \K_N(\mu_0)$, by  interpreting $\mu_1$  as a transport of $\mu_0$. 
This bound suffices to obtain the fluctuation bound in Theorem \ref{firsttheo} and also to control the errors made when replacing $\mutt$ by its approximations above.
It does not however suffice to evaluate the Laplace transform in \eqref{laplace0} with sufficient precision for 
the CLT. 
For that, we use the approach of  \cite{ls2} of comparing two different ways of evaluating $\log \frac{\K_N(\mu_1)}{\K_N(\mu_0)}$: one  via the linearization of $\F_N$ just described above, and one by evaluating independently  $\log \K_N(\mu)$ for a general nonuniform $\mu$. This consists in proving  the free  energy expansion  with a rate, Theorem~\ref{thglob} and more importantly, its localized version Proposition  \ref{th1b}. To do so, one splits the support of $\mu$ into mesoscopic cubes in which $\mu$ is almost uniform, and adds up  the free energies for uniform measures in cubes  obtained in \eqref{1.26}  via the   almost additivity of the free energy proved in \cite{as} (which comes with an additivity error rate). To do so,  we use the control of Lemma \ref{lemcompdesk} to bound the error made when replacing a varying measure with a uniform one in  a small cube.   We also need the assumption \eqref{assfs}    in dimensions $\d \ge 3$  to obtain a good enough error rate because in those dimensions and contrarily to dimension $2$, the free energy dependence in $\mu$ involves a  dependence  inside the function $\fd$, see \eqref{expvar}.

Comparing  these two ways of evaluating $\log  \K_N(\mu_t)$ along a transport and using the good control on the second derivative of this quantity, we are able to obtain an improved estimate on its first derivative, this is the idea borrowed from \cite{ls2} in dimension 2.  
Applying to the thermal equilibrium measure, the first derivative  in $t$ of $\log \K_N((\id + t \psi)\# \mu_\theta)$ gives the mean of the fluctuation variable (which may be unbounded), while its higher derivatives do not contribute. The variance in the end only comes from the constant exponential term in the right-hand side of \eqref{laplace0}. Assembling these elements provides the convergence of the log Laplace transform of the fluctuation, after subtracting the appropriate mean, to an explicit quadratic function, as desired.

\subsection{Plan of the paper}
In Section \ref{sec3} we review the electric formulation of the energy and the associated definitions, we then provide a new  multiscale interaction energy control, Proposition \ref{multiscale}. 
We conclude the section by reviewing the local laws and almost additivity from \cite{as}.

In Section \ref{sec4} we show how to control the variations of the energy along a transport. The main result  there is Proposition \ref{prop:comparaison2}.
 This is then applied to estimate the difference of free energies when perturbing the  background  measure.

In Section \ref{sec5} we choose a specific transport map adapted to the varying thermal equilibrium measure. We then combine the previous elements to provide a first bound on the fluctuations, proving  Theorem \ref{firsttheo} and Corollary \ref{corod2}. 

In Section \ref{varying} we prove Proposition \ref{th1b} and Theorem \ref{thglob}  by   the almost additivity of the free energy.  

In Section \ref{sec7} we prove the main CLT results of Theorems \ref{th2}, \ref{thlowt2}, \ref{th5} and \ref{thlowt3} and their corollaries. 
\\

\noindent
{\bf Acknowledgements:} 
This research was supported by NSF grant  DMS-1700278 and the Simons Investigator program. The author thanks Thomas Lebl\'e and Gaultier Lambert for helpful comments, and Alice Guionnet and Karol Kozlowski  for their careful reading.

\section{Preliminaries}\label{sec3}


\subsection{Electric formulation}\label{secelec}
We first  describe how to reexpress $\F_N(\XN ,\mu)$ in ``electric form", i.e via the electric (or Coulomb) potential generated by the points. This idea originates in \cite{ss1,rs,ps} but we use here the precise formulation of \cite{ls2}.
Here, contrarily to \cite{as} we are working at the normal scale, and not at the blown-up scale.

We consider the electrostatic potential $h$ created by the configuration $\XN$ and the background probability $\mu$, defined by 
\begin{equation}
\label{def:hnmu} h(x)= \int_{\R^\d} \g(x-y) d\(\sum_{i=1}^N \delta_{x_i} -N \mu\)(y),\end{equation} which we will sometimes later denote $h^\mu[X_N](x)$ for less ambiguity.
Since $\g$ is (up to the constant $\cd$), the fundamental solution to Laplace's equation in dimension $\d$, we have 
\be -\Delta h= \cd\( \sum_{i=1}^N \delta_{x_i} -N \mu\).\ee

We note that $h$ tends to $0$ at infinity because $\int \mu=1$ and the system formed by the positive charges at $x_i$ and the negative background charge $N\mu$ is neutral.
We would like to formally rewrite $\F_N(\XN,\mu)$  defined in \eqref{defF} as $\int |\nab h|^2$, however this is not correct due to the singularities of $h$ at the points $x_i$ which make the integral diverge. 
This is why we use a truncation procedure which allows to give a renormalized meaning to this integral.

We will need to consider configurations with number of points $\mn$ not necessarily equal to $N$.
\cor{Turning to the truncation procedure, by abuse of notation we will extend the definition of  $\g(\eta)$ to $\eta$ positive real numbers, by setting 
$\g(\eta)= - \log |\eta|$ if $\d=2$ and $\g(\eta)= \eta^{2-\d}$ if $\d \ge 3$, cf. \eqref{wlog2d}.}
For any number $\eta>0$, we then let
\be\label{def:truncation} \f_{\eta} (x)= (\g(x)-\g(\eta))_+,\ee  where $(\cdot)_+$ denotes the positive part of a number,
and point out that $\f_\eta$ is supported in $B(0,\eta)$. We will also use the notation 
\be \label{defgeta}\g_{\eta} = \g-\f_{\eta} = \min (\g, \g(\eta) ).\ee
This is a truncation of the Coulomb kernel.
We also denote by $\delta_{x}^{(\eta)}$ the uniform measure of mass $1$ supported on $\partial B(x, \eta)$. This is a smearing of the Dirac mass at $x$ on the sphere of radius $\eta$. \cor{Since $\g$ is harmonic away from the origin, by the mean-value formula, $\g_\eta$ and $\g * \delta_0^{(\eta)}$ coincide outside of $B(0,\eta)$, moreover they also have the same Laplacian $-\cd \delta_0^{(\eta)}$   by symmetry and mass considerations, therefore $\g* \delta_0^{(\eta)}= \g_\eta$ everywhere and  }
\be \label{fconv} \f_{\eta}=\g*\( \delta_0-\delta_0^{(\eta)}\)\ee
so that 
\be\label{deltaf}-\Delta \f_{\eta}= \cd\(  \delta_0-\delta_0^{(\eta)}\).\ee
We also note that 
\be\label{intf}
\int_{\R^\d} |\f_{\eta}|\le C \eta^2,\quad
\int_{\R^\d} |\nab \f_{\eta}|\le C \eta.\ee

 For any $\vec{\eta}=(\eta_1, \dots, \eta_\mn)\in \cor{\R_+^\mn}$, and any  function $u$ satisfying a relation of the form 
\begin{equation}\label{formu}
-\Delta u =  \cd\(\sum_{i=1}^\mn \delta_{x_i}- N \mu\)
\end{equation}
we then define  the truncated potential
\begin{equation}\label{formu2}
u_{\vec{\eta}}= u- \sum_{i=1}^\mn \f_{\eta_i}(x-x_i).\end{equation}
  We note that  in view of \eqref{deltaf} the function $u_{\vec{\eta}}$ then satisfies 
  \be -\Delta u_{\vec{\eta}} = \cd\( \sum_{i=1}^\mn \delta_{x_i}^{(\eta_i)}- N \mu\).\ee

We then define a particular choice of truncation parameters: if $X_\mn= (x_1, \dots, x_\mn)$ is an $\mn$-tuple of distinct points in $\R^\d$ we denote for all $i=1, \dots, \mn$,
\begin{equation}\label{def:trxi}
\rr_i= \frac{1}{4} \min\left(\min_{j \neq i} |x_i-x_j|, N^{-\frac1\d} \right)
\end{equation}
which we will think of as the \textit{nearest-neighbor distance} for $x_i$.

The following is proven in \cite[Prop. 2.3]{ls2} and \cite[Prop 3.3]{smf}. It gives a renormalized meaning to the ``electric reformulation" of $\F_N(\XN, \mu)$ as 
$\frac{1}{2\cd} \int |\nab h|^2$.
\begin{lem} \label{lem:monoto} 
Let $\XN$ be in $(\R^\d)^N$ and $\mu$ be a probability measure with bounded density. If $ (\eta_1, \dots, \eta_N)$ is such that $0 < \eta_i \le \rr_i$ for each $i = 1, \dots, N$, we have
 \begin{equation}
\label{fnmeta}
 \F_N(\XN,\mu) = \frac{1}{2\cd} \left(\int_{\R^\d}|\nab h_{ \vec{\eta}}|^2  -\cd \sum_{i=1}^N  \g(\eta_i)  \right)  
 -N \sum_{i=1}^N \int_{\R^{\d}} \f_{\eta_i}(x - x_i) d\mu(x).
\end{equation}
 \end{lem}
 This shows in particular that the expression in the right-hand side is independent of the truncation parameter, as soon as the latter  is small enough.  Choosing for instance $\eta_i=\rr_i$ this provides an exact electric representation for $\F$.

We next   present a Neumann local version of the energy  first introduced in \cite{as} 
:  consider $U$ a subset of $\R^\d$ with piecewise $C^1$ boundary, bounded or unbounded (here we will mostly use hyperrectangles and their complements), $\Omega$ a subset of $U$ (typically a subcube or ball),  and  introduce a modified version of the minimal distance
\begin{equation}\label{rrh}
 \rrc_i := \frac14
\left\{ 
\begin{aligned}
& \min \(\min_{x_j \in \Omega, j\neq i} |x_i-x_j|, \dist(x_i, \partial U \cap \Omega)\)\quad && \text{if} \ \dist(x_i, \partial \Omega \backslash \partial U ) \ge \hal N^{-\frac1\d},\\
& \min\(N^{-\frac1\d}, \dist (x_i, \pa U\cap \Omega) \)&& \text{otherwise.}
\end{aligned}
\right. 
\ee
This ensures that the balls $B(x_i, \rrc_i)$ remain included in $U$.
If $N\mu (U)=\mn$ is an integer,  
 for a configuration $X_{\mn}$ of points in $U$,  and using the notation $\rrc$ for the vector $(\rrc_1, \dots, \rrc_\mn)$, 
 we define
\begin{multline}\label{Glocal}
\F^{\Omega}_N(X_\mn,\mu, U)=\frac{1}{2\cd}\( \int_{\Omega} |\nab v_{\rrc}|^2 - \cd \sum_{i, x_i \in \Omega} \g(\rrc_i) \) - N
\sum_{i, x_i\in \Omega}\int_U \f_{\rrc_i}(x - x_i) d\mu(x)\\
+ \sum_{i, x_i \in \Omega}\(\g(\frac14 \dist(x_i, \pa U))- \g(\frac{N^{-\frac1\d}}{4})\)_+ ,\end{multline}
 where 
\begin{equation}\label{defv}
\left\{\begin{array}{ll}
 -\Delta v = \cd\Big( \sum_{i=1}^{\mn} \delta_{x_i}-N \mu \Big) &\ \text{in} \ U \\
 \frac{\pa v}{\pa \nu}=0 &\ \text{on} \ \partial U, \end{array}\right.
\end{equation}
with $\pa/\pa \nu$ denoting the normal derivative.
Note that  under the condition $ N \mu(U)=\mn $  the solution of \eqref{defv} exists and is unique up to addition of a constant.

 The extra additive term in the second line of \eqref{Glocal} was needed in \cite{as} to control points getting  close to the boundary when proving local laws.

Finally, we write $\F_N(X_\mn, \mu, U)$ for $F_N^{U} (X_\mn, \mu, U)$.

\subsection{Monotonicity and local energy controls}\label{sec33}

We need the following result inspired from \cite{ps,ls2} which expresses a monotonicity with respect to the truncation parameter, and allows to deduce a new control of  the interaction energy at arbitrary scales $\alpha$.
\begin{lem}\label{monoto}
Let  $U$ be any open set and $u$ solve 
\be\label{eqsu}
-\Delta u= \cd\(\sum_{i=1}^\mn \delta_{x_i} - N\mu \) \quad \text{in} \ U,\ee  and let $u_{\vec{\alpha}}, u_{\vec{\eta}}$ be as in~\eqref{formu2}. Assume 
$\alpha_i \le \eta_i$ for each $i$.  Letting $I$ denote $\{i, \alpha_i\neq \eta_i\}$, assume that 
for each $i \in I$   we have  $B(x_i ,\eta_i) \subset U$.
Then 
\begin{multline}
\label{premono}
\int_U|\nab u_{\vec{\eta}}|^2 - \cd \sum_{i=1}^\mn \g(\eta_i) -2N\cd \sum_{i=1}^\mn \int_U \f_{\eta_i}(x-x_i)d\mu\\- \( \int_U |\nab u_{\vec{\alpha}}|^2 - \cd \sum_{i=1}^{\mn} \g(\alpha_i) -2N\cd \sum_{i=1}^\mn \int_U \f_{\alpha_i}(x-x_i)d\mu(x)\) \le 0,\end{multline}
with equality if $\eta_i \le \rr_i$ for each $i$. Moreover, for $\Omega \subset U$, denoting
 temporarily 
 \be \label{defFa}\mathcal F^{\vec{\alpha}}:=\frac1{2\cd}\(\int_{\Omega} |\nab u_{\vec{\alpha}}|^2  -\cd\sum_{i,x_i\in \Omega}  \g(\alpha_i)-2 N\cd\sum_{i, x_i \in \Omega}\int_U \f_{\alpha_i}(x-x_i)d\mu(x)\)
\ee assuming that 
\cord{\be \label{conddistancebord}
\alpha_i=   \frac14 N^{-\frac1\d} \text{ for $x_i$ such that $\dist(x_i, \pa \Omega) \le \hal N^{-\frac1\d} $}\ee}
and $\alpha_i\le \rrc_i$ if $\dist(x_i ,\partial \Omega) \le \alpha_i$,  
 we have
\be\label{supplem}\hal
\sum_{\substack{ i\neq j,x_i , x_j \in \Omega\\ \dist(x_i, \pa \Omega) \ge \alpha_i + \frac14 N^{-1/\d} }} 
\( \g(x_i-x_j) - \g(\alpha_i)\)_+ \le  \F_N^\Omega(X_\mn, \mu, U) -  \mathcal F^{\vec{\alpha}} ,\ee
and 
\be\label{controlmic3} \begin{cases}
\displaystyle\sum_{\substack{ i\neq j, x_i , x_j \in \Omega,  \dist(x_i, \pa \Omega) \ge 4\alpha,\\ \alpha \le  |x_i-x_j|\le 2 \alpha}}\g(\alpha) \le C \(  \mathcal F^{\vec{\eta}} - \mathcal F^{\vec{\eta}' }\) & \text{if }  \d \ge 3\\
\displaystyle \sum_{\substack{ i\neq j, x_i , x_j \in \Omega, \dist(x_i, \pa \Omega) \ge 4\alpha ,\\ \alpha \le  |x_i-x_j|\le 2 \alpha}} 1  \le C \(  \mathcal F^{\vec{\eta}} - \mathcal F^{ \vec{\eta}'}\) & \text{if }  \d=2,\end{cases}\ee
where  $\vec{\eta}$ is set to be $\alpha $ if $\dist (x_i, \pa \Omega) \ge  \alpha$ and $\rrc_i$ otherwise, and $\vec{\eta}'$ is set to be $4\alpha$ if $\dist (x_i, \pa \Omega) \ge 4\alpha$ and $\rrc_i$ otherwise, and $C>0$ depends only on $\d$.
\end{lem}
\begin{proof}
The relation \eqref{premono} is proven for instance in 
  \cite[Proof of Lemma B.1]{as}. There it is also shown that \cor{if  $\alpha_i \le \eta_i$ for each $i$, } $\g_\eta$ being as in \eqref{defgeta}, we have
\be \label{328}\hal
 \sum_{x_i , x_j \in \Omega, i\neq j}\( \g_{\alpha_i}(|x_i-x_j|+\alpha_j) -\g(\eta_i)\)_+
\le \mathcal F^{\vec\alpha}  - \mathcal F^{\vec\eta}\ee

 Letting  $\alpha_i \to 0$  for the points $x_i$ such that $\dist(x_i, \pa \Omega)\ge \eta_i$   while choosing $\alpha_i=\eta_i$ for the others,  we find that 
\be  \hal \sum_{x_i , x_j \in \Omega,i\neq j, \dist(x_i, \pa \Omega) \ge \eta_i
}\( \g(|x_i-x_j|) -\g(\eta_i)\)_+\le \F_N^{\Omega}(X_\mn, \mu)- \mathcal F^{\veta} ,\ee which gives the result \eqref{supplem} by substituting $\eta_i$ by $\alpha_i$. 
Here we observed that for $\vec{\alpha}$ such that $\alpha_i\le \rrc_i$, we have $\mathcal F^{\vec{\alpha}}= \F^\Omega_N(X_\mn,\mu, U)$.

Next, applying \eqref{328} to $\vec{\eta}$ and $\vec{\eta}'$, we find the results \eqref{controlmic3}.\end{proof}

The following result shows that despite the cancellations occurring  between the two possibly very large terms $ \int_{\R^\d}|\nab u_{\vec{\eta}}|^2$ and $\cd \sum_{i=1}^N  \g(\eta_i)   $, when choosing $\eta_i=\rrc_i$ we may control each of these two terms by the energy i.e. by their difference.
It is adapted from \cite[Lemma B.2]{as}.

\begin{lem}
 \label{lem:contrdist1} For any configuration $X_\mn$ in $U$,  and $v$ corresponding via \eqref{defv}, letting $\# I_\Omega$ denote $  \#\(\{X_\mn \} \cap \Omega \)$ \cor{where $\{X_{\mn}\}$ is the set of points formed by the entries of $X_\mn$} and $\# $ denotes the cardinality,
 for any $\Omega \subset U$, and any $\vec{\eta}$ such that 
$\eta_i \in [\frac14 \rrc_i, \rrc_i]$, with $\rrc$ computed with respect to $ \Omega$ as in \eqref{rrh},  \cord{and satisfying condition \eqref{conddistancebord},} we have
\begin{equation}\label{11}
\sum_{x_i \in \Omega} \g(\eta_i)  \le 2\(  \( \F_N^{\Omega}(X_\mn,\mu, U)+ \frac{\#I_\Omega }{2} (\log N) \indic_{\d=2}\) +C_0 \#I_\Omega
N^{1-\frac2\d}\) ,\end{equation}
\be\label{14}
 \int_{\Omega} |\nab v_{\vec{\eta}}|^2 \le 4\cd  \( \( \F_N^{\Omega} (X_\mn, \mu, U)+ \frac{\#I_\Omega }4(\log N) \indic_{\d=2}\) + C_0\#I_\Omega 
N^{1-\frac2\d} \)
 \ee
with $C_0>0$ depending only on an upper bound for $\mu$ in $\Omega$. 
\cord{\footnote{The reader should notice the different factors $1/2$ and $1/4$ in front of $\log N$ in \eqref{11} and \eqref{14}} 
Moreover, for any $\vec{\eta}$ such that $\eta_i\le \rrc_i$  \cord{and satisfying condition \eqref{conddistancebord},} we have 
\be\label{15}\int_{\Omega} |\nab v_{\vec{\eta}}|^2 
 \le  2\cd   \F_N^{\Omega}(X_\mn,\mu, U) +  \cd \sum_{x_i\in \Omega} \g(\eta_i).
\ee
}
 \end{lem}
  
\begin{proof} We prove the result for $\eta_i=\rrc_i$, the general case is a straightforward adaptation.
For every $1\leq i\leq N$, let us choose $\alpha_i=\alpha= \frac14 N^{-1/\d}$. Applying \eqref{supplem}, we have
\begin{multline}\label{236}
\F_N^\Omega(X_\mn, \mu) \geq -  \frac{1}{2} \sum_{i,x_i\in \Omega}  \g(\alpha )- N \|\mu\|_{L^\infty}\|\f_{\alpha}\|_{L^1}\# I_\Omega \\
+ \frac{1}{ 2 }\sum_{\substack{i,j , x_i,x_j \in \Omega\\ \dist(x_i, \pa\Omega) \ge\alpha}} (\g\(|x_i-x_j|) - \g(\alpha)\)_+.
\end{multline}
From the definition of $\rrc_i$, we see that if 
$\dist (x_i, \partial \Omega) \ge \frac14 N^{-1/\d}$, there exists $x_j \in \Omega$ such that $4\rrc_i= \min(\min_{j\neq i}  |x_i-x_j|,  N^{-1/\d}) $
 so that in all cases
\begin{equation}
\(\g(x_i-x_j) - \g(\alpha)\)_+  \geq \g(4\rrc_i)-\g(\alpha).
\end{equation}
In view of \eqref{intf}, it follows that, if $\d\neq 2$,
\begin{equation}
\sum_{\substack{i,x_i \in \Omega\\ \dist(x_i, \pa\Omega) \ge \alpha}} \g(\rrc_i)\le  C\Bigg(\F_N^\Omega(X_\mn, \mu) + \#I_\Omega\g(\alpha) + CN \|\mu\|_{L^\infty}\#I_\Omega  \alpha^{2}\Bigg),
\end{equation}
with $C$ depending only on $\d$. Now in view of our choice of $\alpha$ and the definition of $\rrc_i$, if $x_i\in\Omega$ with $\dist(x_i,\pa\Omega)< \alpha$, then $\rrc_i = \alpha$. Hence,
\begin{equation*}
\sum_{i ,x_i\in\Omega} \g(\rrc_i) \leq  C\Bigg(\F_N^\Omega(X_\mn, \mu) + \#I_\Omega\g(\alpha) + CN \|\mu\|_{L^\infty}\#I_\Omega \alpha^{2}\Bigg) + \#I_{\Omega}\g(\alpha).
\end{equation*}
Inserting the definition of $\alpha$ into this inequality, we conclude that \eqref{11} holds if $\d\neq 2$. If $\d=2$, we start again from \eqref{236} and using the same reasoning, we get instead
\begin{equation*}
 \sum_{i,x_i \in \Omega} \g(\rrc_i/\alpha)\le  2\Bigg(\F_N^\Omega(X_\mn, \mu) +  \#I_\Omega\g(\alpha) + CN\|\mu\|_{L^\infty}\#I_\Omega  \alpha^{2}\Bigg),
\end{equation*}
and the conclusion follows as well.

We next turn to \eqref{14}.
Let us next choose $\alpha_i=\rrc_i$ in \eqref{supplem} where we replace the left-hand side by $0$. Using that $\rrc_i \le \frac14N^{-1/\d}$, we deduce, using again \eqref{intf},
\begin{equation*}
\F_N^\Omega(X_\mn, \mu)
\ge \frac1{2\cd}\(\int_{\Omega}  |\nabla v_{\rrc}|^2  -  \cd \sum_{i,x_i\in \Omega}  \g(\rrc_i) \)-C  \#I_\Omega  N^{1-\frac2{\d}},
\end{equation*}
and in view of \eqref{11}, \eqref{14} follows. In the case $\d=2$ we split $\g(\rrc_i)$ into $\g( 4 \rrc_i N^{1/\d} ) + \g( N^{-1/\d}/4 )$, and then apply \eqref{11}. \cord{Finally \eqref{15} follows from \eqref{supplem} applied to $\alpha_i= \eta_i$.}

\end{proof}



Specializing the relation \eqref{supplem} to $\alpha_i= \rrc_i  $ if $\dist(x_i, \pa \Omega) <2 N^{-\frac1\d}$ and 
$\alpha_i=2 N^{-\frac1\d}$ if $\dist(x_i, \pa \Omega) \ge 2 N^{-\frac1\d}$, bounding from below $\mathcal F^{\vec{\alpha}}$ in an obvious way from \eqref{defFa} and \eqref{intf},  we deduce the following control of short-range interactions
\begin{coro}\label{coro3} Under the same assumptions, we have
\begin{equation}
\begin{cases}\displaystyle
\sum_{\substack{ i\neq j, x_i , x_j \in \Omega,\\ \dist(x_i, \pa \Omega) \ge 3N^{-\frac1\d},\\
 |x_i-x_j|\le N^{-\frac1\d} }}\g(|x_i-x_j|) \le C\(   \F_N^{\Omega}(X_\mn, \mu, U )+ C_0\#I_\Omega N^{1-\frac2\d}\) & \text{if } \d \ge 3\\
\displaystyle\sum_{\substack{ i\neq j, x_i , x_j \in \Omega,  \\\dist(x_i, \pa \Omega) \ge 3 N^{-\frac1\d},\\  |x_i-x_j|\le N^{-\frac1\d} }}\g(2 |x_i-x_j|N^{\frac1\d} ) \le C \(  \F_N^{\Omega}(X_\mn, \mu, U) + \frac{ \#I_\Omega }{4}\log N + C_0 \#I_\Omega \) & \text{if }  \d=2.\end{cases}\ee
\end{coro}

We now present a novel application of the mesoscopic interaction energy control of \eqref{supplem} and \eqref{controlmic3} which allows, by combining the estimates obtained over dyadic scales to control general inverse powers of the distances between the points. It is to be combined with Corollary~\ref{coro3} to estimate the interaction of microscopically close points.

\begin{prop}[Multiscale interaction energy  control]\label{multiscale} Let $s>0$ and $N^{-\frac1\d}\le \ell \le 1$. We have
\cor{\begin{multline}\label{resmultiscale}\sum_{\substack{i\neq j, x_i, x_j \in \Omega\\
 N^{-1/\d}\le |x_i-x_j|\le \ell,  \dist(x_i, \pa \Omega) \ge4\ell}} \frac{1}{|x_i-x_j|^{\d-2+s}}
 \\  \le CN^{\frac{s}{\d}}\( \F^{\Omega}_N(X_{\mn}, \mu, U)+\frac{1}{4}(  \#I_\Omega  \log N) \indic_{\d=2}\) + C \#I_\Omega N^{1-\frac2\d + \frac{s}{\d} } \\
 +  \begin{cases}  C \#I_\Omega   N \ell^{2-s}
 & \text{if} \ s\neq 2\\
  C \#I_\Omega N \log (\ell N^{\frac1\d})  
  & \text{if}  \ s=2, \end{cases} \end{multline}}
where $C>0$ depends only on an upper bound for $\mu$ and on $\d$.
\end{prop}
\begin{proof} Let us for the sake of generality start from any  function $f$ such that $f(x)/\g(|x|)$ is a positive decreasing function of $\R$  if $\d\ge 3$, respectively $f$ a positive decreasing function of $\R$ if $\d=2$.
Decomposing over dyadic scales $\le \ell$, denoting 
$$K:=\left[  \frac{\log (\ell N^{\frac1\d})}{\log 2}\right],$$ \cor{with $[x]$ the smallest integer $\ge x$.}
We have 
\cor{\begin{align*}
\sum_{\substack{i\neq j,  x_i, x_j \in \Omega,\\ N^{-1/\d}\le |x_i-x_j|\le \ell, \\ \dist(x_i, \pa \Omega) \ge 4\ell}}  f(|x_i-x_j|) &\le \sum_{k=0}^{K-1} \sum_{\substack{i\neq j, 2^k N^{-1/\d}\le |x_i-x_j|\le 2^{k+1} N^{-1/\d},\\  \dist(x_i, \pa \Omega) \ge 4\ell}}
 f(|x_i-x_j|)  \\ & \le  \sum_{k=0}^{K-1} \sum_{\substack{i\neq j, 2^k N^{-1/\d}\le |x_i-x_j|\le 2^{k+1} N^{-1/\d}\\        \dist (x_i, \pa \Omega) \ge 4 \cdot 2^k N^{-1/\d} }  } f(2^k N^{-\frac1\d}) \\
 & \le  \sum_{k=0}^{K-1}   \frac{ f(2^k N^{-\frac1\d})}{\g(2^k N^{-\frac1\d}) }  \sum_{\substack{ i\neq j, 2^k N^{-1/\d}\le |x_i-x_j|
 \le 2^{k+1} N^{-1/\d}\\     \dist (x_i, \pa \Omega) \ge 4 \cdot 2^k N^{-1/\d}  } } \g(2^k N^{-\frac1\d}),  \end{align*}}
 where we use that $\dist (x_i, \pa \Omega) \ge 4 \cdot 2^k N^{-1/\d}$
  and the last line follows from the assumption that $f/\g$ is nonincreasing and $\g$ nonincreasing.
 Inserting \eqref{controlmic3}, we deduce 
 $$\sum_{\substack{i\neq j,  x_i, x_j \in \Omega, N^{-1/\d}\le |x_i-x_j|\le \ell, \\ \dist(x_i, \pa \Omega) \ge 4 \ell}}   f(|x_i-x_j|)
  \le  C \sum_{k=0}^{K}   \frac{ f(2^k N^{-\frac1\d})}{\g(2^k N^{-\frac1\d}) } \( \mathcal F^{\vec{\alpha}^k } - \mathcal F^{\vec{\alpha}^{k+2}}  \)
$$
  with for each $k$,  $\alpha^k_i= 2^{k} N^{-\frac1\d}$ if $\dist(x_i, \pa \Omega)\ge 2^{k} N^{-\frac1\d}$ and $\rrc_i$ otherwise.
    
 Using Abel's resummation procedure  we find 

 \begin{align*}
 \sum_{i\neq j, N^{-\frac1\d}\le |x_i-x_j|\le \ell, \dist(x_i, \pa \Omega) \ge  4\ell} &f(|x_i-x_j|)\\& \le 
 \sum_{k=0}^{K}   \frac{ f(2^k N^{-\frac1\d})}{\g(2^k N^{-\frac1\d}) }  \mathcal F^{\vec{\alpha}^k }
 - \sum_{k=2}^{K+2}    \frac{ f(2^{k-2} N^{-\frac1\d})}{\g(2^{k-2} N^{-\frac1\d}) }  \mathcal F^{\vec{\alpha}^k }
 \\ & \le
 \sum_{k=2}^{K}   \( \frac{ f(2^k N^{-\frac1\d})}{\g(2^k N^{-\frac1\d}) }-  \frac{ f(2^{k-2} N^{-\frac1\d})}{\g(2^{k-2} N^{-\frac1\d}) } \)  \mathcal F^{\vec{\alpha}^k }
    \\ &+ \frac{f( 2 N^{-1/\d})} {\g(2    N^{-1/\d})}  \mathcal F^{\vec{\alpha}^1 }+ \frac{f( N^{-1/\d})} {\g(   N^{-1/\d})} \mathcal F^{\vec{\alpha}^0 }
   \cor{ -     \frac{ f(2\ell)}{\g(2 \ell ) }  \mathcal F^{\vec{\alpha}^{K+ 2  } }- \frac{f( \ell )}{\g( \ell)} \mathcal F^{\vec{\alpha}^{K+ 1 }}.}
 \end{align*}
 We next use the decreasing nature of $\mathcal F^{\vec{\alpha}}$ with respect to $\alpha$ of \eqref{premono} (applied to $U= \Omega$), hence that of $\mathcal F^{\vec{\alpha}^k}$ with respect to $k$. This monotonicity also allows to bound from above each $\mathcal F^{\vec{\alpha}^k}$ by 
 $\F_N^{\Omega}(X_\mn, \mu, U)$ and from below (by definition and by \eqref{intf}) as follows
 \begin{align}\label{bbelo}\mathcal F^{\vec{\alpha}^k} \ge - \hal  \sum_i \g(\alpha_i^k) -  C N  \sum_i ( \alpha_i^k )^2 & \ge - \hal  \sum_i \g(\rrc_i)
 -  C N  \sum_i ( \alpha_i^k )^2
 \\  \nonumber & \ge - C\F^{\Omega}_N(X_\mn,\mu, U) 
 - C \#I_\Omega N^{1-\frac2\d}-  C N  \sum_{i, x_i \in \Omega} ( \alpha_i^k )^2  
 \\ 
 \nonumber
 &\cor{\ge - C\F^{\Omega}_N(X_\mn,\mu, U) 
- C  \# I_\Omega N^{1-\frac2\d} 2^{2k}     }
 \end{align} 
 after using \eqref{11}.

 
 Inserting into the above and using the \cor{mean-value theorem and the} monotonicity of $f/\g$, 
 we obtain if $\d\ge 3$, 
 \cor{
  \begin{multline}\sum_{\substack{i\neq j,N^{-\frac1\d} \le |x_i-x_j|\le \ell ,\\ \dist(x_i, \partial \Omega ) \ge 4\ell} } f(|x_i-x_j|)\\ \le 
C
 \sum_{k=2}^{K} - \(\frac{f}{\g}\)'(2^{k-2} N^{-\frac1\d})     2^k N^{-\frac1\d}  |(\mathcal F^{\vec{\alpha}^k})_- | \\+
 2\frac{ f(N^{-\frac1\d}) }{\g(N^{-\frac1\d})}  \F^{\Omega}_N(X_{\mn},\mu, U)  
 + ( C \F_N^{\Omega} (X_{\mn}, \mu, U)+  C \#I_\Omega   N \ell^2   )  \frac{ f(\ell) }{\g( \ell)} .
 \end{multline}
 If $\d \ge 3$, specializing to  $f/\g=|x|^{-s}$ with $s>0$, and using \eqref{bbelo} we find
 \begin{multline}\sum_{i\neq j, N^{-1/\d}\le |x_i-x_j|\le \ell,  \dist(x_i, \pa \Omega) \ge4\ell} f(|x_i-x_j|) \\ \le C
\F^{\Omega}_N(X_\mn,\mu, U)    N^{\frac{s}\d} \sum_{k= 2}^{[\log  (\ell N^{1/\d})  /\log 2]} 2^{-ks}
+C \# I_\Omega N^{1-\frac2\d+\frac{s}\d}  \sum_{k= 2}^{[\log  (\ell N^{1/\d})  /\log 2]} 2^{k(2-s)}
\\
  + CN^{\frac{s}{\d}} \(\F^{\Omega}_N(X_{\mn},\mu, U)+ C \#I_\Omega N^{1-\frac2\d}\) + C \#I_\Omega  N\ell^{2-s}   \end{multline}
hence the result \eqref{resmultiscale}.
If $\d=2$, we replace the use of $f/\g$ by that of $f$ and the use of $\mathcal F^{\vec{\alpha}^k}$ by that of 
$\mathcal F^{\vec{\alpha}^k}+ \frac14 \#I_\Omega  \log N$, and obtain   the result in a similar way using again \eqref{controlmic3}.}
\end{proof}

\subsection{Partition functions and local laws}

We define the partition functions relative to the set $U$ as
\begin{equation}
\label{defK}
\K_N(U, \mu):= \int_{U^\mn} e^{- \beta N^{\frac2\d-1} \F_N(X_\mn,\mu,U) }\, d\mu^{\otimes \mn} (X_\mn)\end{equation}
under the constraint $\mn=N\mu(U)$.
We also let 
\be \label{defQ} \mathsf{Q}_N(U,\mu)= \frac{1}{\K_N(U,\mu)} e^{-\beta N^{\frac2\d-1} \F_N(X_\mn, \mu, U) } d\mu^{\otimes \mn} (X_\mn)\ee
be the associated Gibbs measure.


We note that $\K_N(\R^\d,\mu)$ coincides with  $\K_N(\mu)$ defined in \eqref{pdef} and $\mathsf{Q}_N(\R^\d, \mu)$ coincides with $\PNbeta$ in view of \eqref{splitting} and \eqref{pdef}. 
\cor{In all this sequel, if the set $U$ is not specified for the energy or the partition function, then what is meant is $U=\R^\d$.}

If  $U$  is  partitioned into $p$ disjoint sets $Q_i$,  $i\in [1,p]$ which are such that $N \mu(Q_i)=\mn_i$ with $\mn_i$ integer \cor{(in particular the $Q_i$'s must depend on $N$)} then it is shown in \cite{as} that
\begin{equation}\label{superad2}
\K_N(U,\mu) \ge    \frac{N! }{{\mn_1}! \dots {\mn}_p! }    \prod_{ i=1  }^p \K_N(Q_i,\mu)  ,\end{equation}
an easy consequence of the subadditivity of the energy $\F_N$.
The converse is much harder to prove and was obtained in \cite{as} using the ``screening procedure" as a way to control the additivity defect. The result from \cite{as} is 

 \begin{prop}[Almost additivity of the free energy]   \label{proadd} Assume that $\mu$   is a density bounded above and below by positive constant in $\Sigma$.
      Assume $\hat U$ is a subset of $ \Sigma$ at distance larger than $ d_0$ (as in \eqref{defd0}) from $\pa \Sigma$ 
       and is a disjoint union of $p$ hyperrectangles $Q_i$ such that $N\mu(Q_i)= \mn_i$ with $\mn_i$ integers,  
  of sidelengths \cor{ in $[R, 2R]$} satisfying 
\be 
\label{rrb} 
R N^{\frac1\d} \ge \rb+ \( \frac{1}{\beta\chi(\beta)} \log \frac{R^{\d-1}}{\rb^{\d-1}}\)^{\frac1\d}
\ee  
with $\rb$ as in \eqref{rhobeta},
and in addition, if $\d \ge 4$,  \be\label{rrb2}
RN^{\frac1\d}\ge \max(\beta^{\frac{1}{\d-2}-1},1) N^{\frac1\d} d^{-1}.\ee
Then there exists $C$, depending only on $\d$ and the upper and lower bounds for $\mu$ in $\Sigma$, such that 
\begin{align}
\label{subad3}
& \left| \log \K_N (\R^\d,\mu) 
- \left( \log \K_N (\R^\d\backslash \hat U,\mu)+ \sum_{i=1}^p \log \K_N (Q_i,\mu) \right) \right|
\\ & \qquad \notag
\leq C p N^{1-\frac1\d} \(       \beta R^{\d-1}     \rb   \chi(\beta) +\beta^{1-\frac1\d} \chi(\beta)^{1-\frac1\d} \(\log  \frac{R N^{\frac1\d}}{\rb}\)^{\frac1\d}   R^{\d-1}  \).
\end{align}
     If $U$ is a subset of $\Sigma$ equal to a disjoint union of $p$  hyperrectangles $Q_i$  with $N\mu(Q_i)=\mn_i $  integers, of sidelengths in $[R,2R]$ with $R \ge \rb$ satisfying \eqref{rrb},  then we have, with $C$ as above, 
     \begin{multline}
     \label{subad4} 
     \left| \log  \K_N(U,\mu) -
     \sum_{i=1}^p \log \K_N (Q_i,\mu) \right| \\ \leq C p N^{1-\frac1\d}\Bigg(    \beta R^{\d-1}      \chi(\beta) \rb  +
            \beta^{1-\frac1\d} \chi(\beta)^{1-\frac1\d} \(\log  \frac{R N^{\frac1\d}}{\rb}\)^{\frac1\d} R^{\d-1}  \Bigg).
     \end{multline}
      \end{prop}

Finally, we will need the following local laws from \cite{as} (here rescaled down to the original scale).

\begin{prop}[Local laws]\label{th3} Assume   $\mu$  is a density bounded above and below by positive constants in a set $\Sigma$ whose boundary is a disjoint union of  $C^1$ submanifolds. There exists a constant $C>0$ depending only on $\d$ and the upper and lower bounds for $\mu$  in $\Sigma$ such that the following holds.
Assume  $Q_\ell$ is a cube of sidelength $\ell \ge \rb N^{-1/\d}$,  with in addition
\be\label{conddist0}\dist(Q_\ell , \p \Sigma) \ge     d_0  
 \ee in the case $U\backslash \Sigma \neq \varnothing$.  
We have
\begin{enumerate}
\item (Control of energy) 
  \begin{multline}\label{locallawint0}
  \log \Esp_{\mathsf{Q}_N(U ,\mu) } \( \exp\( \hal  \beta\(  N^{\frac2\d-1}   \F^{Q_{\ell}}_N(\cdot,\mu, U )+ \(\frac{\mn}{4}\log N\)\indic_{\d=2}\)+C \#\(\{X_{\mn}\}\cap Q_{\ell}\) \)\)
  \\  \le 
C  \beta\chi(\beta ) N \ell^\d 
  \end{multline}  
  \item (Control of fluctuations) Letting $D$ denote $\int_{Q_\ell} \( \sum_{i=1}^N \delta_{x_i} - N d\mu\)$ we have 
    \begin{equation}\label{loclawpoints00}
   \left|\log \Esp_{\mathsf{Q}_N(U, \mu) }\( \exp\( \frac{\beta}{C}    \frac{D^2}{N^{1-\frac2\d}\ell^{\d-2} } \min (1, \frac{ |D|}{N\ell^\d})  \) \)\right|\le
  C  \beta\chi(\beta) N \ell^{\d}.\end{equation}
\item (Concentration for linear statistics) If $\varphi$ is a  Lipschitz function  such that $\|\nab \varphi\|_{L^\infty} \le N^{\frac1\d}$ supported in $Q_\ell $, we have 
\be\label{loclawphi}
\left|\log \Esp_{\mathsf{Q}_N(U,\mu)}\( \exp \frac{\beta }{C N \ell^\d}\(  \int_{\R^\d} \varphi\, d( \sum_{i=1}^N \delta_{x_i}-N\mu) \)^2      \)\right|
\le
  C  \beta \chi(\beta) N^{1-\frac2\d}\ell^{\d} \|\nab \varphi\|_{L^\infty}^2  .\ee
  \end{enumerate}
  
  \end{prop}
  When choosing $U=\R^\d$ we get the results for $\PNbeta$ since it coincides with $ \mathsf{Q}_N(\R^\d, \mu)$.

We have the following scaling relation about \eqref{defF}: if $\lambda>0$, letting $Y_\mn= \lambda^{\frac1\d}X_\mn$ and \cor{ $\mu'(x)= \frac{\mu(\lambda^{-1/\d}x)}{\lambda}$}
    \be \label{scalingdeF}
    \F_N(X_\mn, \mu,U)= \lambda^{1-\frac{2}{\d}} \F_N(Y_\mn, \mu', \lambda^{\frac1\d} U)- \( \frac{\mn}{4}\log \lambda\) \indic_{\d=2}\ee
    and 
    \be\label{scalingdeL}
      \K_N^{\beta}(U,\mu)= \K_N^{\beta \lambda^{1-\frac2\d}} (\lambda^{\frac1\d}U, \mu')  e^{ \beta \( \frac{\mn}{4}\log \lambda\) \indic_{\d=2}} ,\ee
    where we highlighted the $\beta$-dependence in a superscript.
    
    \section{Comparison of energies through transport}\label{sec4}
As described in Section \ref{secoutline}, a  major task is to evaluate the difference of energies along a transport, or rather expand it as the transport is close to identity, which is what we describe in this section. 

\subsection{Variations of energies along  a transport}
The first statement is a simple computation.
For $\alpha$ a multiindex, we denote  $|\alpha|=\alpha_1 +\dots+ \alpha_\d$ and  $D^\alpha:= \partial_1^{\alpha_1}\dots \partial_\d^{\alpha_\d}$.

\begin{lem}\label{proptransport2} 
Let $\mu $ be a probability density in $L^\infty(\R^\d)$ such that $\iint \g(x-y) d\mu(x) d\mu(y)<\infty$.
Let $\Phi_t=\id+t\psi$ with  $\psi$ supported in a cube  $Q_\ell $ of sidelength $\ell$.
Assume $\XN$ is a configuration such that $\F_N(\XN, \mu)<\infty$.
Let 
\begin{multline*}
\Ani_1(\XN, \mu, \psi):= \iint_{\triangle^c}\psi(x)\cdot \nab \g(x-y) d(\sum_{i=1}^N \delta_{x_i} - N \mu)(x)  d(\sum_{i=1}^N \delta_{x_i} - N \mu)(y)\\
= \hal \iint_{\triangle^c}(\psi(x)-\psi(y)) \cdot \nab \g(x-y) d(\sum_{i=1}^N \delta_{x_i} - N \mu)(x)  d(\sum_{i=1}^N \delta_{x_i} - N \mu)(y)\end{multline*}
and more generally \begin{multline}
\label{defani}
\Ani_k(\XN, \mu, \psi)\\:=  \hal \iint_{\triangle^c \cap (Q_\ell \times \R^\d)}\sum_{|\alpha|= k} \frac{ D^\alpha \g(x-y) }{\alpha!}(\psi(x)-\psi(y))^{\alpha}d \( \sum_{i=1}^N \delta_{x_i}-N \mu\)(x) d \( \sum_{i=1}^N \delta_{x_i}-N \mu\)(y)   .\end{multline}
The function $\Ani_k(\XN, \mu, \psi)$ is $k$-homogeneous in $\psi$ and is the $k$-th derivative at $t=0$ of $
\F_N(\Phi_t(\XN), \Phi_t\# \mu)$.  Moreover, for $|t||\psi|_{C^1}<1$, we have
\be \label{derF} 
\frac{d}{dt} \F_N(\Phi_t(\XN), \Phi_t\# \mu)= \Ani_1(\Phi_t(\XN), \Phi_t\#\mu, \psi \circ \Phi_t^{-1} )
\ee  and 
\be\label{derK}
\frac{d}{dt} \log \K_N(\Phi_t\#\mu)= - \beta N^{\frac2\d-1}\Esp_{\mathsf{Q}_N(\Phi_t\#\mu)} \( \Ani_1(\Phi_t(\XN), \Phi_t\# \mu, \psi \circ \Phi_t^{-1}) \).\ee
\end{lem}
\begin{proof}
We denote $\mu_t=\Phi_t\#\mu$.
We return to the definition \eqref{defF} and use it to find that if we set
$$\Xi(t):= \F_N(\Phi_t(\XN), \mu_t)$$
we have by definition of the push-forward
$$\Xi(t)= \hal \iint_{\triangle^c} \g(\Phi_t(x)-\Phi_t(y))d \( \sum_{i=1}^N \delta_{x_i}-N \mu\)(x) d \( \sum_{i=1}^N \delta_{x_i}-N \mu\)(y) $$
and we may compute its derivatives
\begin{multline*}\label{derivees}\Xi^{(k)}(t)\\= \hal \iint_{\triangle^c}\sum_{|\alpha|= k} \frac{ D^\alpha \g(\Phi_t(x)-\Phi_t(y)) }{\alpha!}(\psi(x)-\psi(y))^{\alpha} d\( \sum_{i=1}^N \delta_{x_i}-N \mu\)(x) d \( \sum_{i=1}^N \delta_{x_i}-N \mu\)(y) .\end{multline*}
The statement about the derivatives at $t=0$, as well as the relation \eqref{derF} at $t=0$ then follow immediately. The statement \eqref{derF}  can subsequently be extended for any $t$ such that $|t| |\psi|_{C^1}<1$ (this way $\id + t\psi$ is injective) and 
\eqref{derK} follows from \eqref{formul}.
  \end{proof}
  The quantity 
$\Ani_1(\XN, \mu, \psi)$ was also estimated in \cite{smf}, with a functional inequality  that contains additive error terms,  not sharp enough for our purposes here. Instead we get a better bound in the following proposition, whose proof   will occupy Appendix~\ref{appa}. It involves using the electric formulation of the energy (see Section \ref{secelec} for the definitions) and  computing the difference of energies by transporting the ``electric fields", and giving a renormalized meaning to the term $\Ani_1$ via the use of  truncations. This step is what essentially replaces the loop equations.

\begin{prop} \label{prop:comparaison2}Let $\mu$ be a probability measure with a bounded and \cord{$C^{2}$} density.  Let \cord{$\ell \ge 2 N^{-\frac1\d}$.}
Let $\psi\in C^{2}  (\R^\d,\R^\d) $  and assume that there is a set  $U_\ell$ containing an $\ell$-neighborhood of the support of $D\psi$. Let finally $\Phi_t= \id +t \psi$ \cord{and $\mu_t = (\id + t\psi)\#\mu.$}
Set $\#I_N $ for $\# I_{U_\ell}$ and  
\be\label{defvarphi}
\Xi(t): = \F_N^{U_\ell}(\Phi_t(\XN) , \Phi_t\# \mu)+ \(\frac{ \#I_N }{4}\log N \) \indic_{\d=2}+  C_0\#I_N N^{1-\frac2\d},\ee
where $C_0$ is the constant  in Lemma \ref{lem:contrdist1} (hence $\Xi \ge 0$). 
If $t|\psi|_{C^1(U_\ell)}$ is small enough, we have
\be\label{p40} \Xi(t) \le C \Xi(0) \ee
\be \label{p41}|\Xi'(t)| \le C |\psi|_{C^1(U_\ell)} \Xi (t),\ee
\cord{where $C$ depends only on  $\d $ and $\|\mu\|_{L^\infty}$}, 
and  if moreover  $t|\psi|_{C^2} N^{-\frac1\d}\log (\ell N^{\frac1\d})$ is small enough, \cord{for any $\alpha' > 0$ and $0<\sigma\le 1$, 
\begin{align}
\label{p42}& |\Xi''(t)| \le   C
  \Bigg[|\psi|_{C^1}^2  \(1+N^{-\frac1\d} |\mu_t|_{C^1(U_\ell)} + N^{-\frac2\d}|\mu_t|_{C^2(U_\ell)}\)   + |\psi|_{C^2}\|\psi\|_{L^\infty} (1+   N^{-  \frac\sigma\d} |\mu_t|_{C^\sigma(U_\ell)})   \\ \nonumber &+   |\psi|_{C^1} |\psi|_{C^2 }   N^{-\frac1\d} \log (\ell N^{\frac1\d})   (1+ N^{-\frac2\d} |\mu_t|_{C^2(U_\ell)} )  \Bigg]   \\ \nonumber  &\qquad \times \( (N^{\frac1\d}\ell)^{\alpha' (\d-2)} \indic_{\d\ge3} + \log (\ell N^{\frac1\d}) \indic_{\d=2} \)  \Xi(t ) \\ \nonumber  &+
     \ell^{-1} \|\psi\|_{L^\infty}|\psi|_{C^1}\Bigg[     (N^{\frac1\d}\ell)^{-1}  \(1+ N^{-\frac{1+\sigma}\d} |\mu_t|_{C^{1+\sigma}(U_\ell)} +N^{-\frac1\d}|\mu_t|_{C^1(U_\ell)} \) \Xi(t) \indic_{\d=2}  \\
\nonumber    &+  \Bigg(\( (N^{\frac1\d}\ell)^{1-2\alpha'}(1 +  N^{-\frac{1+\sigma}\d} |\mu_t|_{C^{1+\sigma}(U_\ell)} ) +  (N^{\frac1\d}\ell)^{1-\alpha'} N^{-\frac1\d}|\mu_t|_{C^1(U_\ell)}\)
      \Xi(t)\\  \nonumber &\qquad+ (N^{\frac1\d} \ell)^{1-2\alpha'} N^{-\frac1\d}   \Xi(t)^{  \frac{\d-1}{\d-2}}  \Bigg) \indic_{\d\ge 3}\Bigg]
\end{align}
}
where  $C$ depends \cord{only on $ \d$, $\|\mu\|_{L^\infty}(U_\ell)$ and the bounds on $t |\psi|_{C^1}$ and 
$t|\psi|_{C^2} N^{-\frac1\d} \log (\ell N^{\frac1\d})$.}
Moreover, we have $\Xi'(0)=\Ani_1(\XN, \mu, \psi)$, $\Xi'(t)= \Ani_1(\Phi_t(\XN),  \Phi_t\#\mu,\psi \circ \Phi_t^{-1})$
 and 
for any $\veta$ such that $\eta_i \le \rr_i$ for each $i$,  we have
 \begin{multline}\label{formuleA}
\Ani_1(\XN,\mu, \psi)\\ = \frac{1}{2\cd}\int_{\R^\d} \nab h^\mu_{\veta}[X_N] \cdot \( (2D\psi-(\div \psi) \id) \nab h^\mu_{\veta} \)+  \sum_{i=1}^N \dashint_{\pa B(x_i, \eta_i)}  \nab \tilde h_i (x)
   \cdot \( \psi(x)-\psi(x_i)\) 
\\+ \hal
\sum_{i=1}^N  \dashint_{\pa B(x_i, \eta_i)}   \eta_i^{1-\d}  \( (\psi(x)-\psi(x_i))\cdot \nu\)  
 -N \sum_{i=1}^N \int_{B(x_i, \eta_i)}\nab \f_{\eta_i} (x)\cdot (\psi(x)-\psi(x_i) ) d\mu(x) 
\end{multline}
where  $DT $ means $(\partial_i T_j)_{ij}$ and $\tilde h_i= h^\mu[\XN]-\g(\cdot -x_i)$.
Thus the right-hand side in \eqref{formuleA} is independent of $\veta$ as long as  $\eta_i\le \rr_i$.
We also have 
\begin{multline}\label{pourthomas}
|\Ani_1(\XN, \mu, \psi)|\le C \int_{\R^\d} |\nab h_{\frac14 \rr}^\mu |^2|D\psi|
\\+ C \sum_{i=1}^N |\psi|_{C^1(B(x_i , \frac14 \rr_i))}
\( \int_{B(x_i, \rr_i)} |\nab h_{\frac14 \rr}^{\mu}|^2  +\rr_i^{2-\d}+  N^{1-\frac2\d} \|\mu\|_{L^\infty} \)
 ,\end{multline} with $C$ as above.
\end{prop}

\cord{\begin{remark}If the density $\mu$ is bounded below,  the norms $|\mu_t|_{C^1}$ and $|\mu_t|_{C^2}$ can be estimated in terms of the norms of $\mu$ and $\psi$  via \eqref{muci} and \eqref{muci2} applied to $t \psi$.
Disregarding the dependence in the norms of $\mu$, this then yields in place of \eqref{p42}
\begin{align} \label{deriv23}
 &|\Xi''(t)|\le  
    C
  \Bigg[|\psi|_{C^1}^2  \(1+N^{-\frac1\d}   t |\psi|_{C^2}  + N^{-\frac2\d}(t^2 |\psi|_{C^2}^2+t |\psi|_{C^3})      \) + |\psi|_{C^2}\|\psi\|_{L^\infty} (1+   N^{-  \frac1\d} t |\psi|_{C^2})   \\  \nonumber &+   |\psi|_{C^1} |\psi|_{C^2 }   N^{-\frac1\d} \log (\ell N^{\frac1\d})   \(1+ N^{-\frac2\d}(t^2 |\psi|_{C^2}
 ^2+ t|\psi|_{C^3}  ) \)
    \Bigg]  \\ \nonumber & \qquad  \times\( (N^{\frac1\d}\ell)^{\alpha' (\d-2)} \indic_{\d\ge3} + \log (\ell N^{\frac1\d}) \indic_{\d=2} \)  \Xi(t ) \\  \nonumber  &+
     \ell^{-1} \|\psi\|_{L^\infty}|\psi|_{C^1}\Bigg[     (N^{\frac1\d}\ell)^{-1}  \(1+ N^{-\frac{2}\d} 
     (  t^2 |\psi|_{C^2}^2+ t |\psi|_{C^3}   )
     +t N^{-\frac1\d} |\psi|_{C^2}     \) \Xi(t) \indic_{\d=2}  \\
    & \nonumber + \Bigg( \( (N^{\frac1\d}\ell)^{1-2\alpha'}\(1 +    N^{-\frac{2}\d}
  (  t^2 |\psi|_{C^2 }^2+t |\psi|_{C^3}  ) \)
     +   (N^{\frac1\d}\ell)^{1-\alpha'}   N^{-\frac1\d} (1+t|\psi|_{C^2}) \)
      \Xi(t)\\ &\nonumber \qquad+ (N^{\frac1\d} \ell)^{1-2\alpha'} N^{-\frac1\d}   \Xi(t)^{  \frac{\d-1}{\d-2}}  \Bigg) \indic_{\d\ge 3}\Bigg]
,\end{align}
where $C$ depends only on the norms of $\mu$, a lower bound for $\mu$,  and on $\d$.
\end{remark}}
Since $\Xi'(0)= \Ani_1(\XN,  \mu ,\psi)$ and $\Xi''(0)= \Ani_2 (\XN, \mu, \psi)$, in view of \eqref{p41} and \eqref{p40} we have proven 
\begin{coro}
We have 
\be \label{pourm} |\Ani_1(\XN, \mu, \psi) |\le C  |\psi|_{C^1} \(\F_N^{U_\ell} (\XN, \mu) + \( \frac{\# I_N}{4} \log N\) \indic_{\d=2} + C_0 \# I_N N^{1-\frac2\d} \)\ee \cord{where $C$ depends only on $\d$ and $\|\mu\|_{L^\infty}$,}
and \cord{if $\d=2$,
\begin{align} \label{pourm2}  |\Ani_2(\XN, \mu ,\psi) | & \le
 C
  \Bigg[ \Bigg( |\psi|_{C^1}^2   + |\psi|_{C^2}\|\psi\|_{L^\infty}   +   |\psi|_{C^1} |\psi|_{C^2 }   N^{-\frac1\d} \log (\ell N^{\frac1\d})   \Bigg)
      \log (\ell N^{\frac1\d})   
   \\ \nonumber & \qquad + \ell^{-1} \|\psi\|_{L^\infty}|\psi|_{C^1}    (N^{\frac1\d}\ell)^{-1}          \Bigg]\\ & \nonumber
 \qquad \qquad  \qquad \times\(\F_N^{U_\ell} (\XN, \mu) + \( \frac{\# I_N}{4} \log N\) \indic_{\d=2} + C_0 \# I_N N^{1-\frac2\d}  \)     
, \end{align} }
\cord{where $C$ depends only on $\d$ and the norms of $\mu$.}
\end{coro}
The relation \eqref{pourm} provides an improved (and sharp) functional inequality compared to \cite{smf}, while \eqref{pourm2} is new. 
Let us point out that a shorter proof of \eqref{pourm}  was provided in \cite{ros} and, in dimension $\d=2$ an estimate similar to  \eqref{pourm2} but with non-optimal right-hand side in \cite{ros2}, both after the first version of this paper was  completed.

\begin{remark}\label{rem4}Taylor expanding $\psi$ and  using also that for a matrix $A$, we have 
$$\int_{\p B_1}  A \nu \cdot \nu\, dS= tr( A) |B_1|$$ 
where  $B_1$ is the unit ball of $\R^\d$ and $\nu$ stands for the outer unit normal to $\p B_1$,
we find that the sum of the  last three terms in the right-hand side of \eqref{formuleA} is equal to 
\begin{equation*}
 \frac{1}{2 \d}\sum_{i=1}^N  \eta_i^{2-\d}( \div \psi)(x_i)+ O\( \eta_i^{3-\d}  \) +o(1) \quad \text{as} \ \eta_i\to 0.
\end{equation*}
This way one obtains how the ``loop equation" type term 
$$\int_{\R^\d} \nab h^\mu_{\veta}\cdot\( (2D\psi-(\div \psi) \id) \nab h^\mu_{\veta} \)$$
needs to be renormalized as $\eta_i\to 0$. 
In dimension $2$, one finds  as  in \cite{ls2} 
\be \Ani_1(\XN, \mu, \psi) = \lim_{\eta_i \to 0}\frac1{2\cd} \int_{\R^\d} \nab h^\mu_{\veta} \cdot \( (2D\psi-(\div \psi) \id) \nab h^\mu_{\veta} \)  + \frac14 \sum_{i=1}^N  \div \psi(x_i).\ee
In dimension $3$, the renormalization is more complicated, and one needs to assume additional regularity of $\psi$ to compute all the nonvanishing orders. One finds 
\begin{multline*} \Ani_1(\XN, \mu, \psi) = \lim_{\eta_i \to 0} \frac1{2\cd} \int_{\R^\d} \nab h^\mu_{\veta} \cdot \( (2D\psi-(\div \psi) \id) \nab h^\mu_{\veta} \)+\frac{1}{6} \sum_i \frac{ 1}{\eta_i} \div \psi(x_i) \\   + \hal \sum_{j,k,m} \partial_j \partial_k  \psi_m(x_i) \dashint_{\pa B_1}\nu_k \nu_j \nu_m ,\end{multline*} with the last term vanishing by symmetry. 
In higher dimension, more and more derivatives of $\psi$ are needed in order to fully express the expansion.
\end{remark}

We also record the following variant for Neumann problems in cubes.

\begin{lem}\label{derivK}
Assume $\mu_0$ is a positive measure with a  bounded and $C^2$ density in a hyperrectangle $Q_\ell$ of sidelengths in $[\ell, 2\ell]$, with $\ell \ge N^{-\frac1\d}$ and $N\mu_0(Q_\ell)=\mn$ an integer. 
Let $\psi\in C^{2}(Q_\ell,Q_\ell)$ satisfying $\psi\cdot \nu=0$ on $\pa Q_\ell$ where $\nu$ denotes the outer unit normal, and  let $\Phi_t=\id+t\psi$ and $\mu_t=\Phi_t\#\mu_0$.
Let $\mathsf{Q}_N^{(t)}$ denote the Gibbs measure $\mathsf{Q}_N(Q_\ell, \mu_t)$ as in \eqref{defQ}, and let 
$\Xi(t):= \F_N (\Phi_t(X_\mn), \mu_t, Q_\ell)  +\( \frac{\mn}{4}\log N \)\indic_{\d=2}  + C_0\mn N^{1-\frac2\d}$, with $C_0$ the constant  in Lemma~\ref{lem:contrdist1}. 

Then there exists a function $\Ani_1(X_\mn, \mu,\psi )$ linear in $\psi$ such that if $t|\psi|_{C^1}$ is small enough
\begin{itemize}
\item
\be\label{prop2} \Ani_1(- X_\mn, \mu_0(-\cdot), \psi(-\cdot)) = -\Ani_1(X_\mn,  \mu_0,\psi)\ee
\item
\be \label{pourmini}
\Xi'(t)= \Ani_1(\Phi_t(X_\mn),  \mu_t,\psi \circ \Phi_t^{-1})
\ee
\item \be\label{prop1}
|\Xi'(t) |\le  C|\psi|_{C^{1}}\Xi(t)\ee 
\item
\begin{equation}\label{pourgibbs}
\frac{d}{dt} \log \K_N(Q_\ell, \mu_t) = \Esp_{\mathsf{Q}_N^{(t)}} \( -  \beta N^{\frac{2}{\d}-1}\Ani_1(\Phi_t(X_\mn),  \mu_t,\psi\circ \Phi_t^{-1})\)\end{equation}
\item if moreover $t|\psi|_{C^2} N^{-\frac1\d}\log (\ell N^{\frac1\d})$ is small enough,
\cord{for any $\alpha'>0, 0<\sigma\le 1$, 
\begin{align} \label{deriv2}
 &|\Xi''(t)|\le  
    C
  \Bigg[|\psi|_{C^1}^2  \(1+N^{-\frac1\d} |\mu_t|_{C^1(U_\ell)} + N^{-\frac2\d}|\mu_t|_{C^2(U_\ell)}\)   + |\psi|_{C^2}\|\psi\|_{L^\infty} (1+   N^{-  \frac\sigma\d} |\mu_t|_{C^\sigma(U_\ell)})   \\  \nonumber &+   |\psi|_{C^1} |\psi|_{C^2 }   N^{-\frac1\d} \log (\ell N^{\frac1\d})   (1+ N^{-\frac2\d} |\mu_t|_{C^2(U_\ell)} )  \Bigg]   \(1+ (N^{\frac1\d}\ell)^{\alpha' (\d-2)} \indic_{\d\ge3} + \log (\ell N^{\frac1\d}) \indic_{\d=2} \)  \Xi(t ) \\  \nonumber  &+
     \ell^{-1} \|\psi\|_{L^\infty}|\psi|_{C^1}\Bigg[     (N^{\frac1\d}\ell)^{-1}  \(1+ N^{-\frac{1+\sigma}\d} |\mu_t|_{C^{1+\sigma}(U_\ell)} +N^{-\frac1\d}|\mu_t|_{C^1(U_\ell)} \) \Xi(t) \indic_{\d=2}  \\
    & \nonumber +  \Bigg(\( (N^{\frac1\d}\ell)^{1-2\alpha'}(1 +  N^{-\frac{1+\sigma}\d} |\mu_t|_{C^{1+\sigma}(U_\ell)} ) +  (N^{\frac1\d}\ell)^{1-\alpha'} N^{-\frac1\d}|\mu_t|_{C^1(U_\ell)}\)
      \Xi(t)\\ &\nonumber \qquad+ (N^{\frac1\d} \ell)^{1-2\alpha'} N^{-\frac1\d}   \Xi(t)^{  \frac{\d-1}{\d-2}}  \Bigg) \indic_{\d\ge 3}\Bigg]
,\end{align} where $C$ depends only on $\d$ and $\|\mu_0\|_{L^\infty}$. }\end{itemize}
\end{lem}
\begin{proof}
If one ignores  the part of $\F_N$ in the second line of its definition \cord{\eqref{Glocal}},
then the results \eqref{prop1} and \eqref{pourmini} and \eqref{deriv2} can be deduced from Proposition \ref{prop:comparaison2} after periodizing the configuration by  doing a reflection with respect to the boundary of $Q_\ell$, and extending $\psi$ into a compactly supported map. They can also be deduced by following the same steps as in the proof of Proposition \ref{prop:comparaison2}.
Then to include the part 
\be\label{fgg} \sum_{i=1}^{\mn}\(\g(\frac14 \dist(x_i, \pa Q_\ell) -\g(\frac{N^{-\frac1\d}}{4} ) \)_+\, ,\ee
it suffices to remark that 
the first derivative of the  function $t\mapsto \g\(\frac14 \dist(\Phi_t(x_i), \pa Q_\ell)\)$ is 
$\frac14(\dist (x_i, \pa Q_\ell))^{1-\d} \psi (x_i)\cdot \nu ,$ where $\nu$ is the outer unit normal to $Q_\ell$,
 and since $\psi $ is Lipschitz and $\psi\cdot \nu= 0$ on $\pa Q_\ell$, we may bound it by $ O\( |\psi|_{C^1(Q_\ell)} \g\(\frac14 \dist(x_i, \pa Q_\ell)\) \).$  By the same arguments, the second derivative is bounded by  $ O\( |\psi|_{C^1(Q_\ell)}^2 \g\(\frac14 \dist(x_i, \pa Q_\ell)\)\).$
Summing this over $i$ gives terms that are straightforwardly bounded in terms of  \eqref{fgg} hence of $\F_N$ itself, so the results \eqref{prop1}, \eqref{pourmini} and \eqref{deriv2} hold.
 
The statement \eqref{prop2} is a simple symmetry argument.
The result \eqref{pourgibbs} is obtained just as \eqref{derK} from \eqref{formul}. \end{proof}

\subsection{Variation  of free energy}
We now show  estimates that bound the variation of $\log \K$ with respect to $\mu$, taking advantage of the transport approach and \eqref{derK}, respectively \eqref{pourgibbs}.
We start with the setting of a hyperrectangle.

\begin{lem}\label{lemcompdeskcube}
Assume $ \rb N^{-1/\d}\le \ell \le C$.
Let $\mu_0,\mu_1\in C^{1}$   be two  densities bounded  above and below by positive constants in $Q_\ell$, a hyperrectangle of sidelengths in $[\ell, 2\ell]$ with $N\mu_0(Q_\ell)=N\mu_1(Q_\ell)=\mn$ an integer.
Then 
\begin{multline}\label{compdeskcube}
|\log \K_N(Q_\ell, \mu_1)-\log \K_N(Q_\ell, \mu_0)|\\
\le C \beta \chi(\beta) N\ell^\d \(\ell^2   \left\| \frac{1}{\mu_0}\right\|^2_{L^\infty}  |\mu_0|_{C^1}  |\mu_1-\mu_0|_{C^1} 
+ \ell  \left\| \frac{1}{\mu_0}\right\|_{L^\infty} |\mu_1-\mu_0|_{C^1}\),\end{multline}
where $C$ depends only on $\d$. \end{lem}

\begin{proof}
Let us solve 
\be \label{solxi}
\left\{
\begin{array}{ll}
-\Delta \xi= \mu_1-\mu_0& \text{in} \ Q_\ell\\ [2mm]
\frac{\pa \xi}{\pa \nu}=0 & \text{on} \ \pa Q_\ell .\end{array}\right.
\ee
By elliptic regularity and scaling we have  
$$|\xi|_{C^1}\le C \ell^2|\mu_1-\mu_0|_{C^1}, \quad |\xi|_{C^2}\le C \ell|\mu_1-\mu_0|_{C^1}.$$
Setting $$\psi:= \frac{\nab \xi}{\mu_0},$$ we thus have 
\begin{multline} \label{418}|\psi|_{C^1}\le C \(\left\| \frac{1}{\mu_0}\right\|^2_{L^\infty} |\mu_0|_{C^1}    |\xi|_{C^1}+   \left\|\frac{1}{\mu_0}\right\|_{L^\infty}  |\xi|_{C^2}\) \\
\le C \( \left\| \frac{1}{\mu_0}\right\|^2_{L^\infty} \ell^2   |\mu_0|_{C^1}  |\mu_1-\mu_0|_{C^1} 
+ \ell \left\| \frac{1}{\mu_0}\right\|_{L^\infty} |\mu_1-\mu_0|_{C^1}\),
\end{multline} where  
\be -\div (\psi\mu_0)= \mu_1-\mu_0.\ee

Let  now $\nu_s = (\id + s\psi)\#   \mu_0 $ and $\mu_s= (1-s) \mu_0+s\mu_1$.
We have  $$\frac{d}{ds}\Big|_{s=0} \nu_s= -\div (\psi \mu_0)= \mu_1-\mu_0= 
\frac{d}{ds}\Big|_{s=0} \mu_s,$$ thus using \eqref{pourgibbs}, we have
\be\label{dtl}
\frac{d}{ds}\Big|_{s=0}\log \K_N(Q_\ell, \mu_s)= 
\frac{d}{ds}\Big|_{s=0} \log \K_N(Q_\ell, \nu_s)=
\Esp_{\mathsf{Q}_N(Q_\ell, \mu_0)} \( - \beta  N^{\frac{2}{\d}-1}  \Ani_1(X_\mn,  \mu_0,\psi) \).\ee
Inserting \eqref{prop1}, \eqref{418} and the local laws \eqref{locallawint0} 
 we deduce that 
 \begin{equation*}
\left|\frac{d}{ds}\Big|_{s=0} \log \K_N(Q_\ell, \mu_s)\right|
\le C   \beta \chi(\beta) (N\ell^\d+\mn) \(\ell^2   \left\| \frac{1}{\mu_0}\right\|^2_{L^\infty}  |\mu_0|_{C^1}  |\mu_1-\mu_0|_{C^1} 
+ \ell \left\| \frac{1}{\mu_0}\right\|_{L^\infty} |\mu_1-\mu_0|_{C^1}\). \end{equation*}
Since $\mn \le N \|\mu_0\|_{L^\infty}$ we find 
\be
 \left| \frac{d}{ds}\Big|_{s=0} \log \K_N(Q_\ell, \mu_s)\right|\le C    \beta \chi(\beta) N\ell^\d \(\ell^2 \left\| \frac{1}{\mu_0}\right\|^2_{L^\infty}   |\mu_0|_{C^1}  |\mu_1-\mu_0|_{C^1} 
+ \ell \left\| \frac{1}{\mu_0}\right\|_{L^\infty} |\mu_1-\mu_0|_{C^1}\). \ee
The same reasoning can be applied near any $s \in [0,1]$ yielding 
\be
 \left|\frac{d}{ds} \log \K_N(Q_\ell, \mu_s)\right|\le C    \beta \chi(\beta) N\ell^\d \(\ell^2 \left\| \frac{1}{\mu_0}\right\|^2_{L^\infty}  |\mu_0|_{C^1}  |\mu_1-\mu_0|_{C^1} 
+ \ell \left\| \frac{1}{\mu_0}\right\|_{L^\infty} |\mu_1-\mu_0|_{C^1}\). \ee
Integrating between $0$ and $1$ gives the result.
 \end{proof}
 
Next, we want to show the analogous result for $\log \K_N(\R^\d, \mu)$ when $\mu$ varies only in a hyperrectangle $Q_\ell$. The difficulty is to build a transport which also stays compactly supported in $Q_\ell$ (solving Laplace's equation does not work). For that we use the following.
\begin{lem}\label{divU}
Assume $f$ is $C^1$ and compactly supported in $Q_\ell$, a hyperrectangle of sidelengths in $[\ell, 2\ell]$  with $\int_{Q_\ell} f=0$.
Then there exists a vector field $U: Q_\ell \to \R^\d$ compactly supported in $Q_\ell$, such that 
$$\div U= f  \quad \text{in} \ Q_\ell$$
and 
\be\label{estimU} 
 \|U\|_{L^\infty(Q_\ell)} \le  C \ell \|f\|_{L^\infty(Q_\ell)}, \quad |U|_{C^1(Q_\ell)} \le C  (\ell |f|_{C^1(Q_\ell)}+\|f\|_{L^\infty(Q_\ell)}),\ee where $C$ depends only on $\d$.
 \end{lem}
 \begin{proof} Without loss of generality we  may assume that  $Q_\ell=\prod_{i=1}^\d [0,\ell_i]$ with $\ell_i \le 2 \ell.$
We prove the result by induction on $\d$, as  a linearization of Knotte-Rosenblatt rearrangement. The case $\d=1$ is easy, we just let $U(x)=\int_0^x f(s)ds$.
Assume then that  the result is true up to $\d-1$. Then set $$g(x_1, \dots , x_{\d-1})=\frac{1}{\ell_\d} \int_0^{\ell_\d} 
f(x_1, \dots, x_{\d-1}, s) ds.$$
 The function $g $ is compactly supported in $\prod_{i=1}^{\d-1} [0, \ell_i]$ and of integral $0$. Thus by the induction hypothesis we may find a vector field $U'(x_1, \dots, x_{\d-1})$ with values in $\R^{\d-1}$, compactly supported in $\prod_{i=1}^{\d-1} [0, \ell_i]$ such that 
$\div U'= g$ in $\prod_{i=1}^{\d-1} [0, \ell_i]$ and 
\be \|U'\|_{L^\infty} \le C \ell \|g\|_{L^\infty} \le C \ell\|f\|_{L^\infty} , \quad |U'|_{C^1} \le C( \ell |g|_{C^1}+ \|g\|_{L^\infty})
\le  2 C( \ell |f|_{C^1}+\|f\|_{L^\infty}).\ee
Let also $$u(x_1, \dots, x_\d) =\int_0^{x_\d} f(x_1, \dots, x_{\d-1}, s) ds  - \frac{x_\d}{\ell_\d} \int_0^{\ell_\d}f(x_1, \dots, x_{\d-1}, s) ds.$$ Again $u$ is compactly supported in $Q_\ell$, and 
$$\|u\|_{L^\infty} \le 2\ell_\d \|f\|_{L^\infty} \quad |u|_{C^1} \le C \ell_\d |f|_{C^1}.$$
 Setting $U(x_1, \dots, x_\d) =(U'(x_1, \dots, x_{\d-1}) , u(x_1, \dots, x_\d) )$, we have that $U$ is compactly supported in 
$Q_\ell$, that 
$$\div U= g+\pa_{x_\d} u = f$$ and
 that \eqref{estimU} hold.  The result is thus true by induction.
\end{proof}

\begin{lem}\label{lemcompdesk}
Assume $ \ell$ satisfies \eqref{ass1}.
Let $\mu_0,\mu_1\in C^{1}$   be two  densities bounded above and below by positive constants in $Q_\ell$, a hyperrectangle of sidelengths in $[\ell, 2\ell]$ with $N\mu_0(Q_\ell)=N\mu_1(Q_\ell)=\mn$ an integer, and coinciding outside $Q_\ell$.
Then 
\begin{multline}\label{compdesk}
|\log \K_N(\R^\d, \mu_1)-\log \K_N(\R^\d, \mu_0)|\\ 
\le C  \beta \chi(\beta) N\ell^\d \( \ell |\mu_0|_{C^1(Q_\ell)} \|\mu_1-\mu_0\|_{L^\infty(Q_\ell)} + \ell |\mu_1-\mu_0|_{C^1(Q_\ell)}+ \|\mu_1-\mu_0\|_{L^\infty(Q_\ell)} \) \end{multline}
where $C$ depends on $\d$ and the upper and lower bounds for $\mu_0$ and $\mu_1$.
\end{lem}

\begin{proof}
Let us apply Lemma \ref{divU} to $f=\mu_1-\mu_0$, and set $\psi:= \frac{- U}{\mu_0}$.
We thus have 
 \be \label{solxi} -\div (\psi\mu_0)= \mu_1-\mu_0\ee
and 
\be \label{419}|\psi|_{C^1}\le C \( \ell |\mu_0|_{C^1} \|\mu_1-\mu_0\|_{L^\infty} +\ell |\mu_1-\mu_0|_{C^1}+ \|\mu_1-\mu_0\|_{L^\infty} \) \ee where $C$ depends on $\d$  and the upper and lower bounds for $\mu_0$ and $\mu_1$.

Let  now $\nu_s = (\id + s\psi)\#   \mu_0 $ and $\mu_s= (1-s) \mu_0+s\mu_1$.
We have  $\frac{d}{ds}\big|_{s=0} \nu_s= -\div (\psi \mu_0)= \mu_1-\mu_0= 
\frac{d}{ds}\big|_{s=0} \mu_s$, thus using \eqref{derK}, we have
\be\label{dtl}
\frac{d}{ds}\Big|_{s=0} \log \K_N(\R^\d, \mu_s)= 
\frac{d}{ds}\Big|_{s=0} \log \K_N(\R^\d, \nu_s)=
\Esp_{\mathsf{Q}_N(\R^\d, \mu_0)} \( - \beta  N^{\frac{2}{\d}-1}  \Ani_1(X_N,  \mu_0,\psi) \).\ee
Inserting \eqref{p41}, \eqref{419} and the local laws \eqref{locallawint0} 
 we deduce that 
 \begin{multline*}
 \left|\frac{d}{ds}\Big|_{s=0} \log \K_N(\R^\d, \mu_s)\right|\\ \le C   \beta \chi(\beta) (N\ell^\d+\mn) 
  \( \ell |\mu_0|_{C^1} \|\mu_1-\mu_0\|_{L^\infty} + \ell |\mu_1-\mu_0|_{C^1}+ \|\mu_1-\mu_0\|_{L^\infty} \). \end{multline*}
Since $\mn \le N \|\mu_0\|_{L^\infty}$ we find 
\begin{multline*}
 \left|\frac{d}{ds}\Big|_{s=0} \log \K_N(\R^\d, \mu_s)\right|\\ \le C    \beta \chi(\beta) N\ell^\d
  \( \ell |\mu_0|_{C^1} \|\mu_1-\mu_0\|_{L^\infty} + \ell |\mu_1-\mu_0|_{C^1}+ \|\mu_1-\mu_0\|_{L^\infty} \). \end{multline*}
The same reasoning can be applied near any $s \in [0,1]$ yielding 
\begin{multline*}
 \left|\frac{d}{ds} \log \K_N(\R^\d, \mu_s)\right|\\ 
 \le C    \beta \chi(\beta) N\ell^\d \( \ell |\mu_0|_{C^1} \|\mu_1-\mu_0\|_{L^\infty} + \ell |\mu_1-\mu_0|_{C^1}+ \|\mu_1-\mu_0\|_{L^\infty} \). \end{multline*}
Integrating between $0$ and $1$ gives the result.
 \end{proof}

\section{Study of fluctuations}\label{sec5}
We are  now  in a position to return to \eqref{laplace0} and estimate its various terms. As explained in Section \ref{secoutline}, since it is difficult to find and evaluate an exact transport from $\mut$ to $\mutt$, we instead (as in \cite{ls2,bls}) replace $\mutt$ by an approximation $\tilde \mutt$ of the form $(\id+ t\psi)\#\mut$, which is the same as $\mutt$ at first order in $t$.

We recall that 
\be \label{defopL} L:= \frac{1}{\cd \mut} \Delta, \ee
and  that from \eqref{bornemutcn}, $\mut$ is uniformly bounded in $C^{2m+\gamma-4}$. This way the iterates  $L^k$ of $L$ satisfy the estimate 
\be \label{Lin}|L^k (\xi)|_{C^\sigma} \le C \sum_{m=\min(2k,2)}^{ 2k+\sigma } |\xi|_{C^m}\quad  \text{ as long as } 2k+\sigma\le 2m+\gamma-4 \ee
where $C$ depends on $V, \sigma, k$.
 We will use this fact repeatedly.
   
\subsection{Choice of transport}
We now choose $\psi$ to define $\tilde\mutt$.
By definition, $\mutt$ being the thermal equilibrium measure associated to $V_t=V+t\xi$, it satisfies 
\be \label{hhe}
\g*{\mutt}+V+t\xi+\frac{1}{\theta}\log \mutt=C_t \quad \text{in} \ \R^\d.\ee
Comparing with \eqref{eqmb} and linearizing in $t$, we find that we should choose $\psi$ solving 
\be \label{hightemp} - \g*(\div (\psi\mut)) + \xi -\frac{1}{\theta \mut}\div(\psi\mut)=0.\ee
 This can be solved exactly  by
letting $h$ solve 
$$ -\frac{\Delta h}{\cd \theta \mut}+h=\xi$$
then taking
$$\psi =- \frac{\nab h}{\cd \mut}.$$
However, this $\psi$ fails to be localized on the support of $\xi$, and it is delicate to 
show  good bounds for it.

Instead we use two approximations.
The first is the  transport of $\mut$ by the map
\be\label{choicepsi}
\psi :=- \frac1{\cd\mut}\sum_{k=0}^q \frac{\nab L^k (\xi)}{\theta^k}, 
\ee
that is  \be \label{defmutt}
\tilde \mutt:= (\id+ t\psi)\# \mut.\ee The second is
\be\label{defnut} 
\nu_\theta^t := \mut+ \frac{t}{\cd}\sum_{k=0}^q\frac{ \Delta L^k(\xi)}{\theta^k}.\ee
Here $q$ is an integer to be chosen depending on the regularity of $V$ and $\xi$. The larger $q$ the more precise the approximation.
We will show that $\nu_\theta^t$ is a good approximation of $\tilde \mutt$. Also  $\nu_\theta^t $ is convenient because it is easy to compute and because it is an approximate solution to \eqref{hhe}, as we see below.

We note that $\nu_\theta^t-\mut$ is supported in $Q_\ell$ which contains the support of $\xi$.
Moreover $\int \nu_\theta^t=\int \mut=1$ hence, since $\mut \ge \frac{\alpha}{2\cd}$ in $\supp \, \xi\subset\hat\Sigma$ by \eqref{assumpV4},  for $\nu_\theta^t$ to be a probability density it suffices that
\be \label{condsurt}
\left\| t \sum_{k=0}^q \frac{\Delta L^k(\xi)}{\theta^k}\right\|_{L^\infty}<\frac{\alpha}{4}.\ee
We will also need the condition 
\be \label{condsurt2}
\left\| t \frac{1}{\mut} \sum_{k=0}^q \frac{\nab L^k(\xi)}{\theta^k}\right\|_{L^\infty} <\frac{\alpha}{2\cd} \quad \text{and} \ \ 
  \left|  t \frac{1}{\mut}  \sum_{k=0}^q \frac{\nab L^k (\xi)} {\theta^k}  \right|_{C^1}  < \frac{\alpha}{2\cd} \ee
which ensures in view of \eqref{choicepsi} that
\be \label{condsurt3}
 |t|(\|\psi\|_{L^\infty}+|\psi|_{C^1}) <1,\ee
since without loss of generality we may assume that $\alpha<\cd$.

We start with a general lemma about the error made when replacing an exact transport by a linearized transport. The main  point is that the right-hand side is quadratic in $\psi$. \cord{We also insert a general control for transported densities.}
\begin{lem}\label{linea}
Assume $\mu\in C^3$ is a positive density bounded above and  below by positive constants in the support of  $\psi$, where $\psi$ is a $C^1$ map such that 
\be\label{condsurt0}
\|\psi\|_{L^\infty}+|\psi|_{C^1}<1.\ee
Then for any $\sigma \in [0,1]$, we have
\begin{multline}\label{estmu0} 
|(\id+\psi)\# \mu- \( \mu- \div (\psi \mu)\)|_{C^\sigma}\\
\le C \(  |\mu|_{C^2}\|\psi\|_{L^\infty}^2 + |\psi|_{C^1}^2 +|\psi|_{C^2} \|\psi\|_{L^\infty} \)^{1-\sigma}
\\
\times
\Big(|\mu|_{C^2} |\psi|_{C^1} \|\psi\|_{L^\infty} + |\mu|_{C^3} \|\psi\|_{L^\infty}^2+ |\mu|_{C^1}  \|\psi\|_{L^\infty}  |\psi|_{C^1}  +|\mu|_{C^2}   \|\psi\|_{L^\infty}^2   \\+  
     |\psi|_{C^1} |\psi|_{C^2} +|\psi|_{C^3} \|\psi\|_{L^\infty}  \Big)^{\sigma}
\end{multline} where $C$ depends only on $\d$ and the upper and lower bounds for $\mu$.
\cord{Moreover, 
\begin{align}\label{muci}
 |(\id+\psi)\# \mu|_{C^1} \le &  C\(|\mu|_{C^1} + |\psi|_{C^2} \)
\\
\label{muci2}  |(\id+\psi)\# \mu|_{C^2} \le & C\Big( |\mu|_{C^2} + (|\mu|_{C^1}+|\psi|_{C^2}) (1+  |\psi|_{C^2} )  + 
  |\psi|_{C^3} \Big),\end{align} where $C$ depends only on $\d$ and the upper and lower bounds for $\mu$.}
\end{lem}
\begin{proof}
Let $\tilde \mu:= (\id+\psi)\# \mu$ and $\nu= \mu-\div(\psi \mu)$ and $\Phi=\id +\psi$.
By  definition of the push-forward we have
\be\label{explitran}
\tilde \mu= \frac{\mu \circ \Phi^{-1}}{\det(\id +  D\psi)\circ \Phi^{-1}}. \ee
and using a Taylor expansion and \eqref{condsurt0} we may write 
$$\|\mu\circ \Phi^{-1}-\mu-\nab \mu \cdot \psi\|_{L^\infty}\le C   |\mu|_{C^2}\|\psi\|_{L^\infty}^2$$
and also
\begin{multline}\label{mutcphit}
|\mu\circ \Phi^{-1}-\mu-\nab \mu \cdot \psi|_{C^1}\\ \le C \( |\mu|_{C^2}|\psi|_{C^1}\|\psi\|_{L^\infty}
+ |\mu|_{C^3} \|\psi\|_{L^\infty}^2  +|\mu|_{C^1}  \|\psi\|_{L^\infty} |\psi|_{C^1}   +|\mu|_{C^2} \|\psi\|_{L^\infty}^2    \).\end{multline}
Also by Taylor expansion, we find (again with \eqref{condsurt0}) that 
$$\( \det(\id + D\psi)\circ \Phi^{-1}\)^{-1}=1-\div \psi+u$$
with 
$$\|u\|_{L^\infty} \le C \(|\psi|_{C^1}^2 + |\psi|_{C^2}\|\psi\|_{L^\infty}\)$$
and 
$$|u|_{C^1} \le C \(|\psi|_{C^1}|\psi|_{C^2} + |\psi|_{C^3}\|\psi\|_{L^\infty}  \).$$

Combining these relations, it follows that 
$$\|\nu- \tilde \mu\|_{L^\infty} \le C  \( |\mu|_{C^2} \|\psi\|_{L^\infty}^2 + |\psi|_{C^{1}}^2+  |\psi|_{C^2}\|\psi\|_{L^\infty}\)$$
and \begin{multline}
|\nu- \tilde \mu|_{C^{1}}\le  C\Big(
|\mu|_{C^2} |\psi|_{C^1} \|\psi\|_{L^\infty} + |\mu|_{C^3} \|\psi\|_{L^\infty}^2+ |\mu|_{C^1}  \|\psi\|_{L^\infty}  |\psi|_{C^1}  +|\mu|_{C^2}   \|\psi\|_{L^\infty}^2   \\+  
     |\psi|_{C^1} |\psi|_{C^2} +|\psi|_{C^3} \|\psi\|_{L^\infty} \Big)
     \end{multline}
hence  \eqref{estmu0} follows by interpolation.

\cord{In the same way, we check that, using again \eqref{condsurt0},
$$|\mu \circ \Phi^{-1}|_{C^1} \le C |\mu|_{C^1},$$
$$|\mu\circ \Phi^{-1}|_{C^2} \le C\( |\mu|_{C^2} + |\mu|_{C^1}(1+ |\psi|_{C^2}) \).$$
and  also
$$|\det(\id + D\psi)\circ \Phi^{-1}|_{C^1} \le C |\psi|_{C^2}$$
and 
$$|\det(\id + D\psi)\circ \Phi^{-1}|_{C^2}\le C\( |\psi|_{C^3} +|\psi|_{C^2}+ |\psi|_{C^2}^2\).$$
The estimates \eqref{muci} and \eqref{muci2} follow.}
\end{proof}

\begin{lem}\label{choixpsi}
Assume $\theta \ge \theta_0(m) $ so that \eqref{bornemutcn} holds, and  assume \eqref{condsurt} and \eqref{condsurt2} hold.
The choice \eqref{choicepsi} satisfies 
\begin{itemize}
\item The support of $\psi$ is included in the support of $\nab \xi$.
\item We have for every $\sigma\ge 0$  such that $ \sigma+ 2q+4 \le 2m+\gamma$,
\be\label{estpsi} |\psi|_{C^{\sigma}}\le C\sum_{k=0}^q \frac{|\xi|_{C^{\sigma+ 2k+1} (U)}}{\theta^k} \ee
where $C$ depends on $V$, $\sigma $ and $q$.

\item 
If $2m+\gamma \ge 6$ and  \eqref{condsurt3} holds, for $\sigma=1, 2$, we have
\be \label{estnabmut} | \tilde \mutt|_{C^\sigma(\hat \Sigma)}\le C+C t \sum_{k=0}^{\sigma+1} |\psi|_{C^k}
\ee
where $C$ depends on $|\mut|_{C^1}, |\mut|_{C^2}$.

\item If $2m +\gamma\ge 7$, we have for $0\le \sigma \le 1$,
\begin{multline}\label{estmu}|\tilde \mutt-\nu_\theta^t|_{C^{\sigma}}
\le C  t^2 \(  |\mut|_{C^2}\|\psi\|_{L^\infty}^2 + |\psi|_{C^1}^2 +|\psi|_{C^2} \|\psi\|_{L^\infty} \)^{1-\sigma}
\\
\times
\Big(|\mut|_{C^2} |\psi|_{C^1} \|\psi\|_{L^\infty} + |\mut|_{C^3} \|\psi\|_{L^\infty}^2+ |\mut|_{C^1}  \|\psi\|_{L^\infty}  |\psi|_{C^1}  +|\mut|_{C^2}   \|\psi\|_{L^\infty}^2   \\+  
     |\psi|_{C^1} |\psi|_{C^2} +|\psi|_{C^3} \|\psi\|_{L^\infty}  \Big)^{\sigma}
\end{multline}
\item  Letting 
\be \label{defept}
\ep_t:= \g*\nu_\theta^t + V+ t\xi + \frac{1}{\theta}\log \nu_\theta^t-C_\theta \ee with $C_\theta$ as in \eqref{eqmb},
we have  that $\ep_t$ is supported in the support of $\xi$  and  if $2m+\gamma \ge 2q+6$, 
\be \label{estept0}\|\ep_t\|_{L^\infty} \le  C\frac{t^2}{\theta}  \( \sum_{k=0}^q \frac{1}{\theta^k}  |\xi|_{C^{2k+2}}\)^{2}  + C\frac{t}{\theta^{q+1}}   \sum_{k=2}^{2q+2} |\xi|_{C^{k} } ,
\ee
and if in addition $2m+\gamma \ge 2q+7$, 
\be\label{estept} 
|\ep_t|_{C^1}\le  C t^2
\sum_{m=0}^{2q}\frac{1}{\theta^{m+1}} \sum_{p+k=m} |\xi|_{C^{2k+2}}|\xi|_{C^{2p+3}}   + C\frac{t}{\theta^{q+1}}   \sum_{k=2}^{2q+3} |\xi|_{C^{k} } .
\ee
\end{itemize}
Here all the constants $C>0$ depend only on $\d$ and $V$.
\end{lem}
\begin{proof}
The support of $\psi$ is obviously that of $\nab \xi$.
The relation \eqref{estpsi} is a direct calculation following from \eqref{choicepsi} and \eqref{Lin} (and the discussion above it)
with \eqref{bornemutcn}.  The estimate  \eqref{estnabmut} is the result of direct computations starting from the explicit form \eqref{explitran}.

By definition of $\psi$  \eqref{choicepsi} and of $L$ \eqref{defopL}, we have
$$\div(\psi \mut)=- \sum_{k=0}^q\frac{ \Delta L^k (\xi)}{
\cd\theta^k}.$$
Comparing with \eqref{defnut} we thus have that 
\be\label{defmutep}   \tilde \mutt-\nu_{\theta}^t= (\id+ t \psi)\#\mut -( \mut- t \div (\psi \mut) ).\ee
Since we assume $ 2m+\gamma-4\ge 3$, we have that $\mut \in C^3$ by \eqref{bornemutcn}. We may then apply Lemma~\ref{linea} to $\mut$ and $t\psi$. The condition \eqref{condsurt0} is satisfied because it is implied by  \eqref{condsurt2}.
 We then obtain  \eqref{estmu}.

Next, we notice that $\ep_t$ is supported in $\supp  \, \xi$ and we  observe that 
\be \label{hnutmut}
\g* \( \nu_\theta^t- \mut\)  = - t \sum_{k=0}^q  \frac{1}{\theta^k}  L^k (\xi)\ee and is also supported in $\supp\, \xi$.
Since $\g*\mut+ V + \frac{1}{\theta} \log \mut=C_\theta$ by \eqref{eqmb} and by definition \eqref{defopL} and \eqref{defnut}, we deduce that 
\begin{multline}\label{hnut}
\ep_t:= \g* \nu_\theta^t+ V+t\xi+ \frac{1}{\theta}\log \nu_\theta^t-C_\theta=-
 t \sum_{k=1}^q  \frac{1}{\theta^k}  L^k (\xi)+ \frac{1}{\theta} \log \( 1+\frac{t}{\cd} \sum_{k=0}^q \frac{1}{\theta^k\mut} \Delta L^k (\xi)\)\\
 = \frac{1}{\theta} (\log (1+f)- f) + \frac{t}{ \theta^{q+1}} L^{q+1}(\xi) \end{multline}
 where 
 $$f:= \frac{t}{\cd} \sum_{k=0}^q \frac{1}{\theta^k\mut} \Delta L^k (\xi)= t \sum_{k=0}^q \frac{1}{\theta^k}  L^{k+1} (\xi),$$
 hence in view of \eqref{Lin},  if $ 2m+\gamma \ge 2q+\sigma+6$, we have
 \be \label{estsurf}
  |f|_{C^\sigma}   \le Ct  \sum_{k=0}^q \frac{1}{\theta^k} \sum_{m=2}^{2k+2+\sigma} |\xi|_{C^m} \le Ct \sum_{k=0}^q \frac{1}{\theta^k}  |\xi|_{C^{2k+2+\sigma}} .   \ee
   We now compute $\nab( \log (1+f)-f)= \nab f\(\frac{1}{1+f}-1\),$
 with \eqref{estsurf}, if $ 2m+\gamma \ge 2q+\sigma + 6$ we find
  \begin{multline}|\ep_t|_{C^\sigma} \le \frac{C}{\theta} |f|_{C^1}^\sigma \|f\|_{L^\infty}^{2-\sigma} + C\frac{t}{\theta^{q+1}}    \sum_{k=2}^{2q+2+\sigma} |\xi|_{C^{k} }\\ \le  C\frac{t^2}{\theta} \(  \sum_{k=0}^q \frac{1}{\theta^k}  |\xi|_{C^{2k+3 }}\)^{\sigma} \( \sum_{k=0}^q \frac{1}{\theta^k}  |\xi|_{C^{2k+2}}\)^{2-\sigma}  + C\frac{t}{\theta^{q+1}}   \sum_{k=2}^{2q+2+\sigma} |\xi|_{C^{k} } .
  \end{multline} Hence \eqref{estept0} and \eqref{estept} hold.\end{proof}
  
\subsection{Replacement for \eqref{laplace0}}
Instead of the exact relation \eqref{laplace0} obtained via the splitting with respect to $\mut $ and $\mutt$, we use a relation with errors obtained by splitting with respect to $\nut$ instead of $\mutt$.
Instead of \eqref{splitting}, we thus find that if \eqref{condsurt} and \eqref{condsurt2} is satisfied, using \eqref{defept}, we have (with obvious notation)
\begin{multline}
\mathcal H_N^{V_t} (X_N)=  N^2 \mathcal E^{V_t}(\nut) + N \int_{\R^\d} (\g*\nut + V_t ) d\( \sum_{i=1}^N \delta_{x_i} - N \nut\) + \F_N(X_N, \nut)
\\
=  N^2 \mathcal E^{V_t}(\nut) + N \int_{\R^\d} (- \frac{1}{\theta} \log \nut + \ep_t) d  \( \sum_{i=1}^N \delta_{x_i} - N \nut\) + \F_N(X_N, \nut)\\
= N^2 \mathcal E_\theta^{V_t}(\nut) + \F_N(X_N, \nut) - \frac{N}{\theta} \sum_{i=1}^N \log \nut (x_i)  + N \int_{\R^\d} \ep_t \,d \( \sum_{i=1}^N \delta_{x_i} - N \nut\)
\end{multline} with $\mathcal E_\theta^V $ as in \eqref{1.9}.
Inserting into the definition of $Z_{N, \beta}^{V_t}$ and using the definition of $\theta $ \eqref{relbt}, we obtain 
\begin{multline} Z_{N, \beta}^{V_t}=\exp\(- \beta N^{1+\frac2\d} \mathcal E_\theta^{V_t}(\nut) \)\\ \times 
\int_{\R^\d} \exp\(-\theta \int_{\R^\d} \ep_t \, d \( \sum_{i=1}^N \delta_{x_i} - N \nut\)
- \beta N^{\frac{2}{\d}-1}  \F_N(X_N, \nut) \) d(\nut)^{\otimes N}(X_N) .\end{multline}
Using the definitions \eqref{pdef} and \eqref{defQ} we may rewrite this as 
\be
Z_{N, \beta}^{V_t}= \exp\(- \beta N^{1+\frac2\d} \mathcal E_\theta^{V_t}(\nut) \)\K_N(\nut)  \Esp_{\mathsf{Q}_N(\nut)} \( \exp \(-\theta \int_{\R^\d} \ep_t\, d \( \sum_{i=1}^N \delta_{x_i} - N \nut\)\)\).
\ee 
Combining with \eqref{laplace0}  and \eqref{resplitz} we find
\begin{multline} \label{laplace1}
\Esp_{\PNbeta} \( e^{-t\beta N^{\frac2\d} \sum_{i=1}^N \xi(x_i)}\) 
\\= e^{ -\beta N^{1+\frac2\d} \( \mathcal E_\theta^{V_t}(\nut)- \mathcal E_\theta^V(\mut)\) }\frac{\K_N(\nut)}{\K_N(\mut)} 
 \Esp_{\mathsf{Q}_N(\nut)} \( \exp \(-\theta \int_{\R^\d} \ep_t \, d\( \sum_{i=1}^N \delta_{x_i} - N \nut\)\)\).
\end{multline}
We now focus on estimating the terms in the right-hand side.
The first constant term will be expanded explicitly in $t$ and bring out the explicit expression of the variance. The last term will be small because $\ep_t$ is small  thanks to the concentration result \eqref{loclawphi}. The ratio of partition functions $\K_N$ will  for now be estimated by the rough bound of Lemma \ref{lemcompdesk}. This yields the first bounds of Theorem \ref{firsttheo}. For the proof of the CLT the ratio of $\K$'s will be further analyzed and precisely expanded in $t$, this will be done in Section \ref{sec7}.

\subsection{Ratio of the reduced partition functions}
If \eqref{condsurt} is satisfied, applying \eqref{compdesk}, in view of \eqref{defnut} and \eqref{Lin}
we have, if $ 2m+\gamma \ge 2q+7$, 
\be\label{lok}
|\log \K_N(\nut)-\log \K_N(\mut)|\le C\beta \chi(\beta) N\ell^{\d}| t|  \sum_{k=0}^q \(\ell  \frac{|\xi|_{C^{2k+3}}}{\theta^k} +
\frac{|\xi|_{C^{2k+2}}}{\theta^k}
\).
\ee

\subsection{Estimating the leading order term}

\begin{lem}\label{lemleadord}
We have
\be\label{513} \mathcal{E}^{V_t}_\theta(\nut)- \mathcal{E}^V_\theta(\mut)- t\int_{\R^\d} \xi d\mut = - t^2  v(\xi) 
+ O\( \frac{t^3}{\theta} \int_{\R^\d}\mut \left|\sum_{k=0}^q\frac{ L^{k+1}(\xi)}{\theta^k} \right|^3\)
\ee
where 
\begin{multline}
\label{defvari}
v(\xi):= -
\frac{ 1}{2\cd}  \int_{\R^\d} \left|   \sum_{k=0}^q  \frac{1}{\theta^k} \nab L^k (\xi)\right|^2
+ \frac{1}{\cd}\int_{\R^\d}  \sum_{k=0}^q \nab \xi \cdot \frac{ \nab L^{k}( \xi) }{\theta^k} -
\frac{1}{2 \theta}\int_{\R^\d} \mut \left|\sum_{k=0}^q\frac{ L^{k+1}(\xi)}{\theta^k} \right|^2
\end{multline}
and  if $ 2m+\gamma \ge 2q+6$,
\begin{multline}\label{version2}
\left|\mathcal{E}_\theta^{V_t}(\nut)- \mathcal{E}_\theta^V(\mut)- t\int_{\R^\d} \xi d\mut\right|\\ 
 \le Ct^2 |\supp \nab \xi| \( \sum_{k=0}^q \frac{|\xi|_{C^{2k+1}}^2}{\theta^{2k}} + \frac{|\xi|_{C^1} |\xi|_{C^{2k+1}}}{\theta^k} +\frac{ |\xi|_{C^{2k+2}}^2}{\theta^{2k+1}}\) .
\end{multline}

\end{lem}
\begin{proof} We have  
\begin{multline*}
\mathcal{E}^{V_t}_\theta(\nut)- \mathcal{E}^V_\theta(\mut)\\=
\(\hal \iint \g(x-y)d\nut(x)d \nut(y)- \hal \iint \g(x-y) d\muth(x) d\muth(y)+\int V_t\, d\nut- \int V\, d\mut  \)
\\+ \frac1\theta \(\int  \nut \log \nut -\int \mut \log \mut\)
\\
= \hal \iint \g(x-y) d(\nut-\muth)(x)d(\nut-\muth)(y) + \iint \g(x-y) d(\nut-\muth)(x) d\muth(y) \\+ \int V d(\nut-\muth) + t \int \xi d\muth+ t \int \xi d(\nut-\muth)  + \frac1\theta \(\int  \nut \log \nut -\int \mut \log \mut\) \\
= \hal \iint \g(x-y) d(\nut-\muth)(x)d(\nut-\muth)(y) +\int (\g*{\mut}+ V+ \frac{1}{\theta} \log \mut) d(\nut-\muth) \\+ t\int \xi d\mut+ t\int \xi d(\nut-\mut) + \frac{1}{\theta} \int \nut(\log \nut-\log \mut) .\end{multline*}
The second term of the right-hand side  vanishes by characterization of $\mut$ in \eqref{hhe}, and we are left with 
\begin{multline*}
\mathcal{E}^{V_t}_\theta(\nut)- \mathcal{E}^V_\theta(\mut)- t\int \xi d\mut\\=
\frac{1}{2\cd} \int |\nab (\g*({\nut}-{\mut}))|^2 + t\int \xi d(\nut-\mut) + \frac{1}{2\theta} \int \mut \( \frac{\nut}{\mut}-1\)^2 
+  O\(\frac{1}{\theta} \int \(  \frac{\nut}{\mut}-1\)^3 \mut\) 
\end{multline*}
where we Taylor expanded the logarithm. 
We then use \eqref{hnutmut}, \eqref{defnut} and the definition of $L$ to see that 
$$  |\nab (\g*({\nut}-{\mut}))|^2= t^2\left|   \sum_{k=0}^q  \frac{1}{\theta^k} \nab L^k (\xi)\right|^2 $$
and $$\frac{\nu_\theta^t }{\mut}=1+ t \sum_{k=0}^q\frac{ L^{k+1}(\xi)}{\theta^k}.$$
We thus find \eqref{513}. Alternatively we can Taylor expand the log only to first order and get instead  a bound by $$
 C t^2  \( \int_{\R^\d} \left|   \sum_{k=0}^q  \frac{1}{\theta^k} \nab L^k (\xi)\right|^2+
\left|\int_{\R^\d} \sum_{k=0}^q \nab \xi \cdot \frac{ \nab L^{k}( \xi) }{\theta^k}\right| +
 \frac{1}{\theta}\int_{\R^\d} \mut \left|\sum_{k=0}^q\frac{ L^{k+1}(\xi)}{\theta^k} \right|^2\)$$ from which we deduce \eqref{version2} from \eqref{Lin}.

\end{proof}

\subsection{Estimating the last term}
We start by estimating the last expectation in the right-hand side.
We will use two different controls.
\begin{lem} We have
\be\label{laplace10}\left| \log  \Esp_{\mathsf{Q}_N(\nut)} \( \exp \(-\theta \int_{\R^\d} \ep_t \, d\( \sum_{i=1}^N \delta_{x_i} - N \nut\)\)\)
\right|
\le
C \sqrt{\chi(\beta)} \beta N^{1+\frac1\d} \ell^\d|\ep_t|_{C^1}+C \theta N\ell^\d  |\ep_t|_{C^1}^2\ee
and 
\begin{multline}\label{laplace2}\left| \log  \Esp_{\mathsf{Q}_N(\nut)} \( \exp \(-\theta \int \ep_t\, d \( \sum_{i=1}^N \delta_{x_i} - N \nut\)\) \)\right|\\
\le
C \|\ep_t\|_{L^\infty} \beta N^{\frac2\d+1} \ell^\d+  C\|\ep_t\|_{L^\infty}^2 \beta N^{1+\frac2\d} \ell^{\d-2}
\end{multline}
\end{lem}
\begin{proof}
By Proposition \ref{th3}, local laws and concentration hold for $\mathsf{Q}_N(\nu_\theta^t)$ in $\hat\Sigma$ where $\nut$ is bounded below, \eqref{loclawphi} applies and yields for any $\varphi$ such that $\|\nab \varphi\|_{L^\infty} \le N^{\frac1\d}$,
$$\left|\log \Esp_{\mathsf{Q}_N (\nu_\theta^t)} \( \exp \frac\beta{CN\ell^\d}\( \int_{\R^\d} \varphi \,d\( \sum_{i=1}^N \delta_{x_i} - N \nut\) \)^2 \)\right|\le C 
\beta \chi(\beta) N^{1-\frac2\d} \ell^{\d} \|\nab \varphi\|_{L^\infty}^2$$
We may then apply this to $\varphi=\sqrt{C}\ell^{\frac\d2} N^{\frac{1}{\d}+\hal} \sqrt{\lambda} \ep_t$. Thus,  for any $\lambda$ such that $\sqrt{\lambda C} \ell^{\frac\d2}N^{\hal} |\ep_t|_{C^1} \le 1$ (which ensures that $\|\nabla \varphi\|_{L^\infty} \le N^{1/\d}$),  using  also that 
$$\theta \int \ep_t \, d\( \sum_{i=1}^N \delta_{x_i} - N \nut\)\le \theta \lambda  \( \int \ep_t\, \( \sum_{i=1}^N \delta_{x_i} - N \nut\)\)^2 + \frac{\theta}{4 \lambda},$$
we have
$$\log \Esp_{\mathsf{Q}_N(\nut)} \( \exp\(\theta  \int \ep_t \, d \( \sum_{i=1}^N \delta_{x_i}-N\nut \)\)\)
\le C \lambda \beta  \chi(\beta) N^2\ell^{2\d}| \ep_t|_{C^1}^2 + \frac{\theta}{ 4\lambda}$$
and optimizing over $\lambda\le |\ep_t|_{C^1}^{-2}(N\ell^\d)^{-1}$ we find \eqref{laplace10}.
We next turn to proving \eqref{laplace2}. This time we bound 
$$\left| \int \ep_t \, d\( \sum_{i=1}^N \delta_{x_i} - N \nut\)\right|\le \|\ep_t\|_{L^\infty} ( \#I_\Omega+ N \ell^\d)
$$where $\#I_\Omega$ denotes the number of points in each configuration that fall in the set $\Omega$, defined as the support of $\xi$.
We can in turn bound from above 
$$\#I_\Omega \le N\int_\Omega d \nut + D(x,C\ell)$$
where $B(x, C\ell)$ is a ball that contains $Q_\ell$ and $D(x, \ell) =\int_{B(x,C \ell)}  \sum_{i=1}^N \delta_{x_i} - N d\mu$.
Arguing as before, we write 
$$\theta \|\ep_t\|_{L^\infty} D(x,C \ell)\le  \|\ep_t\|_{L^\infty}  \(D^2(x, C\ell) \beta N^{\frac2\d-1} \ell^{2-\d}  \lambda +  \frac{\theta  N\ell^{\d-2} }{4\lambda}\)
$$
and thus using \eqref{loclawpoints00}, we find, 
$$\log \Esp_{\mathsf{Q}_N(\nut)}\(\exp\(  \theta \|\ep_t\|_{L^\infty} D(x, C\ell) \) \) \le  C \|\ep_t\|_{L^\infty} \lambda \beta \chi(\beta) N\ell^\d+ \frac{\beta \|\ep_t\|_{L^\infty} N^{1+\frac2\d} \ell^{\d-2}}{4\lambda}.
$$Optimizing over $\lambda \le \|\ep_t\|_{L^\infty}^{-1}$ we find 
$$\log \Esp_{\mathsf{Q}_N(\nut)}\(\exp\( \theta \|\ep_t\|_{L^\infty} D(x, C\ell) \) \) \le  C \|\ep_t\|_{L^\infty} \sqrt{\chi(\beta)} \beta N^{1+\frac1\d} \ell^{\d-1} +  C\|\ep_t\|_{L^\infty}^2 \beta N^{1+\frac2\d} \ell^{\d-2}.$$
After observing that 
$\sqrt{\chi(\beta)} N^{-\frac1\d}\ell^{-1}\le 1$ by  \eqref{ass1} and \eqref{rhobeta}, 
the result follows. 
\end{proof}

\subsection{First bounds on the fluctuations -- proof of Theorem \ref{firsttheo} and corollaries}
We are now in a position to estimate the terms in \eqref{laplace1}. Under the conditions \eqref{condsurt}, \eqref{condsurt2}, 
inserting \eqref{lok}, \eqref{version2} and \eqref{laplace2}  into \eqref{laplace1},  we obtain that  $2m+\gamma\ge 2q+7$
and $\xi \in C^{2q+3}$, \begin{multline}\label{bornelap}
\left|\log\Esp_{\PNbeta} \(\exp\(-\beta tN^{\frac{2}{\d}} \( \sum_{i=1}^N \xi(x_i)-N\int \xi d\mut\) \)\) \right|\\
 \le   C\beta \chi(\beta) N\ell^{\d}| t|  \sum_{k=0}^q \(\ell  \frac{|\xi|_{C^{2k+3}}}{\theta^k} +
\frac{|\xi|_{C^{2k+2}}}{\theta^k}\)+ \Error_1+\Error_2
\end{multline}
with 
$$|\Error_1|\le   Ct^2 \beta N^{1+\frac2\d} |\supp \nab \xi|  \sum_{k=0}^q \(\frac{|\xi|_{C^{2k+1}}^2}{\theta^{2k}} + \frac{|\xi|_{C^1} |\xi|_{C^{2k+1}}}{\theta^k} +\frac{ |\xi|_{C^{2k+2}}^2}{\theta^{2k+1}}\) .
$$and
\begin{multline}|\Error_2|\le  
C\frac{ t^2 \beta N^{1+\frac2\d} \ell^\d}{  \theta}\(   \sum_{k=0}^{q}\frac{1}{\theta^{k}}  |\xi|_{C^{2k+2}} \)^2 +  C\frac{|t|}{\theta^{q+1}}      \beta N^{1+\frac2\d} \ell^\d  \sum_{k=2}^{2q+2} |\xi|_{C^{k} }\\+  
 C\frac{ t^4}{\theta^2}\(   \sum_{k=0}^{q}\frac{1}{\theta^{k}}  |\xi|_{C^{2k+2}} \)^4  \beta N^{1+\frac2\d} \ell^{\d-2} 
+  C\frac{t^2}{\theta^{2(q+1)}}  \( \sum_{k=2}^{2q+2} |\xi|_{C^{k} } \)^2
\beta N^{1+\frac2\d} \ell^{\d-2} .
 \end{multline}
 
Alternatively, using \eqref{laplace10} instead of \eqref{laplace2} we obtain that 
\begin{multline}\label{bornelap2} 
|\Error_2|\le 
C \sqrt{\chi(\beta)} \beta N^{1+\frac1\d}\ell^\d \(  t^2  \sum_{m=0}^{2q}\frac{1}{\theta^{m+1}} \sum_{p+k=m} |\xi|_{C^{2k+2}} |\xi|_{C^{2p+3}} + \frac{|t|}{\theta^{q+1}}  \sum_{k=2}^{2q+3} |\xi|_{C^{k} }  \) \\ + 
C \theta N\ell^\d \(  t^4  \( \sum_{m=0}^{2q}\frac{1}{\theta^{m+1}} \sum_{p+k=m}|\xi|_{C^{2k+2}} |\xi|_{C^{2p+3}}\)^2
+ \frac{t^2}{\theta^{2q+2}}\( \sum_{k=2}^{2q+3} |\xi|_{C^{k} }\)^2 \).\end{multline}
We first focus on the result requiring the least regularity for $V$ and $\xi$, which are obtained by choosing $q=0$ in \eqref{bornelap}. Then the conditions \eqref{condsurt}, \eqref{condsurt2}  reduce to 
\eqref{assmth1}.
We get that if $V\in C^{2m+\gamma}$ with $2m+\gamma\ge 7$ and $\xi \in C^{3}$
(using $\beta N^{\frac{2}{\d}} \ell^2 \ge 1 $ or $\theta \ell^2\ge 1$ to absorb some terms), 
\begin{multline}\label{bornelapq000}
\left|\log\Esp_{\PNbeta} \(\exp\(-\beta tN^{\frac{2}{\d}} \( \sum_{i=1}^N \xi(x_i)-N\int \xi d\mut\) \)\) \right|\\
\le C |t|\beta N \ell^\d\(    \chi(\beta) \ell |\xi|_{C^{3}} +   \frac{1}{\beta}    |\xi|_{C^{2} }   \)   
\\+ Ct^2  \(N \ell^\d  |\xi|_{C^2}^2+|\supp \nab \xi| \beta N \(   N^{\frac2\d}  |\xi|_{C^1}^2+ \frac{1}{\beta}|\xi|_{C^2}^2   \) \)
+  
 C  N \ell^\d  t^4    |\xi|_{C^{2}}^4 .
  \end{multline}
  This proves Theorem \ref{firsttheo}.

We now prove Corollary \ref{corod2}. The proof will be split into the cases $\beta \le 1$ and $\beta \ge 1$.
 For $\beta \le 1$, applying the result of Theorem \ref{firsttheo} with 
  $|\xi|_{C^k} \le M \ell^{-k}$ with $M \ge 1$, we find
\begin{multline}\label{bornelapq1}
\left|\log\Esp_{\PNbeta} \(\exp\(-\beta tN^{\frac{2}{\d}} \( \sum_{i=1}^N \xi(x_i)-N\int \xi d\mut\) \)\) \right|\\
 \le   C| t| \cor{M} N \ell^{\d-2}  (1+   \beta \chi(\beta) )      
+ Ct^2 \cor{M^2} N\ell^\d \(   \beta N^{\frac2\d} \ell^{-2} +  \ell^{-4}      \)
+ C t^4 \cor{M^4} N\ell^{\d-8}  \\
 \le 
 C N \ell^\d\(   M | t|\ell^{-2} ( \beta +   1      )  
+M^2 t^2    \beta N^{\frac2\d} \ell^{-2}   
+ M^4 t^4 \ell^{-8}    \)  \end{multline}
because  we can absorb 
$ \ell^{-4} $ into $\beta N^{\frac2\d} \ell^{-2}$ and $\beta\chi(\beta) $ into $1$  since $\beta \le 1$.

We then  choose \cor{$t=   \tau  (N^{\frac1\d}\ell)^{-1-\frac{\d}{2}}\ell^2 $
 and plug into \eqref{bornelapq1}.
The condition \eqref{assmth1} is then equivalent to 
$ |t|M\ell^{-2}$ small enough, i.e. 
$C |\tau | (N^{\frac1\d}\ell)^{-1-\frac{\d}{2}}  <1$.
Using that $\d=2$ we then find  that 
\begin{multline*} \left|\log\Esp_{\PNbeta} \(\exp\( \tau  \beta  |\Fluct(\xi)| \)\) \right|\\
\le C\(
|\tau|  M(1+ \beta)        + \tau^2 M^2  + M^4\tau^4 
   (N^{\frac1\d}\ell)^{-4-\d } +   C M^2  (N^{\frac{1}{\d}} \ell)^{-2} \tau^2 
     \).\end{multline*} This concludes the proof for $\beta \le 1$.}
         
\cor{For $\beta \ge 1$, we choose instead  $t= \tau  (N^{\frac1\d}\ell)^{-1-\frac{\d}{2}}\ell^2 \beta^{-1}$. The condition \eqref{assmth1} is then equivalent to 
$  C |\tau | (N^{\frac1\d}\ell)^{-1-\frac{\d}{2}}  \beta^{-1}<1$. With the same reasoning, we then find that 
\begin{multline*} \left|\log\Esp_{\PNbeta} \(\exp\( \tau   |\Fluct(\xi)| \)\) \right|\\
\le C\(
|\tau|  M        + \tau^2 M^2 \beta^{-1} + M^4\tau^4 
   (N^{\frac1\d}\ell)^{-4-\d }  \beta^{-4}+   C   (N^{\frac{1}{\d}} \ell)^{-2} M^2 \tau^2 
  \beta^{-1}   \).\end{multline*}
 and obtain the desired result.}

Choosing $t=\pm \tau \ell^2 ( (1+\beta)  N \ell^\d)^{-1}$
we get the following estimate in dimension $\d \ge 3$. A stronger one will be obtained below, but assuming more regularity on $\xi$.

\begin{coro}\label{5.5}Let $\d \ge 3$. Assume  $V\in C^7$, \eqref{assumpV2}--\eqref{assumpV4} hold, and $\xi\in C^{3}$, $\supp\, \xi \subset B(x, \ell) \subset \hat \Sigma$, for some $\ell$ satisfying \eqref{ass1}
 Assume  $|\xi|_{C^k} \le M \ell^{-k}$  for all $k \le 3$. Then for all $|\tau|<C^{-1}M^{-1}(1+\beta) N\ell^\d$ we have
\be\label{resultcoro55} \left|\log\Esp_{\PNbeta} \(\exp\( \tau \frac{\beta}{\beta+1} (N^{\frac1\d}\ell)^{2-\d} |\Fluct(\xi)| \)\) \right|
\le C (1+\tau^4 M^4)  \ee where $C$ depends only on $V$ and $\d$.
\end{coro}
Again we note that since $N\ell^\d \ge \rb^\d\ge 1$ we can apply this to any $|\tau|<C^{-1}$.

\begin{proof}
We choose the announced $t$ and plug into \eqref{bornelapq1}. The condition \eqref{assmth1} is here equivalent to 
$C |\tau|  ( (1+\beta)  N \ell^\d)^{-1} <1$.
We find that  the left hand side in \eqref{resultcoro55} is bounded by
$$|\tau|  M + \tau^2 M^2 \frac{ \beta (N^{\frac{1}{\d}}\ell )^{2-\d}  }{(1+\beta)^2 } 
 +\frac{\beta}{(1+\beta)^4} \frac{M^4\tau^4}{  (\ell N^{\frac1\d})^{3\d}} \le C (1+\tau^4 M^4). $$
 Since $(N^{1/\d}\ell)^{-4-\d}\le \rb^{-4-\d}\le \min(1,\beta^{2+\d})$ by \eqref{ass1}, we find the announced result.
\end{proof}


We now turn to an estimate that can be obtained by assuming more regularity on $\xi$, starting from \eqref{bornelap2}, and prove Corollary \ref{coro22}. Such an estimate will be more precise when $\beta $ is small.
Since  $V\in C^\infty$, $\xi\in C^\infty$, we can take $q=\infty$. The condition \eqref{assmth1} then becomes  $|t|CM\ell^{-2} <1$.
Using  that $\theta \ell^2 >1$ by \eqref{ass1} we can sum the series, which yields
\begin{multline}\label{bornelap212}
\left|\log\Esp_{\PNbeta} \(\exp\(-\beta tN^{\frac{2}{\d}}\Fluct(\xi) \)\) \right|\\
 \le   C |t| M \beta \chi(\beta) N\ell^{\d-2} +C t^2 M^2 \beta N^{1+\frac2\d}\ell^{\d   -2}  +C \sqrt{\chi(\beta)} \beta N^{1+\frac1\d}\ell^\d   t^2 M^2  \frac{\ell^{-5}}{\theta}     + 
C N\ell^\d    t^4M^4 \frac{ \ell^{-10} }{\theta} 
 \\
\le C |t| M\beta \chi(\beta) N\ell^{\d-2} +C t^2 M^2 \beta N^{1+\frac2\d}\ell^{\d-2}  + 
C  N\ell^\d t^4 M^4 \ell^{-8}   
  ,\end{multline} where the third term was absorbable by the second.

We now optimize over $t$.
When $\beta \ge (\ell N^{\frac1\d} )^{2-\d}$, it leads to choosing $t =\tau(\chi(\beta)\beta)^{-1} N^{-1}\ell^{2-\d}$.
The condition \eqref{assmth1} then becomes 
$|\tau| M\beta^{-1}(N\ell^\d)^{-1}$ small enough. 
Using that $\theta \ell^2 \ge 1$, we find
\begin{multline}\label{bornelap22}
\left|\log\Esp_{\PNbeta} \(\exp\( \tau  ( N^{\frac{1}{\d} }\ell)^{2-\d}|\Fluct(\xi)| \)\) \right|\\
 \le C |\tau| M   +C  M^2\tau^2       \beta^{-1}  (N^{\frac1\d}\ell)^{-\d+2} + 
C  \(  \tau^4     M^4   \beta^{-4} N^{-3}\ell^{-3\d}    \)\\
\le  C|\tau |M + C \frac{ \tau^2 M^2  }{\beta (N^{\frac1\d} \ell)^{\d-2} } + C\frac{\tau^4 M^4}{  \beta^4 (N^{\frac{1}{\d}} \ell)^{3\d} }  
 \le C |\tau| M +C \tau^4 M^4,
  \end{multline}
where we have used that $\beta \ge (N^{\frac{1}{\d}}\ell)^{2-\d} $. (If $\d=2$ then this implies that $\beta \ge 1$ which obviously suffices to conclude. If $\d\ge 3$ then by \eqref{ass1}  $(N^{\frac1\d} \ell)^{3\d} \ge \rb^{3\d} \ge \beta^{-4}$, which also suffices.)

When $\beta\le   (\ell N^{\frac1\d} )^{2-\d}$, it leads to choosing $t= \tau \chi(\beta)^{-1} \ell^{1-\frac\d2} N^{-\hal -\frac1\d}\beta^{-\hal}$.
The condition \eqref{assmth1} becomes $C M|\tau|  <\beta^\hal (N^{\frac1\d}\ell)^{1+\frac\d2}$ (again satisfied as soon as $CM|\tau|<1$) and
we find 

\begin{multline}\label{bornelap23}
\left|\log\Esp_{\PNbeta} \(\exp\(\tau \beta^{\hal}  \(\ell   N^{\frac{1}{\d}})^{1-\frac{\d}{2}}
 \Fluct(\xi) \)\) \)\right|\\
  \le C  \tau  M \beta^{\hal}   N^{\hal   -\frac1\d} \ell^{\frac{\d}{2}-1}  +C  M^2 \tau^2    + 
C  \(      \tau^4    \ell^{-4-\d} N^{-1 -\frac4\d}\beta^{-2} M^4 \)\\
  \le    C\frac{\tau M \beta^{\hal}}{ (N^{\frac1\d}\ell)^{1-\frac\d2}} + C M^2 \tau^2 + \frac{C \tau^4 M^4}{\beta^2 (N^{\frac1\d}\ell)^{\d+4}}
\le  
   C \tau M + C\tau^4 M^4   \end{multline}
where we used that $\beta \le (N^{\frac{1}{\d}}\ell)^{2-\d} $, and again by \eqref{ass1} $N^{\frac1\d}\ell \ge \beta^{-\hal}.$

\section{Free energy expansions for nonuniform densities}\label{varying}
We now have all the ingredients at hand to complete the proof of Theorem \ref{thglob} and  Proposition \ref{th1b}, the free energy expansion. The reader interested in Theorems \ref{th2} and \ref{th5} may skip the details of this section, assuming the result of Proposition \ref{th1b}. 

From \cite{as} we already have the expansion of $\log \K_N(\car_R,1)$, for constant density $1$ (see \eqref{1.26}), then for all constant densities by a simple rescaling \eqref{scalingdeL}. The case of a nonuniform density is treated by transporting the nonuniform density to its average value on a small cube of size $R$ and using Lemma \ref{derivK} to estimate the error. Then the almost additivity  result over cubes (Proposition \ref{proadd}) allows to get an expansion over any domain. The last part is to optimize over $R$, the size of the cubes over which we partition.

Combining Lemma \ref{lemcompdeskcube} with the known expansion for uniform densities, 
this leads to  the following expansion of the free energy in the varying case.
 \begin{lem}\label{resinter} 
 Assume $ \ell $ satisfies \eqref{ass1}. Let $Q_\ell$ be a hyperrectangle of sidelengths in $(\ell, 2\ell)$.
Let  $\mu$ be a $C^{1}$ density bounded  above and below  by positive constants in $Q_\ell$,
and assume $\mn= N  \int_{Q_\ell} \mu $ is an integer.
  We have 
 \begin{multline}\label{main48}
\log \K_N(Q_\ell,\mu) =- \beta N  \int_{Q_\ell} \mu^{2-\frac2\d} \fd(\beta \mu^{1-\frac{2}{\d}} ) + \frac{\beta}{4}N\(\int_{Q_\ell} \mu \log \mu\) \indic_{\d=2} - \(\frac{\beta}{4}\mn \log N \) \indic_{\d=2}\\  +O\(\beta \chi(\beta) \rb N^{1-\frac1\d} \ell^{\d-1} + \beta^{1-\frac1\d}\chi(\beta)^{1-\frac1\d}
\ell^{\d-1} \(\log \frac{ \ell N^{1/\d}}{\rb} \)^{\frac1\d} \) 
\\+ O\( \beta N    \ell^{\d}\(    \chi(\beta)  \ell |\mu|_{C^1} + \ell^2 |\mu|_{C^1}^2    \indic_{\d=2}\)\)
\end{multline}
with $C$ depending only on $\d$ and the upper and lower bounds for $\mu$.
  \end{lem}
  
 \begin{proof} Let $\bar \mu$ denote the average of $\mu $ on $Q_\ell$. 
 We know from \cite{as} an expansion for $\log \K_N(Q_\ell)$ for constant densities, see \eqref{scalingdeL} and \eqref{1.26}.  Scaling these formulae  properly and  inserting  into \eqref{compdeskcube} applied with $\mu_0= \bar \mu$ and $\mu_1=\mu$, we find 
\begin{multline}\label{licasd0}
\log \K_N(Q_\ell, \mu) = N |Q_\ell|\(-\beta  \bar \mu^{2-\frac2\d} \fd(\beta \bar \mu^{ 1- \frac{2}{\d}}) 
- \frac{1}{4}\beta( \bar\mu\log \bar \mu) \indic_{\d=2}\) +\(\frac{\beta}{4}\mn \log N \) \indic_{\d=2} \\
+O\(\beta \chi(\beta) \rb N^{1-\frac1\d} \ell^{\d-1} + \beta^{1-\frac1\d}\chi(\beta)^{1-\frac1\d}
\ell^{\d-1} \(\log \frac{ \ell N^{1/\d}}{\rb}\)^{\frac1\d} \)  + O\(     N \beta\chi(\beta)\ell^{\d+1}  |\mu|_{C^1}\),
\end{multline}where the $O$ depend only on $\d$ and the upper and lower bounds for $\mu$.

If $d=3 $ we write using a Taylor expansion that 
\begin{equation*}
\fd(\beta \mu^{1-\frac2\d})= \fd(\beta \bar \mu^{1-\frac2\d}) 
+O\(\beta \|\fd'\|_{\mu, Q_\ell} \ell^\d \|\mu - \bar \mu\|_{L^\infty(Q_\ell)}\).\end{equation*}
Integrating against $\mu^{2-\frac2\d}$,  using 
$\int_{Q_\ell}\mu - \bar \mu=0$, we find
\begin{equation}
-\beta  |Q_\ell| \bar \mu^{2-\frac2\d} \fd(\beta \bar \mu^{ 1- \frac{2}{\d}})  
=
- \beta  \int_{Q_\ell} \mu^{2-\frac2\d} \fd(\beta \mu^{1-\frac{2}{\d}} )
+O\(\beta^2 \|\fd'\|_{\mu, Q_\ell} \ell^\d \|\mu - \bar \mu\|_{L^\infty(Q_\ell)} \).\end{equation} In dimension 2, we may write instead 
 \begin{multline}
-\beta  |Q_\ell| \bar \mu^{2-\frac2\d} \fd(\beta \bar \mu^{ 1- \frac{2}{\d}})  
- \frac{\beta}{4}|Q_\ell| ( \bar\mu \log \bar \mu) \indic_{\d=2}
\\=
- \beta  \int_{Q_\ell} \mu^{2-\frac2\d} \fd(\beta \mu^{1-\frac{2}{\d}} ) - \frac{\beta}{4}\(\int_{Q_\ell} \mu \log \mu\) \indic_{\d=2}
+O\( \beta \ell^\d \|\mu-\bar \mu\|_{L^\infty(Q_\ell)}^2 \).\end{multline}

Using that $\|\mu- \bar \mu\|_{L^\infty(Q_\ell)}\le  \ell|\mu|_{C^1(Q_\ell)}$, \eqref{bornesurfp}, and inserting into \eqref{licasd0}, we obtain \eqref{main48}.

\end{proof}

By subdividing a cube and using the almost additivity of the free energy, we may improve the error term in the previous expansion.

Assume that $Q_R$  is split into $p$ hyperrectangles $Q_i$ with $N\int_{Q_i} \mu=\mn_i$ an integer, and $Q_i$ of sidelengths in $(\ell, 2\ell)$, $\ell \gg N^{-1/\d} \rb$.
We may always find such a splitting arguing as in \cite[Lemma 3.2]{as}, itself relying on \cite[Lemma 7.5]{ss2}.  
It consists in first splitting $Q_R$ into parallel strips of width close to $\ell$. Because $\mu$ is bounded below and $N$ is large we may modify the width of the strip slightly until the integral of $\mu$ in that strip is in $\frac{1}{N}\mathbb{N}$. We then iterate by splitting each strip into lower dimensional strips in a tranverse direction, so that the integral in each piece is in $\frac{1}{N}\mathbb{N}$. Repeating this $\d$ times we obtain hyperrectangles with quantized mass.


Using Proposition \ref{proadd}, in particular \eqref{subad4}, we have
 \begin{multline*}
\log \K_N(Q_R,\mu) = \sum_{i=1}^p \log \K_N(Q_i, \mu)\\
+  O\( p\beta \chi(\beta) N \ell^{\d}\( \rb\ell^{-1}  N^{-\frac1\d}   +  \beta^{-\frac1\d}\chi(\beta)^{-\frac1\d} \ell^{-1} N^{-\frac1\d}\(\log  \frac{\ell N^{\frac1\d}}{\rb} \)^{\frac1\d} \) \).\end{multline*}
Inserting \eqref{main48}  yields,
\begin{multline}\label{main500}
\log \K_N(Q_R,\mu) =  -
 \beta N  \int_{Q_R} \mu^{2-\frac2\d} \fd(\beta \mu^{1-\frac{2}{\d}} ) -  \frac{\beta}{4}\(  N \int_{Q_R} \mu \log \mu\) \indic_{\d=2}   + \(\frac{\beta}{4}\mn \log N \) \indic_{\d=2}
\\
+  O\( \frac{R^\d}{\ell^\d}  \beta \chi(\beta) N \ell^{\d}\(  \rb\ell^{-1} N^{-\frac1\d}   +  \beta^{-\frac1\d}\chi(\beta)^{-\frac1\d} \ell^{-1} N^{-\frac1\d}\(\log \frac{\ell N^{\frac1\d}}{\rb} \)^{\frac1\d}\) \)
\\ + O\( \frac{R^\d}{\ell^\d}   \beta N    \ell^{\d}\(    \chi(\beta)  \ell |\mu|_{C^1(Q_R)} + \ell^2 |\mu|_{C^1(Q_R)}^2    \indic_{\d=2} \)\).
\end{multline}
We also choose $\ell <|\mu|_{C^1(Q_R)}^{-1}$, so that the $\ell^2 |\mu|_{C^1}^2$ term can be absorbed into  the previous one.
 We are left with choosing $ \ell \le \min (R, |\mu|_{C^1(Q_R)}^{-1})$  minimizing 
$$
\rb(N^{\frac1\d}\ell)^{-1}    +  \beta^{-\frac1\d}\chi(\beta)^{-\frac1\d} (N^{\frac1\d}\ell)^{-1} \(\log \frac{\ell N^{\frac1\d}}{\rb} \)^{\frac1\d} +
 O\(  \ell |\mu|_{C^1} \).
$$
We next show that we can make this  $o(1)$ as $N \to \infty$.
We will use  the notation $r=\ell N^{\frac1\d}$ and $X=  |\mu|_{C^1} $.

We also need to enforce the condition \eqref{rrb} so in total the constraints on $r $ are  
\be \label{1r}
 \rb+ \( \frac{1}{\beta \chi(\beta)} \log \frac{r^{\d-1}}{\rb^{\d-1}}\)^{\frac1\d}\le r\le N^{\frac1\d}  \min \( R, |\mu|_{C^1(Q_R)}^{-1}\) .\ee
 We next find the optimal value.
\begin{lem}
Assume  $R\le C  $ and
\be\label{1R}
N^{\frac1\d} R \ge \rb + \( \frac{1}{\beta \chi(\beta)} \log \frac{R^{\d-1}}{\rb^{\d-1}}\)^{\frac1\d}\ee then 
\begin{multline}\label{723} \min_{
 r \text{ satisfies } \eqref{1r} }   \( \rb  r^{-1} + ( \beta\chi(\beta))^{-\frac1\d} r^{-1}  \( \log \frac{r}{\rb}\)^{\frac1\d} 
+  r N^{-\frac1\d} X \)\\ 
\le C \max\Big(      ( \rb X  N^{-\frac1{\d}})^{\hal} \(1+ 
\(\log\frac{ N^{\frac1\d}}{\rb X} \)^{\frac1\d} \),  \, 
\rb  R^{-1}N^{-\frac1\d}  \(1+  \(\log \frac{ N^{\frac1\d}}{\rb}\)^{\frac1\d} \) ,\\  \rb N^{-\frac1\d}   |\mu|_{C^1}  \(1+ 
 \(\log \frac{N^{\frac1\d}  }{ \rb |\mu|_{C^1}  }\)^{\frac1\d} \)   \Big) 
 \end{multline}
     where $C$ depends on the constants above.
\end{lem}
\begin{proof}
If $\sqrt{\frac{\rb  N^{\frac1{\d}}} {X}}\ge \min \(  N^{\frac1\d} R, N^{\frac1\d}   |\mu|_{C^1(Q_R)}^{-1}\)$, we take     $r=\min \( N^{\frac1\d} R,N^{\frac1\d}  |\mu|_{C^1(Q_R)}^{-1}\)  $.
We then find the min is less than
 \begin{multline*}\max \Big[ \rb  R^{-1} N^{-\frac1\d}  + ( \beta\chi(\beta))^{-\frac1\d} R^{-1} N^{-\frac1\d} \( \log \frac{R N^{\frac1\d}}{\rb}\)^{\frac1\d} 
+     X ,\\     \rb N^{-\frac1\d}  |\mu|_{C^1}  +  (\beta \chi(\beta))^{-\frac1\d}  N^{-\frac1\d} C (1+|\mu|_{C^1}) \( \log \frac{N^{\frac1\d}  }{\rb |\mu|_{C^1}}\)^{\frac1\d} +   |\mu|_{C^1}^{-1}   X  \Big]  \\
\le C\max \Big[ \rb  R^{-1}N^{-\frac1\d}  + ( \beta\chi(\beta))^{-\frac1\d} R^{-1}N^{-\frac1\d}   \( \log \frac{R N^{\frac1\d}}{\rb}\)^{\frac1\d}+X ,  \\  \rb N^{-\frac1\d}  |\mu|_{C^1}  +  (\beta \chi(\beta))^{-\frac1\d}  N^{-\frac1\d}  |\mu|_{C^1} \( \log \frac{N^{\frac1\d} }{\rb  |\mu|_{C^1} }\)^{\frac1\d} +   |\mu|_{C^1}^{-1}    X    \Big].  \end{multline*}
We note here that we are able to bound $X$ from above thanks to  the condition
$\sqrt{\frac{\rb  N^{\frac1{\d}}} {X}}\ge N^{\frac1\d} R$  respectively $ \sqrt{\frac{\rb  N^{\frac1{\d}}} {X}}\ge N^{\frac1\d} |\mu|_{C^1}^{-1} $.
If on the other hand $$\sqrt{\frac{\rb  N^{\frac1{\d}}} {X}}\le \min\( N^{\frac1\d} R, N^{\frac1\d}   |\mu|_{C^1(Q_R)}^{-1}\)$$ we take that value for $r$  (we may check it always satisfies \eqref{1r}) and find the min is less than 
\be \label{dessus} C N^{-\frac1{2\d}}\sqrt{X}  \sqrt{\rb}\( 1+ \frac{(\beta\chi(\beta))^{-\frac1\d}}{\rb}
\(\log \frac{N^{\frac1\d}  }{\rb X}  \)^{\frac1\d}  \).\ee
We also observe that by definition \eqref{rhobeta} we always have 
$\frac{(\beta\chi(\beta))^{-\frac1\d}}{\rb}\le 1$.
It follows that \eqref{723} holds.

\end{proof}
Choosing this optimal $\ell$ as a subdivision size,
inserting this into \eqref{main500}, and rephrasing in terms of the variable $\ell$ instead of $R$, we obtain  the final result.
\begin{prop}[Free energy expansion for general density in a hyperrectangle]\label{coro24}
 Let $\ell$ satisfy \eqref{1R}.  Let $Q_\ell $ be a hyperrectangle of sidelengths in $(\ell, 2\ell)$. Let $\mu$ be a  $C^{1}$  density bounded above and below by positive constants in $Q_\ell$,  and assume $N\int_{Q_\ell} \mu=\mn$ is an integer.  Then,
 \begin{multline}\label{expzcasgb}
\log \K_N(Q_\ell, \mu)=  
  - \beta N \int_{Q_\ell} \mu^{2-\frac2\d} \fd(\beta \mu^{1-\frac{2}{\d}} ) - \frac{\beta}{4}N\(
\int_{Q_\ell}\mu \log \mu\) \indic_{\d=2} + \(\frac{\beta}{4}\mn \log N \) \indic_{\d=2}\\+ O\(\beta \chi(\beta) N\ell^\d \mathcal{R}(N, \ell, \mu)\),
\end{multline}
where  
\be\label{defRRR1}\mathcal{R}(N, \ell, \mu):=
\max\(x(1+|\log x|), (y^\hal+y) (1+|\log y|^{\frac1\d})  \)\ee
 after setting
\be x:= \frac{\rb}{\ell N^{\frac1\d}} ,\quad  y:=  \frac{\rb |\mu|_{C^1} }{N^{\frac1\d} }, \ee
and the  $O$ depend  only on $\d$ and the upper and lower bounds for $\mu$.
\end{prop}
What is useful here is that we get an explicit  error rate. The quantity $x$ is small by \eqref{1R}, the estimate is interesting when  $y$  is small too.

We now conclude
\begin{prop}[Relative expansion, local version] \label{th1b}
Let $\mu$ and $\tilde\mu$ be two densities in $C^{1}$ coinciding outside $Q_\ell$ a hyperrectangle included in $\hat \Sigma$  of sidelengths in $(\ell, 2\ell)$ with $\ell$ satisfying \eqref{ass1}, and  bounded above and below by positive constants  in $Q_\ell$. Assume $N\int_{Q_\ell} \mu=N\int_{Q_\ell} \tilde \mu=\mn$ is an integer. We have
\begin{multline}
\log \K_N( \mu)- \log \K_N(\tilde \mu)
= 
  - \beta N  \int_{Q_\ell} \mu^{2-\frac2\d} \fd(\beta \mu^{1-\frac{2}{\d}} ) -\frac{\beta}{4}N\(
\int_{Q_\ell}\mu \log \mu \) \indic_{\d=2}
  \\+ \beta N \int_{Q_\ell} \tilde \mu^{2-\frac2\d} \fd(\beta\tilde \mu^{1-\frac{2}{\d}} ) +\frac{\beta}{4}N\(
\int_{Q_\ell}\tilde \mu \log  \tilde \mu\)  \indic_{\d=2}
\\+O\( \beta \chi(\beta) N\ell^\d\(\mathcal R(N, \ell, \mu)+ \mathcal R (N, \ell, \tilde \mu)\)\)
\end{multline}
where $\mathcal R$ is as in Proposition \ref{coro24}, and the $O$ depends only on $\d$ and the upper and lower bounds for $\mu$ and $\tilde \mu$ in $Q_\ell$.
\end{prop}
\begin{proof}
We may apply \eqref{subad3} to both $\mu $ and $\tilde \mu$ and subtract the obtained relations  to get that 
\begin{multline*}
\log \K_N(\mu)- \log \K_N(\tilde \mu)
=\log \K_N(Q_\ell,\mu)- \log \K_N(Q_\ell,\tilde \mu)\\
+O  \(    N^{1-\frac1\d}   \beta \ell^{\d-1}     \rb   \chi(\beta) +N^{1-\frac1\d}\beta^{1-\frac1\d} \chi(\beta)^{1-\frac1\d} \(\log  \frac{\ell N^{\frac1\d}}{\rb}\)^{\frac1\d}   \ell^{\d-1}  \).
\end{multline*}
Inserting the result of \eqref{expzcasgb} applied to $\mu$ and $\tilde \mu$, we deduce 
\begin{multline*}
\log \K_N(\mu)- \log \K_N(\tilde \mu)=  - \beta N \int_{Q} \mu^{2-\frac2\d} \fd(\beta \mu^{1-\frac{2}{\d}} ) -\frac{\beta}{4}\( N
\int_{Q}\mu \log \mu \) \indic_{\d=2}
  \\+ \beta N \int_{Q} \tilde \mu^{2-\frac2\d} \fd(\beta\tilde \mu^{1-\frac{2}{\d}} ) +\frac{\beta}{4}N
  \(
\int_{Q}\tilde \mu \log  \tilde \mu\) \indic_{\d=2}
\\+O\( \beta \chi(\beta) N\ell^\d\(\mathcal R(N, \ell, \mu)+ \mathcal R (N, \ell, \tilde \mu)+     N^{-\frac1\d} \ell^{-1}     \rb +\beta^{-\frac1\d} \chi(\beta)^{-\frac1\d} \(\log  \frac{\ell N^{\frac1\d}}{\rb}\)^{\frac1\d}   \ell^{-1}N^{-\frac1\d} \) \).\end{multline*}
Using again that $\beta \chi(\beta)^{-\frac1\d}\le \rb$ by \eqref{rhobeta}, by definition of $x$ we see that we may absorb the last error terms into $\mathcal R$.\end{proof}

We now turn to the more precise version of Theorem \ref{thglob}.

\begin{manualtheorem}{2}[More precise version]\label{thglob2}
 Assume $\d \ge 2$.
Assume $V\in C^5$ satifies \eqref{assumpV1}-- \eqref{assumpV4}.  
We have
\begin{multline}
\label{expvar}
 \log \ZNbeta=-\beta N^{1+\frac{2}{\d}}\mathcal E_\theta^V(\mub) +\frac{\beta}{4} (N\log N) \indic_{\d=2} -N  \frac{\beta}{4} \( \int_{\R^\d} \mub \log \mub\)
  \indic_{\d=2} \\
   + N \beta \int_{\R^\d} \mub^{2-\frac2\d} \mf(\beta \mub^{1-\frac{2}{\d}} )
   +  O\Big(\beta\chi(\beta) N \( d_0 (1+(\log N)\indic_{\d=2} ) + \mathcal{R}\(N, d_0 (1+(\log N)\indic_{\d=2}), \mut  \)\) \Big)  
   \end{multline}  
where $\mathcal{R}$ is as above for the norms of $\mut$ in $\hat{\Sigma} $,
and the $O$ depends only on $\d$, an upper bound for $\mub$ and a lower bound for $\mub $ in $\hat{\Sigma}$.
\end{manualtheorem}
\begin{proof}We take $m=2$ and $\gamma=1$ in the introduction, that is $V\in C^5$. This ensures by \eqref{bornemutcn} that $\mut$ is uniformly bounded in $C^1(\hat \Sigma)$.

We partition   $\hat \Sigma$ into hyperrectangles $Q_i$ of sidelengths  in $(N^{-\frac1\d} r, 2N^{-\frac1\d}  r)$ where 
 $r$ is the minimizer in the right-hand side of \eqref{723} for the choice $R=d_0(1+M(\log N) \indic_{\d=2}) $, such that $N\int_{Q_i}\mut=\mn_i$ is an integer. Again, this can be done as in \cite[Lemma 7.5]{ss2}. We keep only the hyperrectangles that are inside $\hat\Sigma$. This way the local laws are satisfied in $U:=\cup_i Q_i$ and \eqref{subad3} applies.
By \eqref{intromutsc}, \eqref{d0min}, definition of $\hat\Sigma$ \eqref{defocs}  and choice of $R$,  we have \be\label{bcompmt}
\mut\(U^c\) \le  \frac{C}{\sqrt\theta} + C d_0 + CR  \le C R.\ee 

We apply \eqref{subad3} to $\mut$ and combine it with the result of Proposition \ref{coro24} to obtain
\begin{multline}\label{preres}
\log \K_N(\R^\d, \mut)= -\beta N \int_{\cup_i Q_i}  \mut^{2-\frac2\d} \fd(\beta \mut^{1-\frac{2}{\d}} ) - \frac{\beta}{4} N\(
\int_{\cup_i Q_i}\mut \log \mut\) \indic_{\d=2} \\   - \frac{\beta}{4} N \mut(U) (\log N) \indic_{\d=2}+\log \K_N(\R^\d\backslash U, \mut)
+ O\(\beta\chi(\beta) N |U|\mathcal{R}\(N, R, \mut\) \) 
\end{multline}
where again we can absorb the errors in \eqref{subad3} into the $\mathcal R$.
To bound $\log \K_N(\R^\d \backslash U, \mut)
$ we use  \eqref{bcompmt} and a bound proved in \cite[Proposition 3.8]{as} combined with \cite[Lemma 3.7]{as} (after rescaling the coordinates by a $N^{1/\d}$ factor)
\begin{multline}\label{453}\left|\log \K_N(\R^\d \backslash U, \mut)- \frac{\beta}{4} N\mu(U^c)( \log N) \indic_{\d=2}
\right|
\\ \le C\begin{cases}
\beta N \mut\(U^c\)+ \beta N^{1-\frac1\d}\min(\beta^{\frac1{\d-2}}, 1) & \text{if} \ \d \ge 3\\
\beta \chi(\beta) N \mut\(U^c\) & \text{if} \ \d=2\end{cases}\end{multline}
and if $\d=2$ we need to have 
$$\mut( U^c\cap (\hat\Sigma)^c) \le C \frac{ \mut\(U^c\)   }{\log N} .$$
This is ensured by the fact that $\mut$ is bounded below in $\hat \Sigma\cap U^c$ and the definition \eqref{defocs} hence 
$\mut(U^c) \ge \frac{R}{C}$
while $\mut(\hat \Sigma^c) \le  C d_0  $  as seen in \eqref{d0min}, 
so the desired condition follows by definition of $R$ (if $M$ is chosen large enough).

It remains to bound 
$$ -\beta N \int_{U^c}  \mut^{2-\frac2\d} \fd(\beta \mut^{1-\frac{2}{\d}} ) -\frac{\beta}{4}  N
\(\int_{U^c}\mut \log \mut \) \indic_{\d=2} .$$
In dimension $\d \ge 3$ we use that $\fd$ is bounded in view of \eqref{bornesurf} and $\mut$ is bounded to bound all this by $CN \beta\int_{U^c} \mut\le C N\beta(R+  \theta^{-1/2}) \le C N\beta R $ by \eqref{intromutsc} and \eqref{d0min}.

In dimension $\d=2$ we bound $\int_{U^c}\mut \log \mut$ by  $C(R+ \theta^{-\hal} )\le C R$ in view of 
\eqref{intromutsc}.
We conclude that 
\begin{multline*}-\beta N \int_{\cup_i Q_i}  \mut^{2-\frac2\d} \fd(\beta \mut^{1-\frac{2}{\d}} ) - \frac{\beta}{4} N
\( \int_{\cup_i Q_i}\mut \log \mut\) \indic_{\d=2} \\= -\beta N \int_{\R^\d}  \mut^{2-\frac2\d} \fd(\beta \mut^{1-\frac{2}{\d}} ) -\frac{\beta}{4}N\(
\int_{\R^\d}\mut \log \mut\) \indic_{\d=2} + O\(  CN\beta \chi(\beta)  R\).\end{multline*}
Inserting this and \eqref{453} into \eqref{preres} we obtain the result of Theorem \ref{thglob2}.
\end{proof}

\section{Proof of the CLT}\label{sec7}

\subsection{Comparing partition functions}Let us denote 
\be \mathcal{Z}(\beta, \mu)= -\beta \int_{\R^\d} \mu^{2-\frac2\d} \fd(\beta \mu^{1-\frac2\d}) - \frac{\beta}{4}\( \int_{\R^\d}\mu \log \mu\) \indic_{\d=2}.\ee

\begin{lem} 
Let $\mu_0$ be a probability density. 
Let  $\psi\in C^1$ be supported in a cube $Q_\ell$ of sidelength $\ell$ included in a set where  $\mu_0$ is bounded above and below by positive constants, and let $\mu_t:= (\id + t\psi)\# \mu_0$.   If $\d \ge 3$, assume \eqref{assfs} relatively to $\mu_t$ in $Q_\ell$ for all $t $ small enough.
 Let us  denote $\mathcal B_1(\beta, \mu_0, \psi)$ the  derivative at $t=0$ of  the function $\mathcal{Z}(\beta , \mu_t)$.
We have 
\be\label{tayexp}
\mathcal{Z}(\beta,  \mu_t) - \mathcal{Z}(\beta, \mu_0)= \mathcal B_1(\beta, \mu_0, \psi) + O\(t^2   \beta N\ell^{\d}  |\psi|_{C^{1}}^{ 2}\)\ee
and 
\be \label{borneB}
|\mathcal B_1(\beta, \mu_0, \psi)|\le C \beta \ell^{\d}   |\psi|_{C^{1}},\ee
for some constant $C>0$ depending on $\d$ and the upper and lower bounds for $\mu_0$ in $Q_\ell$.
\end{lem}
\begin{proof}

Denoting $\Phi_t=\id + t\psi$,  we may  write 
\begin{equation*}
\int_{\R^\d} \beta \mu_t^{2-\frac{2}{\d}}\fd(\beta  \mu_t^{1-\frac{2}{\d}})=
\int_{\R^\d} \beta  \mu_t^{1-\frac{2}{\d}}\fd(\beta \mu_t^{1-\frac{2}{\d}}) \Phi_t\# \mu_0
= \int_{\R^\d}\beta \(  \mu_t\circ \Phi_t\)^{1-\frac{2}{\d}} \fd\(\beta ( \mu_t\circ \Phi_t)^{1-\frac2\d}\) d\mu_0.
\end{equation*}
Next we recall that by definition of the push forward we have 
$$ \mu_t\circ \Phi_t= \frac{\mu_0}{\det(\id+ tD\psi)}$$
 hence we may bound 
\be \left|\frac{d^j}{dt^j}  \mu_t \circ \Phi_t\right|\le C\|\mu_0\|_{L^\infty} |\psi|_{C^1}^j \ee
Let us first assume $\d \ge 3.$ Setting $g(x)= \beta x^{1-\frac{2}{\d}} \fd(\beta x^{1-\frac{2}{\d}})$, we 
have 
\begin{equation*}
\frac{d}{dt} g( \mu_t \circ \Phi_t)= g'(\mu_t \circ \Phi_t)  \frac{d}{dt}  \mu_t \circ \Phi_t
\end{equation*}
and 
\begin{equation*}
\frac{d^2}{dt^2} g( \mu_t \circ \Phi_t)= g''(\mu_t \circ \Phi_t)  \(\frac{d}{dt}  \mu_t \circ \Phi_t\)^2 + g'(\mu_t \circ \Phi_t)  \frac{d^2}{dt^2}  \mu_t \circ \Phi_t.
\end{equation*}
Moreover,  by  \eqref{bornesurfp} and the assumption \eqref{assfs} we have $|g^{(k)}(x)|\le C\beta $ for all $k$, for $x $ bounded above and below by positive constants. Noting that  $\mu_t\circ \Phi_t$ remains bounded above and below by positive constants,  
we deduce that for $k=1,2$,
$$\left|\frac{d^k}{dt^k} g(  \mu_t \circ \Phi_t)\right|\le C\beta    |\psi|_{C^{1}}^{ k}.$$
 If $\d=2$, there is no dependence in $\mu_0$ inside $\fd$.
In the same way 
$$\int_{\R^\d} \mu_t\log \mu_t=\int_{\R^\d} \log ( \mu_t\circ\Phi_t) d\mu_0$$
and the derivatives of $\log (\mu_t\circ \Phi_t) $ are bounded by $C \beta  |\psi|_{C^{1}}^{ k}.$
Integrating against  $ d\mu_0$ on the support of $\psi$ we deduce that  for $k \le 2$,
\be \label{borneB}  |\phi^{(k)}(t)| \le   C\beta   \ell^{\d}   |\psi|_{C^{1}}^{ k}. \ee The result follows by Taylor expansion.

\end{proof}

As announced in Section \ref{secoutline}, the proof of the CLT  relies on  improving the error on the expansion of the free energy by comparing two ways of expanding the relative free energy: one by transport and one by application of Proposition \ref{th1b}. The idea is that if one knows a quadratic function  on a whole interval up to a given error,  then one can estimate it near zero with a much better error. 

This is done in the following \cord{crucial lemma which contains the ``H\"older trick" and the exponential moment control of the ``anisotropy".}
\begin{lem}\label{pro72}
Let $\mu_t= (\id+t\psi)\# \mu_0$ for  some $\psi$ supported in a cube $Q_\ell$ where $\mu_0$ is bounded below by a positive constant. 
\cord{Assume that 
\be\label{bornespsifin} |\psi|_{C^k} \le C \ell^{-k-1} \quad \text{for } k =0,1,2,3,\ee 
 and that 
 $t\ell^{-2}<1$ is small enough that the result of Proposition \ref{prop:comparaison2} holds. For any  $\alpha'> 0$, the following holds.   Let
 \be \label{defDn} \begin{cases} D_N(\psi)^{-2}:=  \ell^{- 4}\log (N^{\frac1\d}\ell)
  & \text{if}\  \d=2\\
   D_N(\psi)^{-2}:= \ell^{-4}\( \log (\ell N^{\frac1\d})   (N^{\frac1\d}\ell)^{\alpha'(\d-2)}  + (N^{\frac1\d}\ell)^{1-\alpha'}+ (N^{\frac1\d}\ell)^{1-2\alpha'+ \frac{\d}{\d-2}} \)  &   \text{if} \  \d \ge 3.\end{cases}\ee}
Assume we know that for each $s \le  \cord{D_N(\psi)}$, we have
\be \label{relexp}\log \frac{\K_N(\mu_s)}{\K_N(\mu_0)} =  N\( \mathcal{Z}(\beta,\mu_s)- \mathcal{Z}(\beta,\mu_0)\)+O( \beta \chi(\beta) N\ell^\d \mathcal R_s)\ee
with 
\be \label{maxR}
\max_{s\in [0,\cord{D_N(\psi)}]} \mathcal R_s \le C.\ee
Then for every $t $ satisfying 
\be\label{condtnew}
|t| <t_0:= C^{-1} \(    \max_{s\in [0,\cord{D_N(\psi)}]}\mathcal R_{\cord{s}} \)^{\frac12}  \cord{ D_N(\psi)}\ee for some appropriate $C$ depending only on  the bound in \eqref{maxR}, we have 
\cord{ \begin{equation}\label{pconc2d}
\log \frac{\K_N( \mu_t)}{\K_N(\mu_0)}=t  N \mathcal B_1(\beta, \mu_0,\psi )+ 
O \( t  \beta \chi(\beta) N\ell^{\d}  \Big( \max_{s\in [0,D_N(\psi) ]}\mathcal R_s \Big)^{\hal }  D_N(\psi)^{-1}      \)
+o(1)\indic_{\d\ge3}.
 \end{equation}}
\cord{Moreover, if $\d=2$, we have  for all $t$ satisfying \eqref{condtnew}
\begin{multline}\label{borneani}
\log \Esp_{\PNbeta} \Big[ \exp \(- t \beta \( \Ani_1(X_N, \mu_0, \psi) + N \mathcal B_1(\beta, \mu_0, \psi)\)   \)\Big]\\= O\(  t\beta \chi(\beta) N\ell^\d (\max_{s\in [0, D_N(\psi)]}  \mathcal R_s)^{\hal} D_N(\psi)^{-1}  \).\end{multline} }
\end{lem}

\begin{proof}
\cord{
First we define a good event to be 
$$G= \begin{cases}(\R^\d)^N & \text{if} \ \d=2 \\ \{ \XN, 
\F^{Q_\ell}_N (\XN) \le  M (N\ell^\d) N^{1-\frac2\d} \}& \text{if} \ \d \ge 3\end{cases}$$
where $M$ is some constant. In view of the local laws \eqref{locallawint0} we have that if $M$ is chosen large enough,
\begin{equation}\label{goodevent}\log \PNbeta (G^c) \le - \hal M \beta N\ell^\d\end{equation}
for $N$ large enough.  In view of Proposition \ref{prop:comparaison2}, we have 
 $\F_N^{Q_\ell}(\Phi_t(\XN), \Phi_t\# \mu) \le C \F_N^{Q_\ell}(\XN, \mu) $ and thus by \eqref{pdef} and \eqref{locallawint0} again, we may write 
 \begin{multline}
 \log \frac{\K_N( \mu_t)}{\K_N(\mu_0)}
 = 
 \log \Esp_{\mathsf{Q}_N(\mu_0)}
 \( \exp\(- \beta N^{\frac2\d-1} \( \F_N^{Q_\ell} (\Phi_t(\XN), \Phi_t\# \mu) - \F_N^{Q_\ell}(\XN, \mu) \) \) \) 
 \\
 = \log \Esp_{\mathsf{Q}_N(\mu_0)} \( \indic_{G} \exp \( - \beta N^{\frac2\d-1}\( \F_N^{Q_\ell} (\Phi_t(\XN), \Phi_t\# \mu) \le C \F_N^{Q_\ell}(\XN, \mu) \)\) \)+ o(1).
 \end{multline}  
 Next we wish to  insert the expansion of   Proposition \ref{prop:comparaison2}  into the exponent.  More precisely, we use \eqref{deriv23} and make use of \eqref{bornespsifin}, $t\ell^{-2}<1$  and $N^{\frac1\d}\ell \ge 1$ to absorb all the terms in factor of $t$.
  In dimension $\d=2$, it yields 
\begin{multline}\label{exp2forme2}
\log \frac{\K_N( \mu_t)}{\K_N(\mu_0)}=\log \Esp_{\mathsf{Q}_N(\mu_0)} \Bigg[ \exp\Bigg( - \beta    \, t \Ani_1(\XN, \mu_0, \psi) \\
+t^2 O\Big( D_N(\psi)^{-2}   
   \beta (N\ell^\d+ \F_N^{Q_\ell}(\XN, \mu_0)  )    \Big)\Bigg)\Bigg].  
\end{multline}}
   Equating with 
the expansion  \eqref{tayexp},
 and setting 
\be \label{defgk}
\gamma =  \beta N^{\frac{2}{\d}-1} 
\Ani_1(\XN, \mu_0, \psi) +N\mathcal B_1(\beta, \mu_0, \psi)\ee
we thus find 
\begin{multline*}
\log \Esp_{\mathsf{Q}_N(\mu_0)} \(  \exp\( -  t \gamma + O\(   \beta    t^2 \cord{D_N(\psi)^{-2}}   \( N\ell^\d + \F_N^{Q_\ell}(\XN, \mu_t)\) \) \) \)\\
= O\( t^2  N \ell^{\d}|\psi|_{C^1}^2 \beta  \chi(\beta)  \)+O\(\beta \chi(\beta) N\ell^\d(\mathcal R_t+\mathcal R_0)\).
\end{multline*}
\cord{
Using Cauchy-Schwarz's inequality  and the local laws \eqref{locallawint0} we then deduce that   if $|t| D_N(\psi)^{-1}<C^{-1}$ with $C>2$ (which follows from \eqref{condtnew} and \eqref{maxR}),
 \begin{equation}\label{newlabel}
\log \Esp_{\mathsf{Q}_N(\mu_0)} \( \indic_G \exp\( -  \hal t \gamma \)\) =  O\(   \beta  \chi(\beta)    N\ell^\d\( t^2 \cord{D_N(\psi)^{-2}} + \mathcal R_t+\mathcal R_0\)\).
\end{equation}
}
  
\cord{We next turn to dimension $\d\ge 3$ and  apply \eqref{deriv23}, which yields 
\begin{multline*}
\log \frac{\K_N( \mu_t)}{\K_N(\mu_0)}=\log \Esp_{\mathsf{Q}_N(\mu_0)} \Bigg[\indic_{G}  \exp\Bigg( - \beta N^{\frac2\d-1}    t \Ani_1(\XN, \mu_0, \psi) \\
+  \beta N^{\frac2\d-1} t^2  O\Big( \Big(  
\( \ell^{-4} \log (\ell N^{\frac1\d}) (N^{\frac1\d}\ell)^{\alpha'(\d-2)} + \ell^{-4} (N^{\frac1\d}\ell)^{1-\alpha'}\) \(\F_N^{Q_\ell}+  N^{1-\frac2\d} N \ell^\d\) 
\\+      \ell^{-4} ( N^{\frac{1}{\d }  }\ell)^{1-2\alpha'}  N^{-\frac1\d}\(\F_N^{Q_\ell}+  N^{2-\frac2\d}  \ell^\d)^{  \frac{\d-1}{\d-2}}  \)
 \Big)  \Bigg)\Bigg]+o(1).
\end{multline*}
As above, equating with \eqref{tayexp} and setting \eqref{defgk}, we obtain with the local laws  \eqref{locallawint0} and the fact that we are in the good event $G$,
that \eqref{newlabel} also holds  in this case $\d\ge 3$, up to an added $o(1)$, by choice  of \eqref{defDn}. }

\cord{We now choose $\alpha<D_N(\psi)$ small enough that 
$$ \frac{ \alpha^2} {D_N(\psi) ^2 } \le C(\mathcal R_0+\mathcal R_\alpha).$$}
For that we choose
$$\alpha= \cord{C^{-1}} \( \max_{t\in [0,D_N(\psi)]}\mathcal R_t \)^{\frac12}   \cord{D_N}(\psi)
$$
which is indeed \cord{$<D_N(\psi)$ if  $ \max_{t\in [0,D_N(\psi)]}\mathcal R_t $} is bounded  and $C$ is well-chosen.

With this choice we then have   \be\label{debut}
\log \Esp_{\mathsf{Q}_N(\mu_0)} \( \cord{\indic_G} \exp\( -  \alpha \gamma  \) \)= O\( \beta \chi(\beta) N\ell^\d \max_{t\in [0,\cord{D_N}(\psi)]}\mathcal R_t\) \cord{+o(1)\indic_{\d\ge3}}
\ee
and the same applies as well to $-\alpha$.

Using H\"older's inequality we deduce  that if $t/\alpha$ is small enough, more precisely if \eqref{condtnew} holds,
we have
\be\label{hold}\left| \log \Esp_{\mathsf{Q}_N(\mu_0)} \(\cord{\indic_G} \exp\( \gamma t \) \)\right| \le C \frac{|t|}{\alpha}
\beta \chi(\beta) N\ell^\d \max_{s\in [0,\cord{D_N}(\psi) ]}\mathcal R_s\cord{+o(1)\indic_{\d\ge3}}
 .\ee

Inserting \eqref{defgk} and \eqref{hold} into  \eqref{exp2forme2}, and using  the definition of $\alpha$ and \eqref{locallawint0} again, we obtain
 \begin{multline*}
\log \frac{\K_N( \mu_t)}{\K_N(\mu_0)}=t  N \mathcal B_1(\beta, \mu_0,\psi )+ 
O \( t \beta \chi(\beta) N\ell^{\d} \cord{D_N}(\psi)^{-1} \Big( \max_{t\in [0,\cord{D_N}(\psi) ]}\mathcal R_t \Big)^{\frac12}      \)\\+
O\(    \beta \chi(\beta) N\ell^\d  t^2 \cord{D_N(\psi)^{-2}} \) \cord{+o(1)\indic_{\d\ge3}.}
 \end{multline*}
 \cord{Since the second error can be absorbed into the first in view of \eqref{condtnew}, this gives}
 the result.
\end{proof}
We now specialize to $\mut$ with the notation of Section \ref{sec5}.
\begin{coro}\label{pro58} \cord{Under the same assumptions, if  $t$ satisfies  \eqref{condtnew}, then}
if $\d=2$, we have  
\begin{multline}\label{537}
\log \frac{\K_N(\tilde \mutt)}{\K_N(\mut)}=    tN  \frac{\beta}{4}
  \int_{\R^\d}\div(\psi \mut) \log \mut\\+ \cord{O \( t  \beta \chi(\beta) N\ell^{\d}  \Big( \max_{s\in [0,D_N(\psi) ]}\mathcal R_s \Big)^{\hal }  D_N(\psi)^{-1}      \)}   ,\end{multline}
 and if $\d\ge 3$ 
\begin{multline}\label{538}
\log \frac{\K_N(\tilde \mutt)}{\K_N(\mut)}= t N 
 \( 1-\frac2\d\)  \int_{\R^\d} \div(\psi \mut) \( \fd(\beta \mut^{1-\frac2\d})+\beta \mut^{1-\frac2\d}\fd'(\beta \mut^{1-\frac2\d})\)\\+ \cord{O \( t  \beta  N\ell^{\d}  \Big( \max_{s\in [0,D_N(\psi) ]}\mathcal R_s \Big)^{\hal }  D_N(\psi)^{-1}      \) +o(1).}   \end{multline} 

\end{coro}
\begin{proof}
This is just a specialization of Lemma \ref{pro72} to $\mu_0=\mut$,  $\psi$ of \eqref{choicepsi} and 
$\mu_t=\tilde \mutt$.  In dimension $\d=2$ we  compute directly that 
$$\mathcal B_1(\beta, \mut, \psi)= \frac{\beta}{4} \int_{\R^\d} \div(\psi \mut) \log \mut ,$$
In dimension $\d\ge 3$, we  evaluate that $$\mathcal B_1(\beta, \mut, \psi) = \( 1-\frac2\d\) \beta \int_{\R^\d} \div(\psi \mut) \( \fd(\beta \mut^{1-\frac2\d})+\beta \mut^{1-\frac2\d}\fd'(\beta \mut^{1-\frac2\d})\).$$
\end{proof}

\subsection{Conclusion} 
To prove the CLT, 
the correct choice of $t$ is 
\be\label{choicet}  
t= \tau \ell^2 \beta^{-\hal } (N^{\frac1\d}\ell)^{-1-\frac{\d}{2}}.
\ee
and the choice of  $\psi$ is \eqref{choicepsi}.
We note that by definition of $\rb$ in \eqref{rhobeta} and the assumption \eqref{ass1}, $\tau$ being fixed,  we always
have
\be \label{controlt}
|t| \le C \ell^2  \(\frac{N^{\frac1\d}\ell}{\rb}\)^{-1-\frac\d2} \ll \ell^2.\ee

We now wish to evaluate \eqref{laplace1}.
We wish to replace  $\nut$ by $\tilde\mutt$  in that formula, for that  we use \eqref{estmu} and \eqref{estpsi} and inserting it into \eqref{compdesk} we obtain that if $V\in C^{5+2q}\cap C^7$
\begin{multline}
\label{nutmutt2}
|\log \K_N(\nut) -\log \K_N(\tilde \mutt)| 
\\
\le C  \beta \chi(\beta) N\ell^\d  t^2 \Big(    \(\sum_{k=0}^{q +1} \frac{|\xi|_{C^{ 2k+1} }}{\theta^k}\)^2  + \( \sum_{k=0}^q \frac{|\xi|_{C^{ 2k+1} }}{\theta^k} \) \( \sum_{k=0}^q \frac{|\xi|_{C^{ 2k+3} }}{\theta^k} \)\\
+ \ell 
 \(\sum_{k=0}^q \frac{|\xi|_{C^{2k+2} }}{\theta^k}\) \(\sum_{k=0}^q \frac{|\xi|_{C^{ 2k+1} }}{\theta^k}\)
+ 
\ell \(\sum_{k=0}^q \frac{|\xi|_{C^{ 2k+2} }}{\theta^k}\)\(\sum_{k=0}^q \frac{|\xi|_{C^{ 2k+3} }}{\theta^k}\) 
\\+\ell
\(\sum_{k=0}^q \frac{|\xi|_{C^{ 2k+4} }}{\theta^k}\)\(\sum_{k=0}^q \frac{|\xi|_{C^{ 2k+1} }}{\theta^k}\)
\Big) ,\end{multline} where $C$ depends on the norms of $\mut$ in $\supp\, \xi$ up to $C^3$, which are uniformly bounded in terms on $V$ in view of \eqref{bornemutcn}.
We may now  evaluate all the terms in \eqref{laplace1} by combining 
 the results   \eqref{nutmutt2}, \eqref{laplace10},  \eqref{513} and  \eqref{537}--\eqref{538} all applied with the choice  \eqref{choicet} and inserting \eqref{choicepsi}.  Each of these results generates an error.
 
We let $\Error_1$ denote the error in the right-hand side of \eqref{513}, $\Error_2$ denote the error in the right-hand side of \eqref{laplace10},
$\Error_3$  the  error term in  \eqref{537} or \eqref{538} and \cord{$\Error_4$} the error in \eqref{nutmutt2}.
 With this notation, obtain
\be \label{7235}\left|
\log \Esp_{\PNbeta} \(\exp\(-\tau \beta^\hal  ( N^{\frac{1}{\d}} \ell)^{1-\frac{\d}{2}}  \Fluct(\xi) \)\)+
\tau m(\xi)   - \tau^2  \ell^{2-\d}  v(\xi) \right|\le \sum_{i=1}^{\cord{4}} |\Error_i| \ee
with $v$ as in \eqref{defvari}, that is
\be
v(\xi)=-
\frac{ 1}{2\cd}  \int_{\R^\d} \left|   \sum_{k=0}^q  \frac{1}{\theta^k} \nab L^k (\xi)\right|^2
+ \frac{1}{\cd}\int_{\R^\d}  \sum_{k=0}^q \nab \xi \cdot \frac{ \nab L^{k}( \xi) }{\theta^k} -
\frac{1}{2 \theta}\int_{\R^\d} \mut \left|\sum_{k=0}^q\frac{ L^{k+1}(\xi)}{\theta^k} \right|^2
\ee and with
$$m(\xi)= \begin{cases}  - \frac{1}{4}\beta^\hal \displaystyle \int_{\R^\d} 
 \( \sum_{k=0}^q\frac{ \Delta L^k (\xi)}{
\cd\theta^k}\)
 \log \mut & \text{if} \ \d=2\\ -
  N \ell^2 \beta^{\hal} (N^{\frac1\d}\ell)^{-1-\frac\d2}  \( 1-\frac2\d\)  \displaystyle \int_{\R^\d} \(  \sum_{k=0}^q\frac{ \Delta L^k (\xi)}{
\cd\theta^k}\)\( \fd(\beta \mut^{1-\frac2\d})+\beta \mut^{1-\frac2\d}\fd'(\beta \mut^{1-\frac2\d})\) & \text{if} \ \d\ge 3.\end{cases}$$
As soon as we can show that  $\sum_{i=1}^{\cord{4}} \Error_i=o(1)$, we obtain that the Laplace transform of a suitable scaling of $\Fluct (\xi) $ converges to that of a Gaussian, proving  the Central Limit Theorem. 
We will now show this  when specializing to the setting where $|\xi|_{C^k}\le C \ell^{-k}$. The interested reader could estimate the error for more general choices of $\xi$. 
The more regular $\xi$ and $V$ are, the larger $q$ can be taken, and the better the errors in \eqref{7235}, in particular in terms of their dependence in $\theta \gg 1$.
 Also the variance and the mean contain more correction terms. 

In dimension $\d=2$  it suffices to take $q=0$, hence $\xi \in C^4$ suffices, but better estimates of the variance and mean can be obtained if $\xi$ is more regular. If $\d\ge3$ we will need to take $q$ larger as $\beta$ gets small.

 \subsection{Estimating the errors}
 From now on, we assume 
 \be\label{hypxi}
 |\xi|_{C^k}\le C \ell^{-k}\quad \text{for all } \ k\le 2q+4.\ee This way, from \eqref{choicepsi} and $\theta\ell^2 \ge 1$ (see \eqref{ass1coro}), we have \cord{that 
 \eqref{bornespsifin} holds}
 where $C$ depends on the norms of $\mut$ (bounded by \eqref{bornemutcn}).
This implies, using also \eqref{ass1}, that   $\cord{D_N}(\psi)$ defined \cord{in \eqref{defDn} }satisfies \cord{
 \be \label{Dpsi} 
 C^{-1} \ell^2  \(\log (\ell N^{\frac1\d})\)^{-\hal} \le   D_N(\psi)\le C \ell^2 \ee if  $ \d=2$.}
 \cord{In dimension $\d \ge 3$,
 we have instead 
 \be \label{Dpsi3}
 C^{-1}\ell^2  \( \log (\ell N^{\frac1\d})   (N^{\frac1\d}\ell)^{\alpha'(\d-2)}  + (N^{\frac1\d}\ell)^{1-\alpha'}+ (N^{\frac1\d}\ell)^{1-2\alpha'+ \frac{\d}{\d-2}} \)^{-\hal}
    \le D_N(\psi)\le C \ell^2.\ee
  We now make a specific choice of $\alpha'$, so as to minimize the sum appearing above and let  
 \begin{equation}\label{choicealpha}\alpha'= \frac{2\d-2}{\d(\d-2)} .\end{equation}
 We find that 
\be \label{Dpsi33}
 D_N(\psi)^{ -1} \le  C \ell^{-2}    ( N^{\frac1\d}\ell)^{1-\frac1{\d}  } .    \ee
  }
  
 In view of \eqref{controlt} we also deduce that $t|\psi|_{C^1}$ and $t|\psi|_{C^2}N^{-\frac1\d}\log (\ell N^{\frac1\d})$ are small, as needed for Proposition  \ref{prop:comparaison2}. 
 
 \subsubsection{The first error term}
 By definition it is
  \be \label{deferr2}
\Error_1:= C\beta N^{1+\frac2\d} \( \frac{t^3}{\theta} \int_{\R^\d}\mut \left|\sum_{k=0}^q\frac{ L^{k+1}(\xi)}{\theta^k} \right|^3\)\ee
and we have with \eqref{Lin},  \eqref{choicet} and \eqref{hypxi}
\be \label{error2}
|\Error_1|\le C   \tau^3  \beta^{-\frac32 } (N^{\frac1\d}\ell)^{-3-\frac\d2} \le C \(\frac{N^{\frac1\d}\ell}{\rb}\)^{-3-\frac\d2}  ,\ee
where we used that $\beta^{-\hal} \le \rb$ and $\rb \ge 1$ by \eqref{rhobeta}.
This term tends to $0$  with an algebraic rate  in $N^{1/\d} \ell/\rb$ in all dimensions.

\subsubsection{The second error}
The next error is $\Error_2$ equal to the right-hand side of \eqref{laplace10}  and already estimated in \eqref{bornelap2}, hence with $t$ as in \eqref{choicet}, it becomes
\begin{multline}\label{deferror1}
|\Error_2| \le C \sqrt{\chi(\beta)} \beta N^{1+\frac1\d} \ell^\d\(  C \tau^2 \beta^{-1}\ell^4 (N^{\frac{1}{\d}}\ell)^{-2-\d}
   \sum_{m=0}^{2q}\frac{1}{\theta^{m+1}} \sum_{p+k=m} |\xi|_{C^{2k+2}}|\xi|_{C^{2p+3}}\)\\
        +  C \sqrt{\chi(\beta)} \beta N^{1+\frac1\d} \ell^\d\frac{  |\tau| \beta^{-\hal } \ell^2 (N^{\frac{1}{\d}}\ell)^{-1-\frac{\d}{2}} }{\theta^{q+1}}  \sum_{k=2}^{2q+3}   |\xi|_{C^{k} } \\
        + C \theta N\ell^\d\(C \tau^2 \beta^{-1} \ell^4 (N^{\frac{1}{\d}}\ell)^{-2-\d}
  \sum_{m=0}^{2q}\frac{1}{\theta^{m+1}} \sum_{p+k=m} |\xi|_{C^{2k+2}}|\xi|_{C^{2p+3}}     + C\frac{\beta^{-\hal } \ell^2 (N^{\frac{1}{\d}}\ell)^{-1-\frac{\d}{2}}  }{\theta^{q+1}}     \sum_{k=2}^{2q+3}   |\xi|_{C^{k} }\)^2.
\end{multline}
When  \eqref{hypxi} holds,    we find after inserting the definition of $\theta$ and simplifying terms 
\begin{multline}\label{error1}
|\Error_2 |\le  C \sqrt{\chi(\beta)} \beta^{-1} \tau^2  (N^{\frac{1}{\d}}\ell)^{-3}  
  + C| \tau| \sqrt{\chi(\beta)} \beta^{-\hal -q}      (N^{\frac{1}{\d}}\ell)^{\frac{\d}{2} -2-2q} 
  \\ +  C   \tau^4 \beta^{-3}  (N^{\frac1\d}\ell)^{-6-\d} + C \tau^2 \beta^{-2q-2}   
  (N^{\frac{1}{\d}}\ell)^{-4-4q}. \end{multline}
Next we note that by \eqref{rhobeta} we have $\chi(\beta) \beta^{-1} \le  C\rb^2 $ so using also that $\rb\ge 1$ we find 
\begin{multline}\label{error12}
|\Error_2 |\le  C \tau^2  \rb^2  (N^{\frac{1}{\d}}\ell)^{-3}  
  + C| \tau | \rb^{1+2q}      (N^{\frac{1}{\d}}\ell)^{\frac{\d}{2} -2-2q} 
  \\ +  C   \tau^4 \rb^{\frac23} (N^{\frac1\d}\ell)^{-6-\d} + C \tau^2 \rb^{4q+4}
  (N^{\frac{1}{\d}}\ell)^{-4-4q} \\
  \le C \tau^2  \( \frac{N^{\frac1\d}\ell}{\rb}\)^{-3}+ C\tau^4 \(\frac{N^{\frac1\d}\ell}{\rb} \)^{-6-\d} + C| \tau| \rb^{1+2q}      (N^{\frac{1}{\d}}\ell)^{\frac{\d}{2} -2-2q} .
   \end{multline}
   The first two terms always tend to $0$ by the assumption $N^{\frac1\d}\ell \gg \rb$.
   If $\d=2$  we find that the third term is $O\(\frac{N^{\frac1\d}\ell}{\rb}\)^{-1-2q}$ which tends to $0$ for any $q\ge 0$.
   We then have 
   \be \label{error1fin}
   |\Error_2|\le  C (\tau^2 +\tau^4) \( \frac{N^{\frac1\d}\ell}{\rb}\)^{-3} + C| \tau| \(\frac{N^{\frac1\d}\ell}{\rb}\)^{-1-2q}\quad \text{if} \ \d=2.\ee
For $\d\ge 3$, if $\beta$ does not tend to $0$, then $\rb$ is bounded and the third term tends to $0$ since we assumed $q> \frac\d4-1$.
If $\beta \to 0 $ then we use the extra assumption
 \eqref{condsup3d}, from which we may then take $q$ large enough depending on $\ep$   (so $\xi$ needs to be regular enough) so that 
    the last term tends to $0$.   This is the reason for the assumption that $q$ is larger than some constant depending on $\ep$ made in Theorem \ref{th5}.
    
This concludes the analysis of $\Error_2$, with again an algebraic convergence to $0$ as $N^{1/\d} \ell/\rb\to \infty$.

\subsubsection{The third error}

By definition it is the rate error
\be\label{deferr5} \Error_3:= \cord{t  \beta \chi(\beta) N\ell^{\d}  
D_N(\psi)^{-1}
 \Big( \max_{s\in [0,C\ell^2]}\mathcal R_s \Big)^{\hal}   }  .\ee
Inserting  \eqref{choicet} \cord{and \eqref{Dpsi} or \eqref{Dpsi33}} this is 
\be \label{error5} |\Error_3| \le C |\tau|  \cord{\beta^{\hal} \chi(\beta)}     (N^{\frac1\d}\ell)^{\frac\d2-1} \Big( \max_{s\in [0,C \ell^2]}\mathcal R_s \Big)^{\hal}\cord{\times 
\begin{cases}
(  \log (\ell N^{\frac1\d}))^{\hal}& \text{if} \ \d=2\\
   (N^{\frac1\d}\ell)^{1-\frac1{\d}} &\text{if} \ \d\ge 3.\end{cases}}  \ee
For $\d \ge 3$, the convergence of $\Error_3$ is ensured by the assumption \eqref{assR}. 

We now check that this error term can be made small if $\d=2$.

To evaluate $\mathcal R_s$ we need to compare \eqref{relexp} and Proposition  \ref{th1b}. First we note that 
\eqref{condsurt0} and \eqref{condsurt}, \eqref{condsurt2}  are verified by \eqref{controlt}.
Then in view of \eqref{estnabmut} 
for all $\tilde \mutt$ with $t<t_0$ we have $|\tilde \mutt|_{C^1} \le  C+ |t|  \ell^{-3}  \le \ell^{-1} $ by \eqref{controlt}, which we input into the definition of $\mathcal R$.  
In view of \eqref{relexp} and Proposition \ref{th1b},  we may thus bound 
$$\max_{s\in [0,C \ell^2]}\mathcal R_s=   C\(\frac{N^{\frac1\d}\ell}{\rb}\)^{-\frac12}   \(\log    \frac{N^{\frac1\d}\ell}{\rb}\)^{\frac1\d}.
$$
To apply \eqref{537}--\eqref{538} we needed \eqref{condtnew} to be satisfied, that is
$$|\tau| \ell^2 \beta^{-\hal} (N^{\frac1\d}\ell)^{-1-\frac\d2} <C \ell^2  \cord{\( \log (\ell N^{\frac1\d})\)^{-\hal}} \( \max_{[0, C \ell^2]} \mathcal R_t\)^{\frac12},$$
for this it suffices that 
$$ \cord{\( \log (\frac{ N^{\frac1\d}\ell}{\rb})\)^{\hal}} \(\frac{N^{\frac1\d}\ell}{\rb}\)^{-1-\frac\d2} \cord{\ll} C \(\frac{N^{\frac1\d}\ell}{\rb}\)^{-\frac1{4}}   \(\log    \frac{N^{\frac1\d}\ell}{\rb}\)^{\frac1{2\d}},$$
which is clearly  satisfied as soon as $N $ is large enough, since we assume $N^{1/\d}\ell\rb^{-1} \gg 1$.
We  may now  write
$$|\Error_3|\le C \beta^{\hal} \chi(\beta)  \(\frac{N^{\frac1\d}\ell}{\rb}\)^{-\frac14}   \(\log \frac{N^{\frac1\d}\ell}{\rb}\)^{\frac1{2\d}}  \cord{\( \log ( N^{\frac1\d}\ell)\)^{\hal}}  .$$
Thus if $\d=2$, $\Error_3 \to 0$ algebraically as soon as $\beta\le 1 $, and if $\beta \ge 1$ (then $\rb=1$) we use \eqref{condsup3d2}.


\subsubsection{The \cord{fourth} error term}
It is by definition the term in \eqref{nutmutt2}, and inserting \eqref{choicet} and \eqref{hypxi},  we find
$$|\Error_{\cord{4}}|\le C\tau^2 
 \chi(\beta) ( N^{\frac1\d} \ell)^{-2} \le C \( \frac{N^{\frac1\d}\ell}{\rb}\)^{-2}.$$
This term always tends to $0$, algebraically.

\subsubsection{Conclusion}
 We may now conclude that all terms are $o(1)$  under our assumptions if $\d=2$, and that they are $o(1)$ in dimension $\d \ge 3$ provided \eqref{assfs} and \eqref{assR} hold. 

 Moreover, a rate of convergence as a negative power of $N^{\frac1\d}\ell \rb^{-1}$ is provided for most of the error terms.

\subsection{The case of small temperature - proof of Theorems \ref{thlowt2} and \ref{thlowt3} }
Here we may assume $\beta \ge 1$, so $\rb =1$ and the convergence rate will be in terms of $N^{\frac1\d}\ell$.
In that case, we choose instead $ t= s \ell^2 \beta^{-1} (N^{\frac1\d}\ell)^{-1-\frac\d2} $
which is equivalent to taking
$\tau=s\beta^{-\hal} $, and we retrace the same steps to find instead of \eqref{723}  
\be
\left|\log \Esp_{\PNbeta} \(\exp\(-s ( N^{\frac{1}{\d}} \ell)^{1-\frac{\d}{2}}  \Fluct(\xi) \)\)+ s \beta^{-\hal} m(\xi) - s^2 \beta^{-1} v(\xi) \right|\le \sum_{i=1}^{\cord{4}} |\Error_i|.\ee
The errors appearing here are each smaller than the respective errors produced in the previous proofs because the extra factors in powers of $\beta^{-\hal}$ that appear are all $\le 1$.
Hence we only need to check that $\Error_3 $ tends to $0$, with 
$$|\Error_3|\le C (N^{\frac1\d}\ell)^{\frac\d2-1}\(\max_{s\in [0, C \ell^2]}  \mathcal R_s\)^\hal  \cord{\times 
\begin{cases}
(  \log (\ell N^{\frac1\d}))^{\hal}& \text{if} \ \d=2\\  (N^{\frac1\d}\ell)^{1-\frac1{\d}}
    &\text{if} \   \d \ge 3. \end{cases}} $$
For \cord{$\d \ge 3$} this is ensured by \eqref{assR2}. 
For $\d=2$ we have
$$|\Error_3|\le C \( \frac{N^{\frac1\d}\ell}{\rb}\)^{-\frac14} \log^{\cord{\frac34}}  \frac{N^{\frac1\d}\ell}{\rb} $$ which tends to $0$.
We may also choose $q=0$, although the result would be as true with larger $q$ and this   concludes the proof.

\appendix
\section{Proof of Proposition \ref{prop:comparaison2}} 
\label{appa}
 
\subsection{A preliminary bound on the potential near the charges} Let $\mu$ be a bounded and $C^2$ probability density on $\R^\d$, and let  $\XN$ be in $(\R^\d)^N$.  We let $h$ be as in \eqref{def:hnmu} and  sometimes write $h^\mu[\XN]$ to emphasize the $\XN$ and $\mu$ dependence.
 For any $i = 1, \dots, N$ we let
\begin{equation} \label{def:tH}
\tilde h_i(x)  :=  h(x) -\g (x-x_i).
\end{equation}
 We will use in particular the notation of \eqref{formu2}.
 We start by  adapting  to arbitrary dimensions some results of \cite{ls2}.
\begin{lem}  Let $\vec{\eta}$ be such that $\eta_i \le \rr_i$ for each $i$.
We have for $i = 1 , \dots, N$
\begin{equation} \label{HNeta}
h_{ \vec{\eta}} =\begin{cases}
h & \text{ outside } \D(x_i, \eta_i) \\
\tilde h_i \text{ (up to a constant)} & \text{ in each } \D(x_i, \eta_i ).
\end{cases}
\end{equation}

In particular, it holds  that
\begin{equation}\label{idenhh}
\int_{\R^{\d} }|\nab h_{ \vec{\eta} }|^2=
\int_{\R^{\d} \backslash \cup_{i=1}^N \D(x_i, \rr_i )}  |\nab h |^2 + \sum_{i=1}^N \int_{\D(x_i, \rr_i )} |\nab \tilde h_i|^2.
\end{equation}
\end{lem}
\begin{proof}
The first point follows from \eqref{formu2}  with \eqref{fconv} and the fact that  the balls $\D(x_i,\rr_i)$ are disjoint by definition hence the $B(x_i, \eta_i)$'s as well. The second point is a straightforward consequence of the first one.
\end{proof}

We let for $i=1, \dots, N$
\begin{equation} \label{def:lambdai}
\lambda_i(\XN, \mu):= \int_{ \D(x_i, \rr_i )} |\nab \tilde h_{i}|^2.
\end{equation}
We will later often denote it simply by $\lambda_i$. 
\begin{lem} Assume $\mu \in C^{\cord{m-2+ \sigma}}$ for some $\sigma>0$ \cord{and integer $m\ge 1$}. For each $i = 1, \dots, N$,  we have 
\begin{align}\label{contrhi}
& \|\nab \tilde h_i \|_{L^\infty(B(x_i, \hal\rr_i)) } \le   C\cord{\(\rr_i^{-\frac{\d}{2} } \lambda_i(\XN,\mu)^\hal +  N\rr_i \|\mu\|_{L^\infty}  \)}\\
\label{contrhi2} &\|\nab^{\cord{m}} \tilde h_i\|_{L^\infty(B(x_i, \hal\rr_i)) } \le   C\( \rr_i^{\cord{1-m}-\frac{\d}{2} } \lambda_i(\XN,\mu)^\hal +  N\cord{\rr_i^{2-m}}  \|\mu\|_{L^\infty}  +N \cord{\rr_i^{\sigma} |\mu|_{C^{m-2+\sigma} (B(x_i, \rr_i))}} \)  
\end{align}
for some constant $C$ depending only on $\d$.
\end{lem}
\begin{proof}
We exploit the fact that $\tilde h_{i}$ is regular in each $\D(x_i, \rr_i)$.  Recall that $h= \g* \( \sum_{i=1}^N \delta_{x_i}-N\mu\)$ so that 
\be\label{thiii}
\tilde h_{i}= \g* \(\sum_{j\neq i} \delta_{x_j} -N \mu\).\ee
 We may thus write $\tilde h_{i}$ as $\tilde h_{i}=u+v$ where
\begin{equation}\label{equ}
u=\g* \(- N \mu\chi_i \)
\end{equation} with $\chi_i$ a smooth nonnegative function such that $ \chi_i=1$ in $B(x_i, \rr_i)$ and $\chi_i=0$ outside of $B(x_i, 2\rr_i)$; and $v$ solves
\begin{equation}
\label{eqv}
-\Delta v=0 \  \text{in} \ \D(x_i, \rr_i).
\end{equation}

Letting $f(x)= v(x_i+\rr_i x)$, $f$ solves the relation $\Delta f=0$ in $B(0,1)$. 
Elliptic regularity estimates for this equation yield  for any integer $m\ge 1$,
\begin{equation}\label{banm}
\|\nab^m f\|_{L^\infty(B(0, \hal ) ) } \le C \( \int_{B(0,1)}  |\nab f|^2 \)^\hal,\end{equation} for some $C$ depending on $m$.
Rescaling this relation, and using \eqref{def:lambdai} we conclude that 
\begin{multline}\label{reguv}
\|\nab v\|_{L^\infty(B(x_i, \hal\rr_i))} \le \frac{C}{\rr_i} \(\frac{1}{\rr_i^{\d-2}}\int_{B(x_i, \rr_i)}  |\nab v|^2 \)^{\hal}
\le  \frac{C}{\rr_i^{\d/2}  }\( \lambda_i(\XN,\mu)+\int_{B(x_i, \rr_i)}  |\nab u|^2\)^\hal,\end{multline} and \cord{similarly}
\begin{equation}\label{reguv2}
\|\nab^{\cord{m}} v\|_{L^\infty(B(x_i, \hal \rr_i))}  \le   \frac{C}{\rr_i^{\cord{m-1}+\frac{\d}{2} } }\( \lambda_i(\XN,\mu)+\int_{B(x_i, \rr_i)}  |\nab u|^2\)^\hal.\end{equation}
The \cord{following assertions for $u$ are obtained similarly by elliptic regularity and scaling:
\begin{equation} \label{nabu}
\|\nab^m u\|_{L^\infty(B(x_i, \hal \rr_i) )} \le \begin{cases}  C  N \rr_i^{2-m}\|\mu\|_{L^\infty} & \text{for } \ \m \le 1\\
 C  N ( \rr_i^{2-m}\|\mu\|_{L^\infty}+ \rr_i^{\sigma} |\mu|_{C^{m-2+\sigma}(B(x_i, \rr_i))} ) & \text{for } \ \m \ge 2.\end{cases}
\end{equation}}
Inserting into \eqref{reguv}--\eqref{reguv2} and computing explicitly, we deduce that
\begin{equation}
\|\nab v\|_{L^\infty(B(x_i, \hal\rr_i))} 
\le   \frac{C}{\rr_i^{\d/2 } }\( \lambda_i(\XN,\mu))^\hal +  N \|\mu\|_{L^\infty}       \rr_i^{1+\frac{\d}{2}}  \),\end{equation}
and 
\begin{equation}\|\nab^{\cord{m} } v\|_{L^\infty(B(x_i, \hal\rr_i)  )} 
\le  \cord{ \frac{C}{\rr_i^{\cord{m-1}+\frac{\d}{2} } }\( \lambda_i(\XN,\mu))^\hal +  N  \rr_i^{1+\frac{\d}{2}} \|\mu\|_{L^\infty}  \) +  C  N \rr_i^{\sigma} |\mu|_{C^{m-2+\sigma}}}  .
\end{equation}
\cord{Combining with \eqref{nabu}, this concludes the proof.}

\end{proof}

\subsection{Transporting electric fields}
\label{sec:transportedfields}
\begin{lem}\label{lem:divE}
Let $X$ be a vector field on $\R^\d$ and $\Phi$ a diffeomorphism and define 
\be \label{defiPE}
\Phi\#X :=  (D\Phi\circ\Phi^{-1})^T X\circ\Phi^{-1}|\det D\Phi^{-1}|.
\ee
Then
\begin{equation*}
\div  (\Phi\#X) = \Phi\# (\div X)
\end{equation*}
in the sense of distributions, and of push-forward of measures for the right-hand side.
\end{lem}
\begin{proof}
Let $\phi$ be a smooth compactly supported test-function, and let $ f=\div X$ (in the distributional sense).  We have $-\int X \cdot \nab \phi=\int f \phi$, hence 
changing variables, we find
\begin{equation*}
- \int_{\R^\d } X \circ \Phi^{-1} \cdot \nab \phi \circ \Phi^{-1} |\det D\Phi^{-1}| = \int_{\R^{\d}} (\phi \circ \Phi^{-1}  ) (f \circ \Phi^{-1}) |\det D\Phi^{-1}|,
\end{equation*}
and \cord{writing \be\label{nabcirc}\nab \phi \circ \Phi^{-1} = (D\Phi\circ \Phi^{-1} )^T \nab( \phi \circ \Phi^{-1})\ee} we get
\begin{equation*}
-\int_{\R^\d} X \circ \Phi^{-1}   \cdot ( D\Phi\circ \Phi^{-1})^T  \nab( \phi \circ \Phi^{-1})  |\det D\Phi^{-1}|= \int_{\R^\d} \phi \circ \Phi^{-1}  f\circ \Phi^{-1} |\det D\Phi^{-1}|.
\end{equation*}
Since this is true for any $\phi \circ \Phi^{-1}$ with $\phi $ smooth enough, we deduce that in the sense of distributions, we have
\begin{equation*}
\div \left( (D\Phi\circ \Phi^{-1})^T X \circ \Phi^{-1} |\det D \Phi^{-1}| \right)= f \circ \Phi^{-1} |\det D \Phi^{-1}|.
\end{equation*}
which is the desired result. \end{proof}

We now turn to the main proof. 

\subsection{Estimating the first derivative} 
We will denote $\nu= \Phi_t\#\mu$, and for any $\XN \in (\R^\d)^N$ we let $\YN := (\Phi_t(x_1), \dots, \Phi_t(x_N))= (y_1, \dots, y_N)$, hence we let the $t$-dependence be implicit. 
We  use superscripts $\mu$ and $\nu$ to denote the background measure with respect to which $h$ is computed and sometimes use $[\XN]$ or $[Y_N]$ to emphasize the configuration for which $h$ is computed.
    Let $\vec{\eta}$ be such that $\eta_i\le \rr_i$.   
    We wish to compute the energy of the transported configuration by using the transported electric field $\Phi_t (\#\nab h^\mu[\XN])$. The problem is that the transport distorts the truncated measures $\delta_{x_i}^{(\eta_i)}$ and makes them supported in ellipse-like sets instead of spheres. A large part of our work will consist in estimating the error thus made. For this we take a slightly different route than \cite{ls2} which contained an incorrect passage.
    
      Let  us define
\be \label{defiE}
E_{\veta}:= \Phi_t\#(\nab   h_{\veta}^\mu[\XN]),\ee
\be\label{defdhat}\hat \delta_{y_i}:= \Phi_t\# \delta_{x_i}^{(\eta_i)} ,\ee 
and
\be \label{defhh}
\hat h := \g* \( \sum_{i=1}^N \hat \delta_{y_i}- N \nu\).\ee
Note that $\hat h$ implicitly depends on $\veta$.
By Lemma \ref{lem:divE}, we have 
\be 
- \div E_{\veta}= - \Phi_t\# (\Delta h_{\veta}^\mu[\XN]) =  \cd\(\sum_{i=1}^N \hat  \delta_{y_i}-N\nu\) \ee
thus 
\be \label{divfree} \div (E_{\veta}- \nab \hat h)=0.\ee
We next use the fact that  $\nab \hat h$ is the $L^2$ projection of $E_{\veta}$ onto gradients to deduce it has a smaller $L^2$ norm. More precisely, we may write 
$$\int_{\R^\d} |E_{\veta}|^2= \int_{\R^\d} |E_\eta- \nab \hat h|^2 + |\nab \hat h|^2 + 2\int_{\R^\d} (E_\eta -\nab \hat h) \cdot \nab \hat h$$
and use  Green's formula and
 \eqref{divfree} to deduce that the last integral vanishes, hence
\be \label{nrjlocale} 
\int_{\R^\d} |E_{\veta}|^2 = \int_{\R^\d} |\nab \hat h|^2+| E_{\veta}- \nab \hat h|^2.\ee
We also note that  $E_{\veta} = \nab h^\mu[\XN]$ in the interior of the set $\{ \Phi_t \equiv \id \}$.

Without loss of generality, we may assume that $ t|\psi|_{C^1}<\hal$.
We now wish to estimate $\Xi(t)-\Xi(0)= \F_N(\Phi_t(X_N) ,\mu_t)- \F_N(X_N, \mu) = \F_N(Y_N, \nu)- \F_N(X_N, \mu).$

\setcounter{step}{0}
\begin{step}[Splitting the comparison]
Applying Lemma \ref{lem:monoto}   yields
\be \F_N(\XN,\mu)=\frac{1}{2\cd } \int_{\R^{\d}} |\nab h_{\veta}^{\mu}[X_N]|^2  -\hal\sum_{i=1}^N \g(\eta_i) - N  \sum_{i=1}^N \int_{\R^\d} \f_{\eta_i} (x-x_i) d\mu(x)\ee
and 
\begin{equation*}\F_N(\YN,\nu)=\frac{1}{2\cd} \int_{\R^{\d}} |\nab h^\nu_{\vec{\eta}}[Y_N] |^2
- \hal\sum_{i=1}^N \g( \eta_i)-  N  \sum_{i=1}^N \int_{\R^\d} \f_{\eta_i} (x-y_i) d\nu(x).
\end{equation*}
Subtracting these relations and using \eqref{nrjlocale} we find 
\be\label{55}
\F_N(\YN,\nu)- \F_N(\XN,\mu)=\Main +\Rem + \mathrm{Err}-\frac{1}{2\cd} \int_{\R^\d} |E_{\veta}- \nab \hat h|^2\ee
where
\be \label{defmain}
\Main:=  \frac{1}{2\cd}\int_{\R^\d} |E_{\veta}|^2- \frac{1}{2\cd } \int_{\R^{\d}} |\nab h_{\veta}^{\mu}[\XN]|^2 ,
\ee
\be \label{defrem}
\Rem:= 
 \frac{1}{2\cd} \int_{\R^\d} |\nab h_{\veta}^\nu[Y_N]|^2 - \frac{1}{2\cd}\int_{\R^\d} |\nab \hat h|^2
\ee
and 
\be \label{B} \mathrm{Err}:= - N \sum_{i=1}^N \int_{\R^\d} \f_{\eta_i} (x-y_i) d\nu(x)+ N\sum_{i=1}^N \int_{\R^\d}\f_{\eta_i} (x-x_i) d\mu(x).\ee
\end{step}

\cord{
\begin{step}[The last term]
To control the last term, we first show that $E_{\veta}$ is close to $\nab \(h_{\vec{\eta}}^\mu[X_N] \circ\Phi_t^{-1}\)$. 
Indeed, using \eqref{defiPE} and \eqref{nabcirc}, we find that 
$$E_{\veta}= \((D\Phi_t \circ \Phi_t^{-1})^T\)^2 \nab\( h_{\vec{\eta}}^\mu[X_N]\circ \Phi_t^{-1}\)|\det D\Phi_t^{-1}|= \nab \( h_{\vec{\eta}}^\mu[X_N]\circ \Phi_t^{-1}\)\( I+ O(t|\psi|_{C^1(U_\ell) }) \)$$
after a Taylor expansion in $t$. It follows that 
$$\int_{\R^\d} \left|E_{\veta}- \nab \( h_{\vec{\eta}}^\mu[X_N]\circ \Phi_t^{-1}\)\right|^2 \le C t^2 |\psi|_{C^1(U_\ell)}^2 \int_{\R^\d} |\nab  \( h_{\vec{\eta}}^\mu[X_N]\circ \Phi_t^{-1}\)|^2 .$$ After another change of variables we find that 
$$ \int_{\R^\d}\left|E_{\veta}- \nab \( h_{\vec{\eta}}^\mu[X_N]\circ \Phi_t^{-1}\)\right|^2 \le C t^2 |\psi|_{C^1(U_\ell)}^2
\int_{\R^\d} |\nab h_{\vec{\eta}}^\mu[X_N]|^2 .$$
As seen just above,  $\nab \hat h$ is the $L^2$ projection onto gradients of $E_{\veta}$, hence we conclude that
\be\label{firsterm}\int_{\R^\d} |E_{\veta}- \nab \hat h|^2 \le \int_{\R^\d} \left|E_{\veta}- \nab \( h_{\vec{\eta}}^\mu[X_N]\circ \Phi_t^{-1}\)\right|^2 \le C t^2 |\psi|_{C^1(U_\ell) }^2
\int_{\R^\d} |\nab h_{\vec{\eta}}^\mu[X_N]|^2 .\end{equation}
\end{step}
}

\begin{step}[The main term]
The term $\Main$ is evaluated by a simple change of variables using \eqref{defiE}:
\be\label{2pm}
 \Main =\frac1{2\cd} \int_{\R^\d}  \left(|(D\Phi_t)^T \nab h_{\veta}^{\mu}|^2  |\det D \Phi_t^{-1}\circ  \Phi_t| - |\nab h^{\mu}_{\veta}|^2\right). \ee
Writing $\Phi_t= \id +t \psi$ and $D\Phi_t= \id + t D\psi$, we have  $ \Phi_t^{-1}= \id- t\psi + O(t^2 |\psi|^2)$ and  $$|\det D\Phi_t^{-1} \circ \Phi_t|= \det(\id-t D \psi (x+ t\psi(x)))+ O(t^2 |D\psi|^2)= 1-t \, \div \psi + O(t^2( |D\psi|^2 + |\psi|_{C^2}\|\psi\|_{L^\infty})) ,$$  thus 
\begin{equation} \label{mainenergyterm}
 \Main =\frac{t}{2\cd}  \int_{U_\ell}  \nab h_{\veta}^\mu \cdot    \mathcal A  \nab h_{\veta}^\mu +
t^2 O\( (|\psi|_{C^1(U_\ell)}^2 + |\psi|_{C^2(U_\ell)} \|\psi\|_{L^\infty(U_\ell)}  )\int_{\R^\d} |\nab h_{\vec{\eta}}^\mu|^2\)
\end{equation}
where \be\label{defA}
\mathcal A=2 D\psi-( \div \psi)  \id \ee and where the $O$ depends only on $\d$.

\end{step}

\begin{step}

 We now set to evaluate quantities of the form 
$$\int_{\R^\d} f \( \hat \delta_{y_i}- \delta_{y_i}^{(\eta_i)}\)$$ for general functions $f$.
We note that by \eqref{bpt} the supports of $\delta_{y_i}^{(\eta_i)}$ and $\hat \delta_{y_i}$ are included in an annulus of center $y_i$,  inner radius $\eta_i\(1-  |t| |\psi|_{C^1(B(y_i, \eta_i))}\)$ and outer radius $\eta_i \(1+  |t| |\psi|_{C^1(B(y_i, \eta_i))}\)$ and assume that 
\be\label{tpsi} |t| |\psi|_{C^1}<\hal.\ee

By definition of $\hat \delta_{y_i}$ \eqref{defdhat} we have 
\be\label{tolin}\int_{\R^\d} f \( \hat \delta_{y_i}- \delta_{y_i}^{(\eta_i)}\)= \dashint_{\partial B(x_i, \eta_i)} f\circ \Phi_t - \dashint_{\partial B(y_i, \eta_i)} f= \dashint_{\pa B(y_i, \eta_i) } f(\Phi_t(x+x_i-y_i)))-f.  \ee
Note that by definition of $\Phi_t$ and the definition of $y_i$ as $x_i+t\psi(x_i)$  we have 
\be \label{bpt}
\|\Phi_t (x +x_i-y_i)- x\|_{L^\infty(B(y_i, \eta_i))}= \| t\psi(x+x_i-y_i)- t\psi(x_i)\|_{L^\infty(B(y_i, \eta_i))}  \le t \eta_i  |\psi|_{C^1(B(x_i, \eta_i))}.\ee
It follows  that 
\begin{multline}\label{bornefgf}
\left|\int_{\R^\d} f\( \hat \delta_{y_i}- \delta_{y_i}^{(\eta_i)}\)\right|\le C |f|_{C^1(B(y_i, \frac32\eta_i))} \|\Phi_t (x +x_i-y_i)- x\|_{L^\infty(B(y_i, \eta_i))} \\ \le
Ct \eta_i  |\psi|_{C^1(B(x_i, \eta_i))} |f|_{C^1(B(x_i, \frac32\eta_i))} .
\end{multline}
Linearizing  \eqref{tolin} in $t$ we find that if $f \in C^{2}(B(y_i, \frac32\eta_i))$,
\begin{align}  \label{lint} \int_{\R^\d} f \( \hat \delta_{y_i}- \delta_{y_i}^{(\eta_i)}\)&=  
\dashint_{\pa B(y_i, \eta_i)} \nab f(x) \cdot \( \Phi_t(x+x_i-y_i) -x\) +O\(t^{2} \eta_i^{2} |f|_{C^{2} (B(y_i,\frac32\eta_i)) } |\psi|_{C^1(B(x_i, \eta_i))}^{2} \)   \\ \nonumber 
& =\dashint_{\pa B(y_i, \eta_i)} \nab f(x) \cdot \( -t\psi(x_i) +t\psi(x+x_i-y_i)\)\\
\nonumber  & \quad +O\(t^{2}  \eta_i^{2} |f|_{C^{2} (B(y_i,\frac32\eta_i)) } |\psi|_{C^1(B(x_i, \eta_i))}^{2} \)
\\ \nonumber
 &=t \dashint_{\pa B(x_i, \eta_i)} \nab f(x+y_i-x_i) \cdot \( \psi(x)-\psi(x_i)\) \\ \nonumber 
 & \quad +O\(t^{2} \eta_i^{2} |f|_{C^{2} (B(y_i,\frac32\eta_i)) }|\psi|_{C^1(B(x_i, \eta_i))}^{2} \).
\end{align}
\cor{Linearizing further $\psi$ and $\nab f$ \cord{and using that $\int_{\partial B(x_i, \eta_i)} D\psi(x_i)(x-x_i)=0$,} we may also get 
\begin{align}
\label{lint2}  \int_{\R^\d} f \( \hat \delta_{y_i}- \delta_{y_i}^{(\eta_i)}\)& = t \dashint_{ \pa B(x_i, \eta_i)} \nab f(x+y_i-x_i) \cdot  D\psi(x_i)(x-x_i) 
\\ \nonumber & + O\( t\eta_i^2  |\psi|_{C^2(B(x_i, \eta_i)) }  |f|_{C^1(B(y_i,\frac32\eta_i))}+ t^{2}  \eta_i^2  |f|_{C^{2} (B(y_i,\frac32\eta_i)) } |\psi|_{C^1(B(x_i, \eta_i))}^{2} \) \\
 \nonumber & =  O\( t \eta_i^2  \(  |f|_{C^2(B(y_i,\frac32\eta_i))}   |\psi|_{C^1(B(x_i, \eta_i))}  + |\psi|_{C^2(B(x_i, \eta_i)) }   |f|_{C^1(B(y_i,\frac32\eta_i))}\)  \).
\end{align}}

\end{step}

\begin{step}[The remainder term]
Let us denote \be \label{defvi}
v_i= \g* \(\hat \delta_{y_i}-  \delta_{y_i}^{(\eta_i)}\).\ee
By \eqref{defhh}
we have  \cord{\begin{equation}\label{hchapeau}
 \hat h = h_{\veta}^\nu[Y_N]+ \sum_{i=1}^N v_i.\end{equation} }Thus, integrating by parts we find 
\begin{align}\label{reecrrem} 2\cd\Rem & =   \sum_{i,j}\int \nab v_i \cdot \nab v_j -2\sum_i \int \nab v_i \cdot \nab\hat h
\\ \nonumber & = 
\cd  \sum_{i,j} \int v_i \(  \hat  \delta_{y_j} - \delta_{y_j}^{(\eta_j)} \)- 2\cd \sum_i   \int v_i\(\hat \delta_{y_i}+ \sum_{j\neq i} \hat \delta_{y_j} - N\nu\)  \\ \nonumber 
&  = -  \cd \sum_i \int v_i (  \hat \delta_{y_i} + \delta_{y_i}^{(\eta_i)}) 
\cord{-}\cd  \sum_{i\neq j} \int v_i \(  \hat  \delta_{y_j} - \delta_{y_j}^{(\eta_j)}  \) - 2\cd \sum_{i} \int v_i \(\sum_{j\neq i} \cord{  \delta_{y_j}^{(\eta_j)}} - N\nu\)\\ \nonumber
:& = 2\cd(\Rem_1+ \Rem_2+\Rem_3).
\end{align}

{\bf Substep \cord{5.1}.}\\
Let us start by  $\Rem_1$. We observe that  
  \begin{multline*}
  \Rem_1=\frac{1}{2} \sum_i \int \g* \delta_{y_i}^{(\eta_i)} \delta_{y_i}^{(\eta_i)} - \int \g*\hat \delta_{y_i}\hat  \delta_{y_i}\\
  =\frac1{2}\sum_i
 \dashint_{\p B(y_i, \eta_i)}   \dashint_{\p B(y_i, \eta_i)} \( \g(\Phi_t(x)-\Phi_t(y)) -\g(x-y) \) dxdy.
 \end{multline*}
 Breaking the double integral into $|x-y|>\delta$ and $|x-y|\le \delta$ we may write 
 that 
\begin{multline*}
 \hal \dashint_{\p B(y_i, \eta_i)}   \dashint_{\p B(y_i, \eta_i)} \( \g(\Phi_t(x)-\Phi_t(y)) -\g(x-y) \) dxdy.
 \\
 = O\( \dashint_{(\p B(y_i, \eta_i))^2,   |x-y|\le C \delta } \g(x-y) dxdy\)+ \frac{t}2 \dashint_{(\p B(y_i, \eta_i))^2, |x-y|> \delta}  
   \nab \g(x-y) \cdot (\psi(x)-\psi(y)) \\+O\(t^2|\psi|_{C^1(B(y_i, \eta_i))}^2 \dashint_{(\pa B(y_i, \eta_i))^2,C |x-y|>\delta}  |x-y|^{2-\d}\).
\end{multline*}
Letting $\delta\to 0$ we find 
\begin{multline*}\hal
 \dashint_{\p B(y_i, \eta_i)}   \dashint_{\p B(y_i, \eta_i)} \( \g(\Phi_t(x)-\Phi_t(y)) -\g(x-y) \) dxdy
 \\=\frac{t}{2} \dashint_{(\p B(y_i, \eta_i))^2}  
   \nab \g(x-y) \cdot (\psi(x)-\psi(y)) +O\(t^2|\psi|_{C^1(B(y_i, \eta_i))}^2 \dashint_{(\pa B(y_i, \eta_i))^2}  |x-y|^{2-\d}\).
\end{multline*}
  With this we claim that 
 \be \label{diffg}\Rem_1 =\frac{ t}{2}\sum_{i=1}^N  \eta_i^{1-\d} \dashint_{\pa B(y_i, \eta_i)} (\psi(x)-\psi(y_i))  \cdot \nu  +O\(t^2\sum_{i=1}^N \eta_i^{2-\d}   |\psi|_{C^1(B(y_i, \eta_i))}^2 \)\ee
 where $\nu$ is the outer unit normal.
 To see this, introduce $\g_{\eta_i}= \g * \delta_{y_i}^{(\eta_i)}$ and observe that  by splitting $\psi(x)-\psi(y) $ into $\psi(x)-\psi(y_i) + \psi(y_i) -\psi(y)$ and symmetrizing the variables
 \begin{multline*}
 \hal\dashint_{\p B(y_i, \eta_i)}  \dashint_{\p B(y_i, \eta_i)}  \nab \g(x-y) \cdot (\psi(x)-\psi(y)) 
\\ =  \dashint_{\p B(y_i, \eta_i)}  \dashint_{\p B(y_i, \eta_i)}  \nab \g(x-y) \cdot (\psi(x)-\psi(y_i)) =  \dashint_{\p B(y_i, \eta_i)} \nab \g_{\eta_i}(x-y_i) \cdot  (\psi(x)-\psi(y_i)) \, dx
\end{multline*}
and the right-hand side is equal to 
$$\hal \int_{\R^\d}\( \psi(x)-\psi(y_i) \) \cdot \div T_{\nab \g_{\eta_i}}$$
where $\div$ here is a vector-valued divergence, and for any function $h$ we let  $T_{\nab h}$ denote the stress-energy tensor
$$T_{\nab h}:= 2 (\nab h)\otimes (\nab h) - |\nab h|^2  \id ,$$ see for instance \cite[Lemma 4.2]{smf}. We have the identity  $\div T_{\nab h}= 2 \nab h \Delta h$ for smooth functions,  and then notice 
 that $\div T_{\nab \g_{\eta_i}}=0$  away from $\pa B(y_i, \eta_i)$ so that each component of $\div T_{\nab \g_{\eta_i}}$ is the jump of normal component of the corresponding row of $T_{\nab \g_{\eta_i}}$. Since $T_{\nab \g_{\eta_i}}$ jumps from $0$ inside $B(y_i, \eta_i)$ to $ T_{\nab \g}$ outside $B(y_i, \eta_i)$, the integral transforms into a boundary integral equal to that of  \eqref{diffg}.
 
\smallskip

{\bf Substep \cord{5.2.}}\\
We next turn to $\Rem_2$.
\cor{First we estimate $v_i$ defined in \eqref{defvi}.  
Using \eqref{lint2} with $f=\g(x-\cdot)$ and $f= \nab \g(x-\cdot)$ we obtain 
\begin{align}
\label{bornedvi}
 &  \forall x\notin B(y_i, 2 \eta_i ), \quad | v_i|(x) \le C |  t| \eta_i^2  \( \frac{1}{|x-y_i|^{\d}}  |\psi|_{C^1(B(x_i , \eta_i))} + \frac{1}{|x-y_i|^{\d-1}}  |\psi|_{C^2(B(x_i, \eta_i))} \) \\
   \label{bornedvi2} & \forall x\notin B(y_i, 2 \eta_i ), \quad |\nab v_i|(x) \le C| t| \eta_i^2 \( \frac{1}{|x-y_i|^{\d+1}}   |\psi|_{C^1(B(x_i , \eta_i))} + \frac{1}{|x-y_i|^\d} |\psi|_{C^2(B(x_i, \eta_i))} \),
\end{align}
}hence inserting into \eqref{bornefgf}, we find
\be
\left|\sum_i \int_{\R^\d} v_i   \sum_{j\neq i}  \(\hat \delta_{y_j}-\delta_{y_j}^{(\eta_j)} \)\right|
\le    C t^2 \sum_{i\neq j}  \eta_i^2 \eta_j  \( \frac{   |\psi|_{C^1  (B(x_i, \eta_i)) }^2  }{|y_i-y_j|^{\d+1}}+ \frac{   |\psi|_{C^1  (B(x_i, \eta_i)) } |\psi|_{C^{2}(B(x_i, \eta_i)) }  }{|y_i-y_j|^{\d} } \)
.\ee
We next split the sum into pairs at distance $ \le N^{-1/\d}$ which we control by Corollary \ref{coro3} using that  $$\frac{\eta_i^2 \eta_j}{|y_i-y_j|^{\d+1}} \le \frac{1}{|y_i-y_j|^{\d-2}}, \quad \frac{\eta_i^2\eta_j}{|y_i-y_j|^\d} \le \frac{N^{-1/\d} }{|y_i-y_j|^{\d-2}}$$ since $\eta_i \le \rr_i$, and pairs at distance $\ge N^{-1/\d}$ for which we use \cor{$\eta_i^2\eta_j\le N^{-3/\d}$} and use Proposition \ref{multiscale} \cor{applied with $s=2,3$}. This way, absorbing some terms and using that $\ell \ge N^{-1/\d}$, we conclude that 
\begin{multline} \label{rem2}
|\Rem_2 |\le C t^2 \cor{\( |\psi|_{C^1(U_\ell)}^2 +N^{-\frac1\d}\log (\ell N^{\frac1\d})  |\psi|_{C^1(U_\ell)} |\psi|_{C^2(U_\ell)}\) 
} \\ \times \(  \F^{U_\ell}(Y_N, \nu) +  \(\frac{\#I_{U_\ell}}{4} \log N\) \indic_{\d=2}+C_0  \#I _{U_\ell} N^{1-\frac2\d}\).
\end{multline}

{\bf Substep \cord{5.3}.}\\
We finish by analyzing the term $\Rem_3$.
By integration by parts, we may write 
\begin{equation*}
 \Rem_3=    \int_{\R^\d}    \sum_{i=1}^N \( \delta_{y_i}^{(\eta_i)}- \hat \delta_{y_i}\) \tilde h_{i,\veta} [Y_N]
  \end{equation*}  where $\tilde  h_{i,\veta} [Y_N]= h^\nu_{\veta}[Y_N] -\g(\cdot -y_i)$. 
      We also let $\tilde h_i [Y_N]=h^\nu[Y_N]-\g(\cdot -y_i) $ and observe that 
 $\tilde h_i[Y_N] $ and $\tilde h_{i, \veta}[Y_N]$ coincide in $B(y_i, \eta_i)$.  Using \eqref{lint}, \eqref{contrhi2} \cord{, $\eta_i\le \rr_i$}  and Young's inequality we deduce
\begin{multline*}\sum_{i=1}^N \int_{\R^\d} \tilde h_{i,\veta}[Y_N] \( \delta_{y_i}^{(\eta_i)}- \hat \delta_{y_i}\)
- t\sum_i  \dashint_{\pa B(x_i, \eta_i)} \nab \tilde h_i[Y_N] (x+y_i-x_i) \cdot (\psi(x)-\psi(x_i))\\
=
 O\(t^{2} 
\cord{\sum_{i=1}^N} |\tilde h_i[Y_N]|_{C^{2} ( B(\cord{y_i}, 2\eta_i) ) } \eta_i^2 |\psi|_{C^1(B(x_i, \eta_i)) }^2\)
\\= O\(  t^2    \sum_{\cord{i\in I_{U_\ell}}} |\psi|_{C^1(B(x_i,\eta_i) ) }^2  \eta_i^2  \( C\rr_i^{-1-\frac{\d}{2} } \lambda_i(Y_N,\nu)^\hal +  N (\|\mu\|_{L^\infty} +\cord{\rr_i} |\mu|_{C^{\cord{1}} (B(\cord{y_i}, \rr_i)} )\) 
\)\\ = O\( t^2 \sum_{\cord{i\in I_{U_\ell}}} |\psi|_{C^1(B(x_i,\rr_i) )}^2 \( \rr_i^{2-\d} + \int_{B(\cord{y_i}, \rr_i)} |\nab \cord{h_{\vec{\eta}}^{\nu}}[Y_N]|^2 + N^{1-\frac2\d} ( \|\mu\|_{L^\infty} + \cord{N^{-\frac1\d}} |\mu|_{C^{\cord{1}} (B(\cord{y_i}, \rr_i)} \)\) .
\end{multline*}
Using \cord{\eqref{11} in the case $\d\ge 3$ (and the fact that the $\rr_i$'s computed for $X_N$ and for $Y_N$ differ by at most a multiplicative factor of $2$)} we conclude that 
\begin{multline}\label{rem3}\Rem_3= t \sum_{i=1}^N \dashint_{\pa B(x_i, \eta_i)}  \nab \tilde h_i [Y_N] (x+y_i-x_i) 
   \cdot \( \psi(x)-\psi(x_i)\) \\+ t^2 O 
   \( |\psi|_{C^1(U_\ell)}^2 \cord{(1+\|\mu\|_{L^\infty}+N^{-\frac1\d}|\mu|_{C^1(U_\ell)} )}\) \\
   \times
\(  \F^{U_\ell} (Y_N, \nu) +  \(\frac{\#I_{U_\ell}}{4} \log N \) \indic_{\d=2}+ C_0 \#I _{U_\ell} N^{1-\frac2\d}   \cord{+ \int_{U_\ell}|\nab h_{\vec{\eta}}^\nu[Y_N]|^2  }\) .\end{multline}
\cord{where the $O$ depends only on $\d$.}

\end{step}

\begin{step}[The error term \eqref{B}]
First we write 
$$ - \int_{\R^\d} \f_{\eta_i} (x-y_i) d\nu+ \int_{\R^\d} \f_{\eta_i} (x-x_i) d\mu=
-\int_{\R^\d} \f_{\eta_i} (x-x_i) \( d\nu (x+x_i-y_i) - d\mu(x) \).$$
Then we may write $\nu(x-x_i+y_i)= \hat \Phi_t \#\nu=\hat \Phi_t( \#\Phi_t \#\mu)=( \hat \Phi_t \circ \Phi_t )\# \mu$,
where we let $\hat \Phi_t=\id + x_i-y_i= \id - t\psi(x_i)$. Since $\hat \Phi_t \circ \Phi_t= \id + t(\psi-\psi(x_i))$, we may write  in view of Lemma \ref{linea} that
 $$\nu(x+x_i-y_i)=
 \mu -t \div((\psi-\psi(x_i))  \mu) + u$$
 with $$\|u\|_{L^\infty} \le C t^2 \(|\mu|_{C^2} \|\psi-\psi(x_i) \|_{L^\infty}^2 + |\psi|_{C^1}^2 + |\psi|_{C^2}\|\psi-\psi(x_i)\|_{L^\infty} \).$$
Thus 
\begin{multline*}- \int_{\R^\d} \f_{\eta_i} (x-y_i) d\nu+ \int_{\R^\d} \f_{\eta_i} (x-x_i) d\mu=- 
t\int_{\R^\d} \nab \f_{\eta_i} \cdot (\psi(x)-\psi(x_i)) d\mu \\
+t^2 O\(\(1+ \cord{\eta_i^2} |\mu|_{C^2}) |\psi|_{C^1}^2   +  \eta_i |\psi|_{C^2} |\psi|_{C^1}  \) \int_{\R^\d} |\f_{\eta_i}|\).
\end{multline*}
Using \eqref{intf}, summing over $i$ and using $\eta_i \le N^{-1/\d}$ we find 
\begin{multline}
\label{main4}\mathrm{Err}=- t N \sum_{i=1}^N \int_{\R^\d}\nab \f_{\eta_i} \cdot (\psi(x)-\psi(x_i)) d\mu (x)
\\+ t^2 O\( \sum_{i=1}^N \cord{(1+N^{-\frac2\d} |\mu|_{C^2(B(x_i, \eta_i))})} (|\psi|_{C^1(B(x_i, \eta_i) ) }^2  + |\psi|_{C^1(B(x_i ,\eta_i))} |\psi|_{C^2(B(x_i, \eta_i) ) } N^{-\frac1\d} )  N^{1-\frac2\d}  \),\end{multline}
\cord{where $O$ depends only on $\d$.}

\end{step}

\begin{step}[Conclusion]
We now define 
\begin{multline}\label{defL}
\cord{L_t}:= 
\frac{1}{2\cd} \int \nab h_{\veta}^\mu \cdot (2D\psi - (\div \psi) \id) \nab h_{\veta}^\mu[X_N]
+ 
 \sum_{i=1}^N \dashint_{\pa B(x_i, \eta_i)}  \nab \tilde h_i  [Y_N] (x+y_i-x_i) 
   \cdot \( \psi(x)-\psi(x_i)\) 
\\+\hal \sum_{i=1}^N   \dashint_{\pa B(y_i, \eta_i)}   \eta_i^{1-\d}  \( (\psi(x)-\psi(y_i))\cdot \nu\)  -N \sum_{i=1}^N  \int_{\R^\d} \nab \f_{\eta_i} \cdot (\psi(x)-\psi(x_i)) d\mu.  
\end{multline}
Combining \eqref{55},  \cord{\eqref{firsterm},} \eqref{B}, \eqref{mainenergyterm}, \eqref{rem2}, \eqref{rem3} and \eqref{main4}, we find  
\begin{multline}\label{prtrans}
\F_N(\YN,\nu)- \F_N(\XN,\mu)\\=t\cord{L_t} + t^2 O\(  (|\psi|_{C^1(U_\ell)}^2 + |\psi|_{C^2(U_\ell)}\|\psi\|_{L^\infty(U_\ell) } ) \int_{\R^\d}    |\nab h_{\veta}^\mu[\XN]|^2 \) \\
 +
\cord{\Bigg(}t^2 \cord{\Big(} |\psi|_{C^1(U_\ell)}^2    \cord{(1+\|\mu\|_{L^\infty} + N^{-\frac1\d} |\mu|_{C^1(U_\ell)}+N^{-\frac2\d}|\mu|_{C^2(U_
\ell)})}\\
\qquad +|\psi|_{C^1(U_\ell)} |\psi|_{C^2(U_\ell)}\cord{ (1+N^{-\frac2\d}|\mu|_{C^2(U_\ell)})} N^{-\frac1\d} \cor{\log (\ell N^{\frac1\d})} \Bigg)
 \\
 \times O \(  \F^{U_\ell} (Y_N, \nu)+ \(\frac{ \#I_{U_\ell}}{4} \log N\) \indic_{\d=2}+C_0 \#I_{U_\ell}  N^{1-\frac2\d}  \cord{+ \int_{U_\ell} |\nab h_{\vec{\eta}}^\nu[Y_N]|^2}\) 
\end{multline} where the $O$ depends \cord{only on $\d$}.

In particular 
$$\F_N(\YN,\nu)- \F_N(\XN,\mu)\cord{= t L_t}+ o(t).$$
Comparing with Proposition~\ref{proptransport2}, we find  that  letting $\Xi$ be as in \eqref{defvarphi}, we have 
\begin{multline}\label{defanil}\Xi'(0)=\Ani_1(\XN, \psi, \mu)=\lim_{t\to 0}  \cord{L_t}\\
=\frac{1}{2\cd} \int_{\R^\d} \nab h_{\veta}^\mu \cdot (2D\psi - (\div \psi) \id) \nab h_{\veta}^\mu
+ 
 \sum_{i=1}^N  \dashint_{\pa B(x_i, \eta_i)}  \nab \tilde h_i[\XN]  
   \cdot \( \psi(x)-\psi(x_i)\) 
\\+\hal \sum_{i=1}^N  \dashint_{\pa B(x_i, \eta_i)}   \eta_i^{1-\d}  \( (\psi(x)-\psi(x_i))\cdot \nu\)  - N \sum_{i=1}^N \int_{\R^\d} \nab \f_{\eta_i} \cdot (\psi(x)-\psi(x_i)) d\mu.  
\end{multline}
In addition, this implies that the quantity in the right-hand side is independent of the choice of $\veta$ as long as $\eta_i \le \rr_i$. 
Taking $\eta_i= \frac14 \rr_i$ and bounding the terms in \eqref{defanil}  we obtain 
\begin{multline}
|\Ani_1(\XN, \mu, \psi)|\le C \Big( \int_{\R^\d} |\nab h_{\frac14 \rr}^\mu |^2|D\psi|
+  \sum_{i=1}^N  \rr_i \| \nab \tilde h_i \|_{L^\infty(B(x_i, \frac14 \rr_i))}|\psi|_{C^1(B(x_i , \frac14 \rr_i))}
\\+\sum_{i=1}^N  \rr_i^{2-\d} |\psi|_{C^1(B(x_i , \frac14\rr_i))}  + \sum_{i=1}^N      |\psi|_{C^1(B(x_i , \frac14 \rr_i))}    \|\mu\|_{L^\infty} N^{1-\frac2\d}\Big) \end{multline}
Using \eqref{contrhi} and  Young's inequality, we deduce that 
\begin{multline}\label{pourthomas0}
|\Ani_1(\XN, \mu, \psi)|\le C \int_{\R^\d} |\nab h_{\frac14 \rr}^\mu |^2|D\psi|
\\+ C \sum_{i=1}^N |\psi|_{C^1(B(x_i , \frac14 \rr_i))}
\( \rr_i^{2-\d}+ \int_{B(x_i, \rr_i)} |\nab h_{\frac14 \rr}^{\mu}[\XN]|^2 + N^{1-\frac2\d} \|\mu\|_{L^\infty} )\)
  \end{multline}
  which in view of Lemma \ref{lem:contrdist1} proves \eqref{pourthomas}, from which it also follows that $$|\Xi'(0)|\le C \cord{(1+\|\mu\|_{L^\infty})}  |\psi|_{C^1(U_\ell)} \Xi(0),$$ \cord{where $C$ depends only on $\d$.}
  By  the same reasoning, for every  $t$ such that \eqref{tpsi} holds, and  since $|\psi \circ \Phi_t^{-1}|_{C^1 (U_\ell)} \le |\psi|_{C^1(U_\ell)} (1+ C |t| |\psi|_{C^1(U_\ell)})$, we have $|\Xi'(t)|\le C \cord{(1+\|\mu\|_{L^\infty}) }  |\psi|_{C^1}  \Xi(t)$.
Thus  applying Gronwall's lemma we deduce that if \eqref{tpsi} holds we  have
\be\label{diffpath}
|\Xi(t)-\Xi(0) |\le C  t \Xi(0)\ee and thus also 
\be\label{diffpath2}\Xi(t)\le C \Xi(0)\ee \cord{where $C$ depends only on $\d$ and $\|\mu\|_{L^\infty}$,} proving \eqref{p41}.

\end{step}

\subsection{Estimating the second derivative}

We now wish to bound $|\Xi''(t)|$, which is new compared to \cite{ls2}.  \cord{We will repeatedly use \eqref{diffpath2}, i.e. that the energy for $Y_N$ is controlled by that of $X_N$.}

\noindent
\cord{{\bf Step 1}  (Choice of $\eta_i$).  We are going to make a different, smaller, choice of $\eta_i\le \frac14\rr_i$, for reasons that will appear later, and set
\begin{equation}\label{choiceetai}
\eta_i= \frac{m}{4} \min \( (\rr_i N^{\frac1\d})^{\alpha}   , 1\) \rr_i,\end{equation}
for 
\begin{equation}
\label{defm}m=  \begin{cases} (N^{\frac{1}{\d}}\ell)^{-1} &   \text{if} \ \d=2\\
(N^{\frac1\d}\ell)^{-\alpha'}, \alpha'> 0 & \text{if } \ \d\ge 3
\end{cases}\end{equation}
and \begin{equation}\label{defalpha}
\alpha = \begin{cases}\hal & \text{if} \  \d=2\\  0 & \text{if} \ \d\ge3.\end{cases}
\end{equation}
 This change in $\vec{\eta}$ leads us to a modification in the  evaluation of  terms $\int |\nab h_{\vec{\eta}}|^2$ and $\sum_{i} \eta_i^{2-\d}$.  In particular, using the definition \eqref{choiceetai} and \eqref{11}, we may write if $\d \ge 3$
\begin{equation}\label{newboundg3}\sum_{i\in I_N} \g(\eta_i) \le C m^{2-\d} \Xi(0)
\end{equation} and if $\d =2$, using the additivity property of the log,
\begin{equation}\label{newboundg2} \sum_{i\in I_N} \g(N^{\frac1\d}\eta_i) \le \sum_{i\in I_N} \g\( \frac{m}4 (\rr_i N^{\frac1\d}) ^{1+\alpha}\)    \le  C \sum_{i\in I_N} \g (\rr_i N^{\frac1\d}) + C \#I_N \g(m) \le C \Xi(0)+ C \#I_N \g(m)  .\end{equation} 
On the other hand, in view of  \eqref{15} (and adding resp. subtracting $\sum_{i\in I_N}\g(N^{\frac1\d})$ to the terms in the right-hand side in the case $\d=2$) we have 
\begin{multline*}
\int_{U_\ell}  |\nab h_{\vec{\eta}}|^2 \le 2 \cd  \( F_N^{U_\ell}(X_N, \mu) + \frac{\#I_N}{4}( \log N)\indic_{\d=2}   + C_0 \#I_N N^{1-\frac2\d}\)    \\ + \cd \sum_{i\in I_N} \g(N^{\frac1\d} \eta_i)  \indic_{\d=2}+ \cd \sum_{i \in I_N} \g(\eta_i) \indic_{\d\ge 3}\end{multline*} so that by \eqref{11} and \eqref{newboundg3}, resp. \eqref{newboundg2} and rearranging terms, we find 
\begin{equation}\label{conthe}
\int_{U_\ell}  |\nab h_{\vec{\eta}}|^2\le C \Xi(0) (1+ m^{2-\d} \indic_{\d\ge 3}) +  C\#I_N \g(m) \indic_{\d=2}\end{equation}
{\bf Step 2} (Main term) 
  To bound $|\Xi''(t)|$ we  first need to estimate $|L_t-L_0|$. First, let us} evaluate the Lipschitz norm  in $t$  of $\int  \nab h_{\veta}^\mu [\XN] \cdot \( \mathcal A \nab h_{\veta}^\mu[\XN]\) $  or more precisely bound 
$$\int_{\R^\d}  \nab h_{\veta}^\nu [Y_N] \cdot \( \mathcal A \nab h_{\veta}^\nu[Y_N] \) -  \int_{\R^\d} \nab h_{\veta}^\mu [\XN]\cdot\( \mathcal A \nab h_{\veta}^\mu[\XN]\) $$ where $\mathcal A =  2D\psi-(\div \psi) \id$.

We start by observing that, \cord{using Young's inequality and $t|\psi|_{C^1}\le \hal$,}
\begin{multline*}
\left|\int_{\R^\d} \nab \hat h \cdot\mathcal A \nab \hat h   - \int_{\R^\d} E_{\veta} \cdot\mathcal A  E_{\veta}\right|\le C |\psi|_{C^1}
\|\nab \hat h - E_{\veta}\|_{L^2(U_\ell)}\cord{\( \|\nab \hat h\|_{L^2(U_\ell)} +\|\nab \hat h- E_{\veta}\|_{L^2(U_\ell)}\)} \\\le \frac{1}{4\cd |t|} \int_{U_\ell} |\nab \hat h - E_{\veta}|^2 + C|t| |\psi|_{C^1}^2 
\int_{U_\ell} |\nab \hat h|^2 .
\end{multline*}
 To control  $\int_{\R^\d} |\nab \hat h|^2 $  we use \eqref{defrem} and the bounds on $\Rem_1, \Rem_2, \Rem_3$ \cord{\eqref{diffg}, \eqref{rem2} and \eqref{rem3}} obtained previously to get \cord{
 \begin{align*} \frac1{2\cd} \int_{\R^\d}  |\nab \hat h|^2 &\le
 \frac{1}{2\cd} \int_{\R^\d} |\nab h_{\vec{\eta}}^\nu[Y_N]|^2 +C|t||\psi|_{C^1}\sum_{i\in I_N} \eta_i^{2-\d}
 \\
& + C t^2 \( |\psi|_{C^1}^2 (1+\|\mu\|_{L^\infty} + N^{-\frac1\d}|\mu|_{C^1(U_\ell)}) +N^{-\frac1\d}\log (\ell N^{\frac1\d})  |\psi|_{C^1} |\psi|_{C^2}\) 
 \\
& \times \(  \F^{U_\ell}(Y_N, \nu) +  \(\frac{\#I_{U_\ell}}{4} \log N\) \indic_{\d=2}+C_0  \#I _{N} N^{1-\frac2\d} + \int_{U_\ell}|\nab h_{\vec{\eta}}^\nu[Y_N]|^2 \)
 \\&+ t  |\psi|_{C^1} \sum_{i\in I_N} \eta_i \|\nab \tilde h_i\|_{L^\infty(B(y_i, \rr_i))} .
  \end{align*} 
  Inserting \eqref{contrhi}, \eqref{conthe} and \eqref{newboundg3},  and absorbing terms via $t|\psi|_{C^1}<\hal$,  we obtain
  \begin{multline*} 
  \frac1{2\cd} \int_{\R^\d}  |\nab \hat h|^2 \le
 C \(\Xi(0) (1+   m^{2-\d} \indic_{\d\ge3}) +  \# I_N \g(m)\indic_{\d=2}\) \\ \times\(1+ C \(| t| |\psi|_{C^1} + t^2 |\psi|_{C^1}^2 N^{-\frac1\d} |\mu|_{C^1(U_\ell)} +  t^2 |\psi|_{C^1} |\psi|_{C^2} N^{-\frac1\d} \log (\ell N^{\frac1\d}) \)\)  \end{multline*}
 with $C$ depending only on $\d$ and $\|\mu\|_{L^\infty}$.
When   $t|\psi|_{C^1}$ and $t|\psi|_{C^2} N^{-\frac1\d}\log (\ell N^{\frac1\d})$ are small enough, in view of \eqref{firsterm}, it follows that 
 \begin{multline}\label{autili}
  \left|\int_{\R^\d} \nab \hat h \cdot\mathcal A \nab \hat h   - \int_{\R^\d}  E_{\veta} \cdot\mathcal A  E_{\veta}\right|\\
  \le
C |t| |\psi|_{C^1}^2    \( \Xi(0) (1+ C m^{2-\d} \indic_{\d\ge3} )+  \# I_N \g(m)\indic_{\d=2} \) +O(t^3)
 \end{multline}
 where $C$ depends only on $\d$ and $\|\mu\|_{L^\infty}$.}
Next, we  show that \cord{
\begin{multline} \label{conclst12} \left|\int_{\R^\d} \nab \hat h \cdot \mathcal A \nab \hat h- \int_{\R^\d} \nab h_{\veta}^\nu [Y_N]\cdot \mathcal A \nab h_{\veta}^\nu [Y_N]\right|\\ \le
C|t| 
\( |\psi|_{C^1}^2    +|\psi|_{C^1} |\psi|_{C^2 } N^{-\frac1\d} \log (\ell N^{\frac1\d}) \)\\     
\( \Xi(0) (1+ m^{2-\d} \indic_{\d\ge 3})    + C \#I_N \g(m)\indic_{\d=2}    \)+O(t^2),
 \end{multline}}
 where the $O(t^2)$ depends on $X_N$, $\vec{\eta}$ and $\mu$.
First we claim that 
\cord{\begin{multline}\label{a62} \int_{\R^\d} |\nab (\hat h-h_{\veta}^\nu[Y_N])|^2 \\ \le C 
\( |t| |\psi|_{C^1} + t^2  N^{-\frac1\d} \log(\ell N^{\frac1\d}) |\psi|_{C^1}|\psi|_{C^2} \)
\Xi(0) +    C |t| |\psi|_{C^1} m^{2-\d}  \Xi(0) \indic_{\d \ge 3}.
\end{multline}}
Indeed $h_{\veta}^\nu -\hat h= \sum_{i=1}^N  v_i$ (\cord{see \eqref{hchapeau}}) hence 
\be \int_{\R^\d} |\nab (\hat h-h_{\veta}^{\cord{\nu}})|^2 = \sum_{i=1}^N \int_{\R^\d} |\nab v_i|^2  + \sum_{i\neq j} \int_{\R^\d} \nab v_i \cdot \nab v_j.\ee
But the second term on the right-hand side is exactly $\Rem_2$ so from  \eqref{rem2} we get \be\label{prhh} \left|  \int_{\R^\d} |\nab (\hat h-h_{\veta}^\nu[Y_N] )|^2 - \sum_{i=1}^N \int_{\R^\d} |\nab v_i|^2  \right|\le  C t^2 \cor{ \(|\psi|_{C^1}^2  + N^{-\frac1\d} \log(\ell N^{\frac1\d}) |\psi|_{C^1}|\psi|_{C^2} \)}  \cord{\Xi(0)}.\ee

On the other hand \begin{multline*}
\int_{\R^\d} |\nab v_i|^2 = \int_{\R^\d} \g*(\delta_{y_i}^{(\eta_i)} - \hat \delta_{y_i}) (\delta_{y_i}^{(\eta_i)} - \hat \delta_{y_i}) \\ =-
\int_{\R^\d} \g* \delta_{y_i}^{(\eta_i)}\delta_{y_i}^{(\eta_i)}  +\int_{\R^\d} \g* \hat \delta_{y_i} \hat \delta_{y_i}+2\int_{\R^\d} \g*\delta_{y_i}^{(\eta_i)}\( \delta_{y_i}^{(\eta_i)}-\hat \delta_{y_i}\) . \end{multline*} But $\g* \delta_{y_i}^{(\eta_i)}= \g_{\eta_i} $ by definition, it satisfies $|D\g_{\eta_i}|\le \eta_i^{1-\d}$ so by \eqref{bornefgf}   the last term is bounded by $ Ct  |\psi|_{C^1}\eta_i^{2-\d}$.
In addition the calculation of $\Rem_1$ also gives us 
$O(t |\psi|_{C^1} \eta_i^{2-\d})$ for the first two terms.
We conclude \cord{directly for $\d=2$ and  in view of \eqref{newboundg3}  for $\d\ge 3$} that \cord{\eqref{a62} holds. }

Next, we note that by \eqref{bornedvi2},  and $\eta_i \le N^{-\frac1\d}$, 
\begin{multline} \int_{B(y_i, 2\eta_i)^c} |\nab v_i|^2 \le \cor{  Ct^2\eta_i^4\(  |\psi|_{C^1(B(x_i, \eta_i))}^2 \int_{|x|\ge 2 \eta_i}  \frac{dx}{ |x|^{2\d+2}}+      |\psi|_{C^2(B(x_i, \eta_i)) }^2 \int_{|x|\ge 2\eta_i}   \frac{dx}{ |x|^{2\d} }  \) } \\
\cor{\le  C t^2 \(  |\psi|_{C^1(B(x_i, \eta_i))}^2 \eta_i^{2-\d} +   |\psi|_{C^2(B(x_i, \eta_i)) }^2 N^{-\frac2\d}\eta_i^{2-\d} \)} 
.\end{multline}
Thus, using again \cord{\eqref{newboundg3} and combining with \eqref{prhh} we deduce that 
\begin{multline}\label{prelimb} \left|\int_{\R^\d} |\nab (\hat h- h_{\veta}^\nu[Y_N])|^2 - \sum_{i=1}^N \int_{B(y_i, 2 \eta_i)} |\nab v_i|^2\right| 
\le 
C t^2  \(|\psi|_{C^1}^2  + N^{-\frac1\d} \log(\ell N^{\frac1\d}) |\psi|_{C^1}|\psi|_{C^2}  \)\Xi(0)\\+ C t^2 ( |\psi|_{C^1}^2+ |\psi|_{C^2}^2 N^{-\frac2\d})   \(  \Xi(0) ( 1+   m^{2-\d}   \indic_{\d \ge 3})\) .
\end{multline}}
We can now evaluate, \cord{using \eqref{hchapeau},}
\begin{multline} \label{prhhh}
\int_{B(y_i, 2\eta_i)} |\nab (\hat h- h_{\veta}^\nu[Y_N] - v_i)|^2 
\\ = \int_{B(y_i, 2\eta_i)} |\nab (\hat h-h_{\veta}^\nu[Y_N] )|^2 -  \int_{B(y_i, 2\eta_i)}|\nab v_i|^2-2   \int_{B(y_i, 2\eta_i)}  \nab v_i \cdot (\sum_{j\neq i} \nab v_j)\\
\cord{\le \int_{\R^\d} |\nab (\hat h-h_{\veta}^\nu[Y_N] )|^2 -  \int_{B(y_i, 2\eta_i)}|\nab v_i|^2-2   \int_{B(y_i, 2\eta_i)}  \nab v_i \cdot (\sum_{j\neq i} \nab v_j).}
\end{multline}
\cord{But using that $v_j$ is harmonic in $B(y_i, 2\eta_i)$ for $j\neq i$ and \eqref{bornedvi}, \eqref{bornedvi2}, we may write}
\begin{multline*}\left|\sum_{j\neq i} \int_{B(y_i, 2\eta_i)} \nab v_i \cdot \nab v_j\right|
= \left|\cord{\sum_{j\neq i} }\int_{\pa B(y_i, 2\eta_i)} v_i \frac{\pa v_j}{\pa \nu}\right| \\
\le  
\cor{C
t^2 \eta_j^2\sum_{j\neq i} \( \eta_i |\psi|_{C^1} + \eta_i^2 |\psi|_{C^2} \) \( \frac{|\psi|_{C^1} }{|y_j-y_i|^{\d+1}} + \frac{|\psi|_{C^2}}{|y_i-y_j|^\d}  \)} . \end{multline*}
\cor{After summing over $i$ we may control this term similarly as we did for \eqref{rem2} via Proposition~\ref{multiscale} combined with Corollary \ref{coro3}}. \cord{Since the left-hand side of  \eqref{prhhh} is nonnegative, we thus deduce}
\begin{multline} \cord{ \left|  \sum_{i=1}^N \int_{B(y_i, 2\eta_i)} |\nab (\hat h- h_{\veta}^\nu[Y_N]
)|^2 -\int_{B(y_i, 2\eta_i)} |\nab  v_i|^2 \right|}\\
\le C t^2
\cor{ \(|\psi|_{C^1}^2   + |\psi|_{C^1} |\psi|_{C^2} N^{-\frac1\d}\log (\ell N^{\frac1\d})  +|\psi|_{C^2}^2 N^{-\frac2\d} \log (\ell N^{\frac1\d} ) \)} \cord{ \Xi(0)}  \\
\cord{+ C t^2 ( |\psi|_{C^1}^2+ |\psi|_{C^2}^2 N^{-\frac2\d})      m^{2-\d} \Xi(0) \indic_{\d \ge 3}}
\end{multline}
and combining with  \eqref{prelimb}, we deduce that  \cord{
\begin{multline} \label{prh3}
\int_{\R^\d\backslash \cup_i B(y_i, 2\eta_i)} |\nab (\hat h-h_{\veta}^\nu)[Y_N]|^2\\
 \le Ct^2
 \(|\psi|_{C^1}^2   + |\psi|_{C^1} |\psi|_{C^2} N^{-\frac1\d}\log (\ell N^{\frac1\d})  +|\psi|_{C^2}^2 N^{-\frac2\d} \log (\ell N^{\frac1\d} ) \)
 \Xi(0) \\+ C t^2 ( |\psi|_{C^1}^2+ |\psi|_{C^2}^2 N^{-\frac2\d})      m^{2-\d}  \Xi(0)  \indic_{\d \ge 3} .\end{multline}}

To prove \eqref{conclst12}, in view of \eqref{a62} \cord{and $|\mathcal A|\le |\psi|_{C^1}$} it suffices to show that \cord{
\begin{multline} \label{A38} \int |\nab (h_{\veta}^\nu [Y_N]- \hat h)| | \mathcal A|| \nab h_{\veta}^\nu [Y_N]|\\
\le   C |t| \( |\psi|_{C^1}^2 + |\psi|_{C^1} |\psi|_{C^2} N^{-\frac1\d}\log (\ell N^{\frac1\d})\) \\
\times  \(\Xi(0) (1+ m^{2-\d}\indic_{\d\ge3})    + C \#I_N \g(m)\indic_{\d=2}\)
  +O(t^2).\end{multline}}
We break the integral into $\cup_i B(\cord{y_i}, \eta_i)$ and the complement. In the complement,  the bound comes from \eqref{prh3} and Cauchy-Schwarz, using that $|\mathcal A|\le |\psi|_{C^1}$ and  \cord{\eqref{conthe}, and we obtain
\begin{multline} \label{exteri}
\int_{\R^\d  \backslash \cup_i B(y_i, 2 \eta_i)} |\nab (h_{\veta}^\nu [Y_N]- \hat h)| | \mathcal A|| \nab h_{\veta}^\nu [Y_N]|\\
\le C|t|  |\psi|_{C^1} \(|\psi|_{C^1}^2   + |\psi|_{C^1} |\psi|_{C^2} N^{-\frac1\d}\log (\ell N^{\frac1\d})  +|\psi|_{C^2}^2 N^{-\frac2\d} \log (\ell N^{\frac1\d} ) \)^{\hal} \\ \times \( \Xi(0) (1+ m^{2-\d}\indic_{\d\ge3}) + C \#I_N \g(m)\indic_{\d=2} \)    \\
 \le C |t| \( |\psi|_{C^1}^2 + |\psi|_{C^1} |\psi|_{C^2} N^{-\frac1\d}\log (\ell N^{\frac1\d})\)  \(\Xi(0) (1+ m^{2-\d} \indic_{\d\ge 3})   + C \#I_N \g(m)\indic_{\d=2} \)   
 .\end{multline}}
There remains to study the contribution in $\cup_i B(y_i, 2\eta_i)$.
We may bound it by
\be\label{3pieces}
\cord{\sum_{i\in I_N}} \int_{B(y_i, 2\eta_i)} |\nab v_i||\mathcal A| |\nab \g* \delta_{y_i}^{(\eta_i)}|+
\int_{B(y_i, 2\eta_i)} |\nab v_i||\mathcal A||\nab \tilde h_i[Y_N]| + \int_{B(y_i, 2\eta_i)} |\nab (\cord{\hat h}-h_{\veta}^\nu-v_i)||\mathcal A| |\nab h_{\veta}^\nu|.\ee Here we noted that $h_{\veta}^\nu[Y_N]-\g*\delta_{y_i} $ coincides with $\tilde h_i[Y_N]$  in $B(x_i, 2\eta_i)$ because $2\eta_i \le \rr_i$.

We bound the first piece by computing explicitly.
We recall that $\nab \g* \delta_{y_i}^{(\eta_i)}=( \nab \g (\cdot - y_i) ) \indic_{|x-y_i|\ge \eta_i}$.
We compute that for $x\notin \pa B(y_i, \eta_i)$, using $y_i-x_i=t\psi(x_i)$,
\begin{multline*}
\nab v_i(x) =\dashint_{\pa \cord{ B(x_i, \eta_i)}}  \nab \g(x-y-t\psi(y))- \nab \g(x-y+x_i-y_i)\\ =
t\dashint_{\pa \cord{B(x_i, \eta_i)}} D^2 \g(x-y) \cdot (\psi(x_i)- \psi(x)+\psi(x)-\psi(y)) +O_N(t^2).\end{multline*}
We break the right-hand side into 
$$t D^2 \g * \delta_{y_i}^{(\eta_i)} (\psi(x_i)-\psi(x)) + O \( t |\psi|_{C^1} \dashint_{\pa \cord{ B(x_i, \eta_i)}} |x-y|^{1-\d} \)+ O(t^2).$$
The convolution with the singular measure $\delta_{y_i}^{(\eta_i)}$ can be justified     by  smoothing out $\delta_{y_i}^{(\eta_i)}.$
We deduce that 
\begin{multline}\label{bnabvi}
\int_{B(y_i, 2\eta_i)} |\nab v_i| 
\\
\le C |t| |\psi|_{C^1}  \int_{B(x_i, 2\eta_i)} |D^2\g_{\eta_i }(x-x_i)|  \eta_i + 
C|t| |\psi|_{C^1} \int_{\cord{B(y_i,2 \eta_i)} }\dashint_{\pa \cord{ B(x_i , \eta_i)}} \frac{dx dy}{|x-y|^{\d-1}}+O(t^2)\\
\le  C |t| |\psi|_{C^1}   \eta_i  +O(t^2) .\end{multline}
For the last relation we have used that $|D^2\g_{\eta_i}|$ is a measure with a singular part of  mass $\le C$ on $\pa B(y_i, \eta_i)$ and a diffuse part of density $\le \eta_i^{-\d}$  in $B(y_i, 2\eta_i)\backslash B(y_i , \eta_i)$.

Thus we conclude that 
\begin{equation*}\sum_{i=1}^N  \int_{B(y_i, 2\eta_i)} |\nab v_i||\mathcal A| |\nab \g* \delta_{y_i}^{(\eta_i)}|\le C |t| |\psi|_{C^1}^2 \sum_i \eta_i^{2-\d} \le C| t| |\psi|_{C^1}^2 \cord{ \Xi(0) (1  + m^{2-\d}  \indic_{\d \ge 3}) }   .\end{equation*}

For the second piece in \eqref{3pieces} we bound $|\nab \tilde h_i[Y_N]|$ via \eqref{contrhi}, noticing that $B(y_i, 2\eta_i) \subset B(x_i, \hal \rr_i)$. Combining with \eqref{bnabvi}, we find 
$$\int_{B(y_i, 2\eta_i)} |\nab v_i||\mathcal A||\nab \tilde h_i[Y_N]| \le C |t| |\psi|_{C^1}^2\eta_i \( \rr_i^{-\frac\d2}  \lambda_i^{\hal}+N\rr_i \|\mu\|_{L^\infty}   \) 
+O(t^2).$$ 
Summing, \cord{using $\eta_i\le \rr_i$ and Young's inequality,} we obtain 
\begin{multline*}
 \sum_{i=1}^N \int_{B(y_i, 2\eta_i)} |\nab v_i||\mathcal A||\nab \tilde h_i[Y_N]| \\
 \le C |t| |\psi|_{C^1}^2 \( \sum_{i=1}^N  \rr_i^{2-\d}  +\lambda_i + N \sum_{i=1}^N \rr_i^2\) +O(t^2) \le C |t||\psi|_{C^1}^2  \cord{ \Xi(0)}+O(t^2),
\end{multline*}
using again Lemma \ref{lem:contrdist1}. 
This concludes the evaluation of \eqref{3pieces} and combining with \eqref{exteri} we have finished the proof of \eqref{A38} hence of \eqref{conclst12}.

Finally, since $E_\eta = \Phi_t\#   \nab h_{\veta}^\mu[\XN]$ we check with a change of variables and direct calculations that \cord{
\begin{multline*}
\left|\int_{\R^\d} E_{\veta} \cdot \mathcal A E_{\veta} - \int_{\R^\d} \nab h_{\veta}^\mu[X_N], \mathcal A  \nab h_{\veta}^\mu[X_N]\right|\\ \le C| t| \int_{\R^\d} |\nab h_{\veta}^\mu [\XN]|^2 |D\psi|^2 +| t|\int_{\R^\d}  |D^2 \psi||\psi| |\nab h_{\veta}^\mu[X_N] |^2 +  O (t^2)  \\ \le C| t| \(  |\psi|^2_{C^1} +|\psi|_{C^2}\|\psi\|_{L^\infty} \) \(  \Xi(0) (1+m^{2-\d}\indic_{\d\ge 3}) + C \#I_N \g(m)\indic_{\d=2}   \)  + O(t^2),\end{multline*}
where we used \eqref{conthe}.}

Combining with \eqref{conclst12} and \eqref{autili} we then bound \cord{ 
\begin{multline}
 \left|\int_{\R^\d} \nab h_{\veta}^\nu [Y_N] \cdot \mathcal A  \nab h_{\veta}^\nu [Y_N]
- \int_{\R^\d} \nab h_{\veta}^\mu[\XN]) \cdot \mathcal A \nab h_{\veta}^\mu[\XN]\right| \\
\le  |t| ( |\psi|_{C^1}^2 + |\psi|_{C^2}\|\psi\|_{L^\infty})    \( \Xi(0) (1+m^{2-\d}\indic_{\d\ge 3})  +  \# I_N \g(m)\indic_{\d=2} \) \\+
C |t|  |\psi|_{C^1} |\psi|_{C^2 } N^{-\frac1\d} \log (\ell N^{\frac1\d})  \Xi(0)
+\frac{1}{4\cd |t|} \int_{\R^\d} |\nab \hat h-E_{\veta}|^2 +O(t^2),\end{multline} where $C$ depends only on $\|\mu\|_{L^\infty}$ and $O(t^2)$ may depend on $\XN, \vec{\eta}, \mu, \d$.}
In view of \eqref{mainenergyterm}  \cord{and \eqref{conthe}, we deduce that  
\begin{multline}\label{derivmain}
|\mathrm{Main}'(t)-\mathrm{Main}'(0)|\\
\le C |t| ( |\psi|_{C^1}^2 + |\psi|_{C^2}\|\psi\|_{L^\infty})    \( \Xi(0) (1+m^{2-\d}\indic_{\d\ge 3}) +  \# I_N \g(m)\indic_{\d=2} \) \\+
C |t|  |\psi|_{C^1} |\psi|_{C^2 } N^{-\frac1\d} \log (\ell N^{\frac1\d})  \Xi(0)
+\frac{1}{4\cd |t|} \int_{\R^\d} |\nab \hat h-E_{\veta}|^2 +O(t^2).\end{multline}}

\noindent
\cord{ {\bf Step 3.} Next, let us evaluate the first derivative in $t$ of  the second term in $L_t$ of \eqref{defL}.
Making this term more explicit as 
\begin{multline*}T_t:=\sum_{i=1}^N \dashint_{\partial B(x_i, \eta_i)}   \nab \g(x+y_i-x_i-y) \Phi_t\#\( \sum_{j \neq i} \delta_{x_j}-N\mu\) (y) \cdot \( \psi(x)-\psi(x_i)\) \\
= \sum_{i\in I_{U_\ell}} \dashint_{\partial B(x_i, \eta_i)} \int_{\R^\d} \nab \g(x-y+t \psi(x_i) - t\psi(y) ) \( \sum_{j \neq i} \delta_{x_j}-N\mu\) (y) \cdot \( \psi(x)-\psi(x_i)\) ,\end{multline*} 
we may compute its first derivative at $t=0$ as 
\begin{align} \label{derivT}  
T_0' & = \sum_{i\in I_{U_\ell}} \dashint_{\partial B(x_i,\eta_i)}\int_{\R^\d}  D^2 \g(x-y) d\( \sum_{j \neq i} \delta_{x_j}-N\mu\) (y)  (\psi(x_i)-\psi(y))  \cdot \( \psi(x)-\psi(x_i)\)\\  \label{derivT1} &=
 \sum_{i\in I_{U_\ell}} \dashint_{\partial B(x_i,\eta_i)}    D^2 \tilde h_i (x)  \psi(x_i) \cdot  (\psi(x)-\psi(x_i)) \\ \label{derivT2}
 & -\sum_{i\in I_{U_{\ell}}}\dashint_{\partial B(x_i,\eta_i)}\int_{\R^\d}  D^2 \g(x-y) d\( \sum_{j \neq i} \delta_{x_j}-N\mu(y) (1-\chi_i) (y) \) \psi(y)  \cdot \( \psi(x)-\psi(x_i)\)  \\ 
 \label{derivT3}
 &     +\sum_{i\in I_{U_{\ell}}   }\dashint_{\partial B(x_i,\eta_i)}\int_{\R^\d}  D^2 \g(x-y) N\mu(y) \chi_i  (y)  \psi(y)  \cdot \( \psi(x)-\psi(x_i)\) 
 \end{align} where $\chi_i$ is a cutoff function equal to $1 $ in $B(x_i, N^{-1/\d})$, vanishing outside $B(x_i, 2 N^{-1/\d})$ and such that $|\nab \chi_i|\le N^{1/\d}$. To bound these terms, we want to insert that $\psi(x)-\psi(x_i)= D\psi(x_i)(x-x_i) + O(|\psi|_{C^2} \eta_i^2)$  and use that $D\psi(x_i) (x-x_i)$ integrates to $0$ over $\partial B(x_i, \eta_i)$.  Inserting that decomposition, the first term in \eqref{derivT1} is then  equal to 
 $$\sum_{i\in U_\ell} \dashint_{\partial B(x_i, \eta_i)}  D\psi(x_i)^T D^2 \tilde h_i (x)  \psi(x_i) \cdot  (x-x_i) + O \( \sum_{i \in U_\ell} \|D^2\tilde h_i\|_{L^\infty(B(x_i, \eta_i))} \|\psi\|_{L^\infty} |\psi|_{C^2} \eta_i^2 \)
 .$$  After expanding  $D^2 \tilde h_i (x)= D^2 \tilde h_i(x_i) + O(|D^{3}\tilde h_i||x-x_i|)$  and using the cancellation, we find that the term
 in \eqref{derivT1} is   bounded by 
 $$ |\eqref{derivT1}|\le C \sum_{i\in U_\ell} \|D^2\tilde h_i\|_{L^\infty(B(x_i, \hal \rr_i) )} \|\psi\|_{L^\infty} |\psi|_{C^2}\eta_i^2
 + \|D^{3}\tilde h_i\|_{L^\infty(B(x_i, \hal \rr_i) )} \|\psi\|_{L^\infty} |\psi|_{C^1}   \eta_i^{2}.
 $$ Thanks to \eqref{contrhi2}, we may then bound this 
by 
\begin{multline*}|\eqref{derivT1}|\le C\sum_{i\in I_N}   \eta_i^2 \(\rr_i^{-1-\frac{\d}{2} } \lambda_i(\XN,\mu)^\hal +  N (\|\mu\|_{L^\infty} + \rr_i^\sigma |\mu|_{C^{\sigma} (B(x_i, \rr_i))} ) \) \|\psi\|_{L^\infty}  |\psi|_{C^2}   
\\+ C\sum_{i\in I_N}  \eta_i^{2}   \( \rr_i^{-2-\frac{\d}{2}}  \lambda_i(\XN,\mu)^\hal  + N( \rr_i^{-1} \|\mu\|_{L^\infty} +\rr_i^{\sigma} |\mu|_{C^{1+\sigma} (B(x_i, \rr_i))} ) \)  \|\psi\|_{L^\infty} |\psi|_{C^1},
  \end{multline*}
  for any $\sigma>0$.
  Next, we insert the choice \eqref{choiceetai} and find, using $\eta_i \le m \rr_i \le m N^{-1/\d}$ and Young's inequality  
  \begin{align*}
  |\eqref{derivT1}| & \le C\( \sum_{i\in I_N}   \rr_i^{2-\d}
   + \sum_{i\in I_N} \lambda_i (\XN, \mu) \)  \|\psi\|_{L^\infty}  |\psi|_{C^2}  \\ 
   & +
   m^2  \# I_N N^{1-\frac2\d}    \(\|\mu\|_{L^\infty} +N^{-\frac\sigma\d}    |\mu|_{C^{\sigma} (U_\ell)}  \) \|\psi\|_{L^\infty}  |\psi|_{C^2}   
\\ &+ C   m^2 N^{\frac{2\alpha}{\d}}  \( \sum_{i\in I_N} \rr_i ^{p(2\alpha-\frac{\d}{2}) } + \lambda_i(X_N, \mu)^{\frac{q}{2}}\) \|\psi\|_{L^\infty} |\psi|_{C^1}\\ & 
 + \( m^2 N \# I_N\( N^{-\frac1\d}   \|\mu\|_{L^\infty} + N^{-\frac2\d-\frac\sigma\d} |\mu|_{C^{1+\sigma }(U_\ell) }  \) \) \|\psi\|_{L^\infty} |\psi|_{C^1},
\end{align*} for $p\le 2$ and $q\ge 2$ conjuguate Sobolev exponents. In dimension 2, we choose $p=q=2$ and by choice $\alpha= \hal $  (see \eqref{defalpha})  the corresponding sums are all bounded by $\Xi(0)$. In dimension $\d\ge 3$, we use the elementary inequality $\sum_i u_i^r \le \(\sum_i u_i\)^r$  for all $u_i\ge 0$ and  $r \ge 1$ to  bound 
 $\sum_{i\in I_N} \rr_i ^{-p\frac{\d}{2} } $ by   $ \( \sum_{i\in I_N} \rr_i^{2-\d}\)^{\frac{p \frac\d2}{\d-2}}$  and $ \sum_{i\in I_N}\lambda_i(X_N, \mu)^{\frac{q}{2}}
$ by $\(\sum_{i\in I_N} \lambda_i(X_N, \mu)\)^{\frac{q}{2}}\le C \Xi(0)^{\frac{q}{2}}$. Optimizing gives the choice $q= 1+ \frac{\d}{\d-2}$ and so the sum in the third line is then bounded by $C  \Xi(0)^{  \frac{\d-1}{\d-2} }$.  In view of the definition of $\Xi$, we
 finally obtain
\begin{multline}\label{revderivT1}
 |\eqref{derivT1}|\le C   \( \|\psi\|_{L^\infty} |\psi|_{C^2} +m^2 N^{\frac1\d} \|\psi\|_{L^\infty} |\psi|_{C^1}\)  \Xi(0)+  Cm^2    \|\psi\|_{L^\infty}|\psi|_{C^1}   \Xi(0)^{  \frac{\d-1}{\d-2} }  \indic_{\d\ge 3}\\
 + 
 m^2  \#I_N N^{1-\frac2\d} \( N^{-\frac\sigma\d}  \|\psi\|_{L^\infty} |\psi|_{C^2} |\mu|_{C^\sigma(U_\ell)} + N^{-\frac{\sigma}\d}  \|\psi\|_{L^\infty}|\psi|_{C^1} |\mu|_{C^{1+\sigma}(U_\ell)}\) .
 \end{multline}
  For the term  in \eqref{derivT2}, we argue similarly to bound it by 
\begin{multline*} |\eqref{derivT2}|\le C  \sum_{x_i, x_j\in U_\ell, i\neq j}   \frac{ \eta_i^2 \|\psi\|_{L^\infty}|\psi|_{C^2}  }{|x_i-x_j|^{\d}} + C  \sum_{x_i, x_j\in U_\ell, i\neq j}   \frac{ \eta_i^2 \|\psi\|_{L^\infty}|\psi|_{C^1}  }{|x_i-x_j|^{\d+1}} 
\\+C N \sum_{i\in I_N}   \eta_i^2 \|\psi\|_{L^\infty}|\psi|_{C^2} \int_{y\in U_\ell, |y-x_i|\ge N^{-1/\d} }  \frac{d\mu(y)}{|x_i-y|^\d}
\\+
C N \sum_{i\in I_N}   \eta_i^2 \|\psi\|_{L^\infty}|\psi|_{C^1} \int_{y\in U_\ell, |y-x_i|\ge N^{-1/\d} }  \frac{d\mu(y)}{|x_i-y|^{\d+1}}.\end{multline*} The first sum  may be controlled similarly as we did for \eqref{rem2} via Proposition~\ref{multiscale} combined with Corollary \ref{coro3}, using that  $\eta_i \le  m N^{-1/\d}$, and they are bounded by $ Cm^2 ( \log (\ell N^{1/\d})   \|\psi\|_{L^\infty}|\psi|_{C^2}+N^{1/\d}   \|\psi\|_{L^\infty}|\psi|_{C^1})   \Xi(0)$. 
Writing  that with  \eqref{choiceetai} $$\frac{\eta_i^2 }{|x_i-x_j|^{\d+1}}\le m^2N^{\frac{2\alpha}{\d}}  \frac{\rr_i^{2+2\alpha} }{ |x_i-x_j|^{\d+1}}\le m^2N^{\frac{2\alpha}{\d}}   \frac{1}{|x_i-x_j|^{\d-1-2\alpha} }$$ and for the use of  Corollary \ref{coro3} using that  $\sum_{i\neq j} |x_i-x_j|^{-\d+1+2\alpha}\le \(\sum_{i\neq j} |x_i-x_j|^{2-\d} \)^{\frac{ -\d+1+ 2\alpha}{2-\d}} $, we may then bound the  second  sum in the same way.
Finally the last two integrals are bounded by similar terms as in \eqref{derivT1}.
We next examine \eqref{derivT3}. Integrating once by parts in $y$ we may bound this term by 
\begin{equation}\label{rederivT3}
|\eqref{derivT3}|\le C N^{1-\frac1\d} \sum_{i\in I_N} \eta_i |\psi|_{C^1} ( |\mu|_{C^1(U_\ell)}\|\psi\|_{L^\infty}+ \|\mu\|_{L^\infty} |\psi|_{C^1}+\|\mu\|_{L^\infty} \|\psi\|_{L^\infty} N^{\frac1\d}) \end{equation}
Inserting all these results into \eqref{derivT}, we conclude that 
\begin{multline} \label{derivrem}|T_0'|\le  C \Big( |\psi|_{C^1}^2 + m^2  \|\psi\|_{L^\infty}|\psi|_{C^2} \(  \log (\ell N^{\frac1\d})    +   N^{-\frac\sigma\d}|\mu|_{C^\sigma}\) \\ +  \|\psi\|_{L^\infty}|\psi|_{C^1}
\( m^2 N^{\frac{1}{\d }}+ m^2 N^{-\frac{\sigma}{\d}}  |\mu|_{C^{1+\sigma}}+ m   |\mu|_{C^1}+ m \|\mu\|_{L^\infty}  N^{\frac1\d})  \)
\Big)   \Xi(0)  \\+  Cm^2  \|\psi\|_{L^\infty}|\psi|_{C^1}   \Xi(0)^{  \frac{\d-1}{\d-2} }
\indic_{\d\ge 3}
 ,\end{multline} where $C$ depends only on $\|\mu\|_{L^\infty}$ and $\d$.
}

\noindent
\cord{{\bf Step 4.}  We examine the first derivative in $t$ of the third term in $L_t$ of \eqref{defL}. Rewriting this term as $\hal \sum_i \eta_i^{1-\d}\dashint_{\partial B(x_i, \eta_i)} (\psi(x+y_i-x_i)-\psi(y_i)) $, we find that derivative is 
\begin{equation*}\hal \sum_i \dashint_{\partial B(x_i, \eta_i)} \eta_i^{1-\d} (D\psi(x)-D\psi(x_i)) \cdot \psi(x_i)\le C |\psi|_{C^2}\|\psi\|_{L^\infty}  \Xi(0)(1+ Cm^{2-\d}\indic_{\d\ge 3}),\end{equation*}
where we used \eqref{newboundg3}.}

\noindent
\cord{{\bf Step 5.} (Conclusion) The sum of all the controls we have obtained   for $|L_t-L_0|$ and  in the order $t^2$ term in \eqref{prtrans} is 
\begin{align*}
 &   \Bigg(|\psi|_{C^1}^2  \(1+N^{-\frac1\d} |\mu|_{C^1} + N^{-\frac2\d}|\mu|_{C^2}\)   + |\psi|_{C^2}\|\psi\|_{L^\infty}   \\ &   \qquad +   |\psi|_{C^1} |\psi|_{C^2 }   N^{-\frac1\d} \log (\ell N^{\frac1\d})   (1+ N^{-\frac2\d} |\mu|_{C^2} )  \Bigg)    \( \Xi(0 ) (1+ m^{2-\d} \indic_{\d\ge3})  +  \# I_N \g(m)\indic_{\d=2} \) \\ & +
\Bigg(     \|\psi\|_{L^\infty}|\psi|_{C^2} m^2    \(  \log (\ell N^{\frac1\d})  +N^{-\frac\sigma\d} |\mu|_{C^\sigma} \) \\ &
\qquad +  \|\psi\|_{L^\infty}|\psi|_{C^1}\(  m^2 N^{\frac{1}{\d }}  + m^2 N^{-\frac{\sigma}{\d}}  |\mu|_{C^{1+\sigma}}+ m   |\mu|_{C^1}   + m \|\mu\|_{L^\infty}  N^{\frac1\d})    \) \Bigg) \Xi(0)\\
&+  m^2 \|\psi\|_{L^\infty}|\psi|_{C^1}   \Xi(0)^{  \frac{\d-1}{\d-2}}\indic_{\d\ge3}
.\end{align*}
 Inserting the choice \eqref{defm} we may  bound it by a constant times
\begin{align*}
 A_{N,\psi} :=  &
  \Bigg[|\psi|_{C^1}^2  \(1+N^{-\frac1\d} |\mu|_{C^1} + N^{-\frac2\d}|\mu|_{C^2}\)   + |\psi|_{C^2}\|\psi\|_{L^\infty} (1+   N^{-  \frac\sigma\d} |\mu|_{C^\sigma})   \\ &+   |\psi|_{C^1} |\psi|_{C^2 }   N^{-\frac1\d} \log (\ell N^{\frac1\d})   (1+ N^{-\frac2\d} |\mu|_{C^2} )  \Bigg]   \(1+ (N^{\frac1\d}\ell)^{\alpha' (\d-2)} \indic_{\d\ge3} + \log (\ell N^{\frac1\d}) \indic_{\d=2} \)  \Xi(0 ) \\  &+
     \ell^{-1} \|\psi\|_{L^\infty}|\psi|_{C^1}\Bigg[     (N^{\frac1\d}\ell)^{-1}  \(1+ N^{-\frac{1+\sigma}\d} |\mu|_{C^{1+\sigma}} +N^{-\frac1\d}|\mu|_{C^1} \) \Xi(0) \indic_{\d=2}  \\
    &+  \Bigg(\( (N^{\frac1\d}\ell)^{1-2\alpha'}(1 +  N^{-\frac{1+\sigma}\d} |\mu|_{C^{1+\sigma}} ) +  (N^{\frac1\d}\ell)^{1-\alpha'} N^{-\frac1\d}|\mu|_{C^1}\)
      \Xi(0)\\ &\qquad+ (N^{\frac1\d} \ell)^{1-2\alpha'} N^{-\frac1\d}   \Xi(0)^{  \frac{\d-1}{\d-2}}  \Bigg) \indic_{\d\ge 3}\Bigg]
\end{align*}
}

\cord{By \eqref{55} we have
 \be \label{vp1} \Xi(t)-\Xi(0)= \Main(t)+\Rem(t)+\mathrm{Err}(t) - \frac{1}{2\cd}\int_{\R^\d} |\nab \hat h-E_{\veta}|^2,\ee and we thus have obtained   in view of \eqref{prtrans}, \eqref{derivmain}, \eqref{derivrem}, combined with \eqref{conthe},  that
 $$ |\Xi(t)-\Xi(0)-t L_0| \le  C t^2 \cord{A_{N,\psi}}+O(t^3) $$
 yielding $$|\Xi''(0)|\le C\cord{A_{N,\psi}} .$$
 We thus  obtain 
  the desired result at $t=0$. }
The same reasoning applied near $t$ together with the fact that with \eqref{tpsi},
$$|\psi\circ \Phi_t^{-1}|_{C^2} \le C |\psi|_{C^2} $$
allows to conclude with \eqref{diffpath} that  \eqref{p41} and \eqref{p42} hold.




\end{document}